%% file: paper.tex
\documentclass[sigplan,10pt]{acmart}

\acmConference[PLDI'19]{ACM SIGPLAN Conference on Programming Languages}{June 22--28, 2019}{Phoenix, AZ, USA}
\acmYear{2019}
\acmISBN{} % \acmISBN{978-x-xxxx-xxxx-x/YY/MM}
\acmDOI{} % \acmDOI{10.1145/nnnnnnn.nnnnnnn}
\startPage{1}

\setcopyright{none}
\input{std}

\begin{document}
%% Title information
\title[Modular D\&C Parallelization]{Modular Synthesis of Divide-and-Conquer\\ Parallelism for Nested Loops \\ (Extended Version)}

%{Modular Divide-and-Conquer Parallelization \\ of Nested Loops}

% \titlenote{with title note}             %% \titlenote is optional;
                                        %% can be repeated if necessary;
                                        %% contents suppressed with 'anonymous'
% \subtitle{Subtitle}                     %% \subtitle is optional
% \subtitlenote{with subtitle note}       %% \subtitlenote is optional;
%                                         %% can be repeated if necessary;
%                                         %% contents suppressed with 'anonymous'

%% Author information
%% Contents and number of authors suppressed with 'anonymous'.
%% Each author should be introduced by \author, followed by
%% \authornote (optional), \orcid (optional), \affiliation, and
%% \email.
%% An author may have multiple affiliations and/or emails; repeat the
%% appropriate command.
%% Many elements are not rendered, but should be provided for metadata
%% extraction tools.

\author{Azadeh Farzan}

\affiliation{
  \department{Department of Computer Science}              %% \department is recommended
  \institution{University of Toronto}            %% \institution is required
  \city{Toronto}
  \country{Canada}                    %% \country is recommended
}
\email{azadeh@cs.toronto.edu}          %% \email is recommended

\author{Victor Nicolet}

\affiliation{
  \department{Department of Computer Science}              %% \department is recommended
  \institution{University of Toronto}            %% \institution is required
  \city{Toronto}
  \country{Canada}                    %% \country is recommended
}
\email{victorn@cs.toronto.edu}          %% \email is recommended

\thanks{This is the extended version of PLDI 2019 paper by the same authors which includes the proofs of theorems and additional details.}

\begin{abstract}
We propose a methodology for automatic generation of divide-and-conquer parallel implementations of sequential nested loops. We focus on a class of loops that traverse read-only multidimensional collections (lists or arrays) and compute a function over these collections. Our approach is {\em modular}, in that, the inner loop nest is abstracted away to produce a simpler loop nest for parallelization. Then, the {\em summarized} version of the loop nest is parallelized. The main challenge addressed by this paper is that to perform the code transformations necessary in each step, the loop nest may have to be augmented (automatically) with extra computation to make possible the abstraction and/or the parallelization tasks. We present theoretical results to justify the correctness of our modular approach, and  algorithmic solutions for automation. Experimental results  demonstrate that our approach can parallelize highly non-trivial loop nests efficiently.
\end{abstract}

%% Keywords
%% comma separated list
\keywords{Divide and Conquer Parallelism, Program Synthesis, Homomorphisms}

%% \maketitle
%% Note: \maketitle command must come after title commands, author
%% commands, abstract environment, Computing Classification System
%% environment and commands, and keywords command.
\maketitle

%% Acknowledgments
% \begin{acks}

% \end{acks}

\input{introduction}
\input{overview}
\input{preliminaries}
\input{function}
\input{lift}
\input{parallelizability}
\input{join}
\input{algorithm}
\input{extended_example}
\input{experiments}
\input{conclusion}
\input{relwork}

\clearpage

\citestyle{acmnumeric}
\bibliography{biblio}

\end{document}

%% file: std.tex
\usepackage{stmaryrd,wrapfig}
\usepackage[noend,vlined,ruled]{algorithm2e}
\usepackage[prologue]{xcolor}
\usepackage{listings}
\usepackage{tikz}
\usepackage{bold-extra}
\usepackage{multirow}
\usepackage{booktabs}
\usepackage{picins,pifont}
\usepackage[labelformat=simple]{subcaption}

\DeclareMathAlphabet{\mathbbm}{U}{bbm}{m}{n}% from bbm.sty

\newcommand{\Loc}{\textit{Loc}}
\newcommand{\State}{\textit{State}} % States
\newcommand{\Var}{\textsf{Var}}     % Variables

\newcommand{\cf}{\mathcal{C}}
\newcommand{\cfrec}{\mathcal{C}_{rec}}
\newcommand{\lrhole}{??_{LR}}
\newcommand{\lrrechole}{??_{Rec}}
\newcommand{\rhole}{??_R}
\newcommand{\loopfor}[2]{\mathtt{for}\:(#1)\;\left\{\; #2 \;\right\}}
\newcommand{\rrules}{\mathcal{R}}
\newcommand\tern[3]{{#1 \: ? \: #2 \: : \: #3}}
\renewcommand{\phi}{\varphi}

\makeatletter
\newcommand{\killpic}{%
  \hangindent=0pt
  \let\par=\old@par
}
\makeatother
\makeatletter
\newcommand\incircbin
{%
  \mathpalette\@incircbin
}
\newcommand\@incircbin[2]
{%
  \mathbin%
  {%
    \ooalign{\hidewidth$#1#2$\hidewidth\crcr$#1\bigcirc$}%
  }%
}
\newcommand{\oland}{\incircbin{\land}}
\makeatother

\makeatletter
\newcommand*{\pmzerodot}{%
  \nfss@text{%
    \sbox0{$\vcenter{}$}% math axis
    \sbox2{0}%
    \sbox4{0\/}%
    \ooalign{%
      0\cr
      \hidewidth
      \kern\dimexpr\wd4-\wd2\relax % compensate for slanted fonts
      \raise\dimexpr(\ht2-\dp2)/2-\ht0\relax\hbox{%
        \if b\expandafter\@car\f@series\@nil\relax
          \mathversion{bold}%
        \fi
        $\cdot\m@th$%
      }%
      \hidewidth
      \cr
      \vphantom{0}% correct depth of final symbol
    }%
  }%
}
\newcommand*{\pmzeroslash}{%
  \nfss@text{%
    \sbox0{0}%
    \sbox2{/}%
    \sbox4{%
      \raise\dimexpr((\ht0-\dp0)-(\ht2-\dp2))/2\relax\copy2 %
    }%
    \ooalign{%
      \hfill\copy4 \hfill\cr
      \hfill0\hfill\cr
    }%
    \vphantom{0\copy4 }% correct overall height and depth of the symbol
  }%
}
\makeatother

\def\tool{{\sc ParSynt }}
\def\btool{{\sc \bf ParSynt }}

\def\seq{\mathcal{S}}
\def\ccat{\bullet}

\def\map{\mathit{map}}

\def\foldl{\mathit{foldl}}
\def\scalars{\mathit{Sc}}

\def\min2{\mathit{min2}}
\def\max{\mathit{max}}
\def\min{\mathit{min}}
\def\int{\mathit{int}}

\def\zero{\pmzeroslash}
\newcommand{\lift}[2]{\hat{#1}^{{#2}}}

\newcommand\ite[3]{{\sf if}\ #1\ {\sf then}\ #2\ {\sf else}\ #3}

\newcommand{\newkw}[2]{\def#1{{\rm\sf#2}}}

\newkw{\Loc}{Loc} % Program locations
\newkw{\DF}{DF}            % Data flow edges
\newkw{\AP}{AP}   % Access paths
\newkw{\State}{State}   % States
\newkw{\GV}{GV}   % Global variables
\newkw{\LV}{LV}   % Local variables
\newkw{\Var}{Var}   % Variables
\newkw{\Seq}{SeVar}   % Sequence Variables
\newkw{\BSeq}{BSVar}   % Sequence Variables
\newkw{\BVar}{BVar}   % Boolean Variables
\newkw{\Exp}{Exp}   % Expressions
\newkw{\LhVar}{LhVar} % Subscripted vars
\newkw{\Act}{Act}   % Actions
\newkw{\Heap}{Heap}
\newkw{\Contour}{Contour}
\newkw{\Record}{Record}
\newkw{\Env}{Env}
\newkw{\act}{act}
\newkw{\Trace}{Trace}
\newkw{\Formula}{Formula}
\newkw{\Term}{Term}
\newkw{\Live}{Live}
\newkw{\true}{true}
\newkw{\false}{false}
\newkw{\Pr}{Program}
\newkw{\iter}{Iterator}
\newkw{\Int}{Int}
\newkw{\IF}{if}
\newkw{\When}{when}

\newkw{\SVar}{SVar}   % State Variables
\newkw{\IVar}{IVar}   % Input Variables
\newkw{\Aux}{Aux}     % Auxiliary variables
\newkw{\exp}{exp}   % Input Variables

\newkw{\cost}{Cost} % Cost function

\def\numbinary{\oplus}

\def\compare{\olessthan}

\def\boolbinary{\oland}

\newkw{\oAux}{OldAux}
\newkw{\ufold}{unfold}
\newkw{\norm}{\textit{{normalize}}}
\newkw{\collect}{collect}

\SetKwProg{Fn}{Function}{}{}

\definecolor{inputcolor}{HTML}{0066FF}
\definecolor{statecolor}{HTML}{b43030}
\definecolor{normalformcolor}{HTML}{304FFE}
\definecolor{auxcolor}{HTML}{004D40}
\definecolor{state}{HTML}{d50000}

%% file: introduction.tex
% !TEX root =  paper.tex

\section{Introduction} \label{sec:introduction}

The advent of multicore computers and development of APIs like OpenMP \cite{OPENMP}, CUDA \cite{CUDA}, and TBB \cite{TBB} has increased the popularity of parallel programming for performance gains. Despite big advances in parallelizing compilers,  {\em correct and efficient} parallel code is often hand-crafted through a time-consuming and error-prone process. These APIs implement commonly used {\em parallel programming skeletons} that ease the task of parallel programming. Instead of writing a parallel program from scratch, a programmer needs to only specify the key components of a particular skeleton. Divide-and-conquer parallelism is the most commonly used of such skeletons for which the programmer has to specify a {\em split}, a {\em work}, and a {\em  join} function.  We propose a methodology to automatically generate these components.

We focus on a class of divide-and-conquer parallel programs that operate on {\em multidimensional sequences} (e.g. multidimensional arrays, or in general  any collection type with similar recursive structure) in which the divide ({\em split}) operator is assumed to be the inverse of the default sequence {\em concatenation} operator (i.e. divide $s$ into $s_1$ and $s_2$ where $s = s_1 \ccat s_2$). Our input programs are {\em loop nests} that traverse the multidimensional data in accordance with their recursive structure. These programs are assumed to have fundamentally unbreakable data flow dependencies.

\intextsep=0pt
\columnsep=10pt

Consider the code in Figure \ref{fig:mbbs}(a), that implements
a sequential solution to the problem of computing, for a three dimensional $n \times m \times \ell$ array $A$ (with both positive and negative elements), the sum of the elements of a subarray $A[k..n-1,0..m-1,0..\ell-1]$ (for all $0 \le k < n$) which has the maximum sum compared to all other such subarrays. Intuitively, considering the array as a 3D box with height $n$, the goal is to
\begin{wrapfigure}{r}[0pt]{5.85cm}
\includegraphics[scale=0.27]{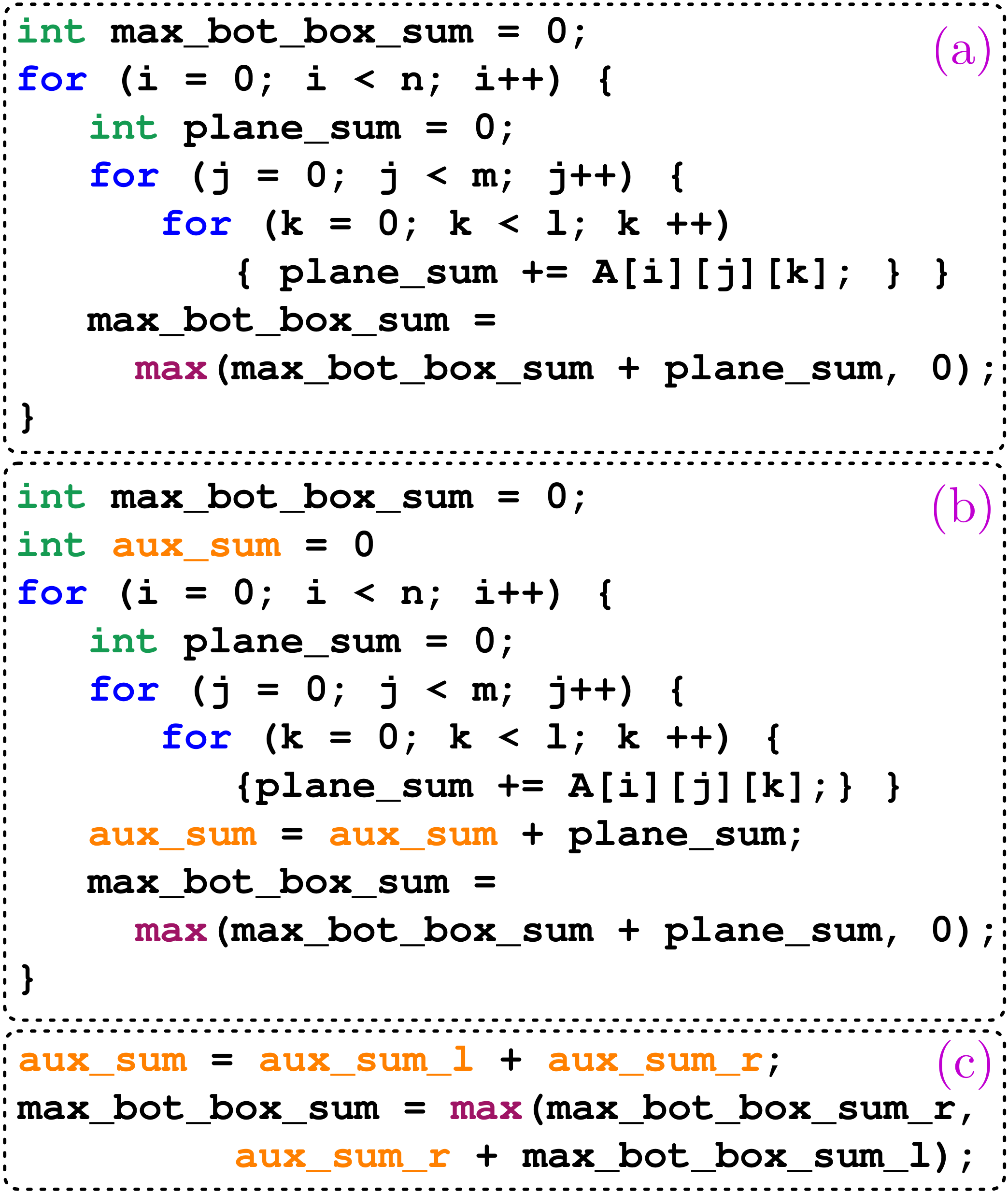}\vspace{-7pt}
\caption{\small Maximum bottom box sum.\label{fig:mbbs}}
\end{wrapfigure}
discover the maximum sum of boxes of different heights, with the same width, length and bottom as the input box.

Note that this optimal sequential implementation runs in {\em single pass} linear time over the input 3D array, at the cost of creating unbreakable loop dependencies. A less efficient solution that would enumerate all boxes would have been easier to parallelize.

It is easy to observe that the code is not (divide-and-conquer) parallelizable. Let us assume it is. There then exists a binary function $\odot$ that can combine results of two instances of the code ($mbbs$) run on two adjacent boxes to produce the same results for the {\em concatenated} box. \vspace{-5pt}
\begin{center}
\includegraphics[scale=0.19]{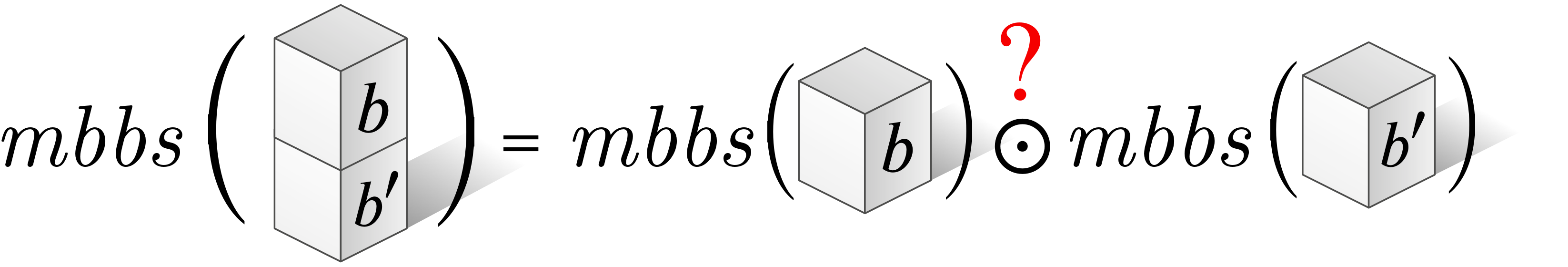}\vspace{-9pt}
\end{center}

Let $b = [5]$ (a $1 \times 1 \times 1$ box) and consider two choices for $b'$, namely $[-3,3]$ and $[0,3]$  ($2 \times 1 \times 1$ boxes). Although $mbbs(b')$ is $3$ in both cases, the join needs to produce two different answers for $mbbs(b\ccat b')$.

Nonexistence of the join operator indicates that $mbbs$, the function computed by the loop, is not a {\em homomorphism}. Now, consider the modified code illustrated in Figure \ref{fig:mbbs}(b). A new accumulator {\tt aux\_sum} is added (in orange), which maintains the sum of the elements in  $A[0..i-1,0..m-1,0..\ell-1]$ at the
\begin{wrapfigure}{r}{2.4cm}% \vspace{-10pt}
\hspace{-10pt}
\includegraphics[scale=0.29]{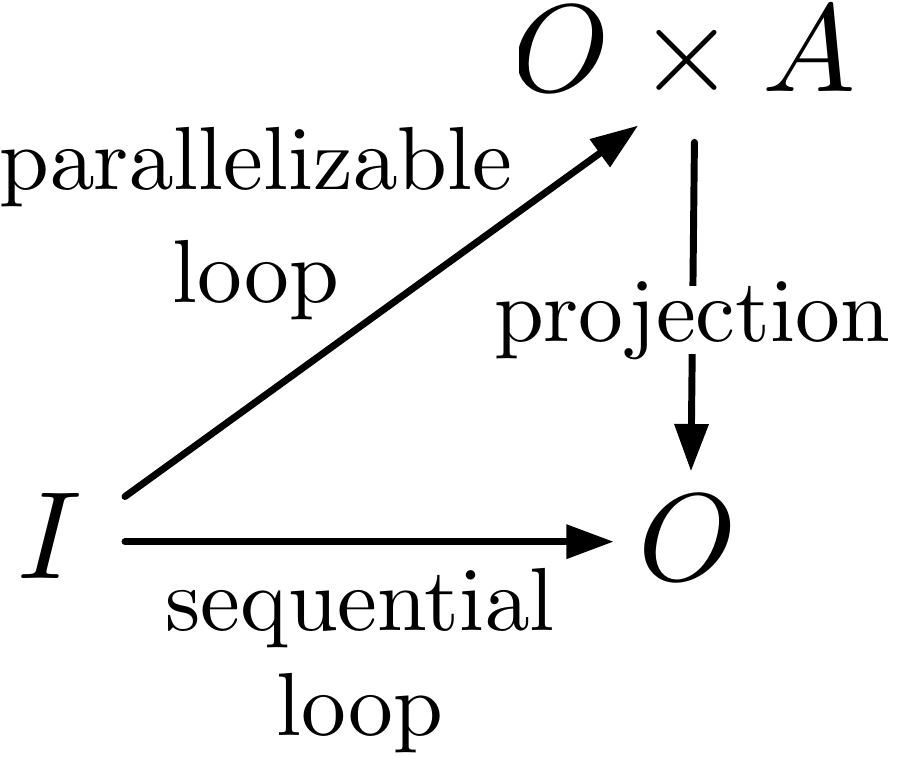}\vspace{-1pt}
\end{wrapfigure}
$i$-th iteration of the outer loop. Note that $mbbs(b)$ is producing a pair of integers now, instead of a single integer. This extending of
a function's signature is called {\em lifting}, in the standard sense of lifting a morphism in category theory, and is
illustrated on the right. $(I,O)$ denotes the input and output of the original sequential loop, and a {\em lifting} of the code additionally computes {\em auxiliary} information denoted by $A$. If the {\em lifted} function is a {\em homomorphism}, then a parallel join exists for it. Figure \ref{fig:mbbs}(c) illustrates the parallel join for the lifted maximum bottom box code.

\vspace{-5pt}
\subsection{Modular Parallelization}\label{sec:mp}

Figure \ref{fig:diagram}(a) illustrates the flow of data in a generic nested loop (of arbitrary depth), where $s_i$ denotes the {\em state} of the loop nest (e.g. a tuple of program variables). The {\em black} arrows correspond to the computation of one instance of the body of outermost loop, while the {\em blue} arrows correspond to the computation of one instance of the inner loop nest.

The goal is to parallelize, {\em divide-and-conquer} style, the outermost loop with the assumption that the dependencies are unbreakable.

In \cite{PLDI17} we proposed a semantic solution to this problem for {\bf simple} (non-nested) loops by {\em lifting} their computations to homomorphisms. To generalize such a semantic solution to nested loops, one comes across the very hard problem of computing a {\em semantic summary} of the functionality of the inner loop nest, to be used in the analysis of the outer loop. Despite big strides in program analysis techniques \cite{gustafsson2006automatic,cordes2009fast}, this type of semantic summary computation remains limited to classes of loops whose invariants (summaries) are within decidable theories, and even then, mostly proof-driven rather than summarizing full functionality.

\begin{figure}[t]
\begin{center}
\includegraphics[scale=0.27]{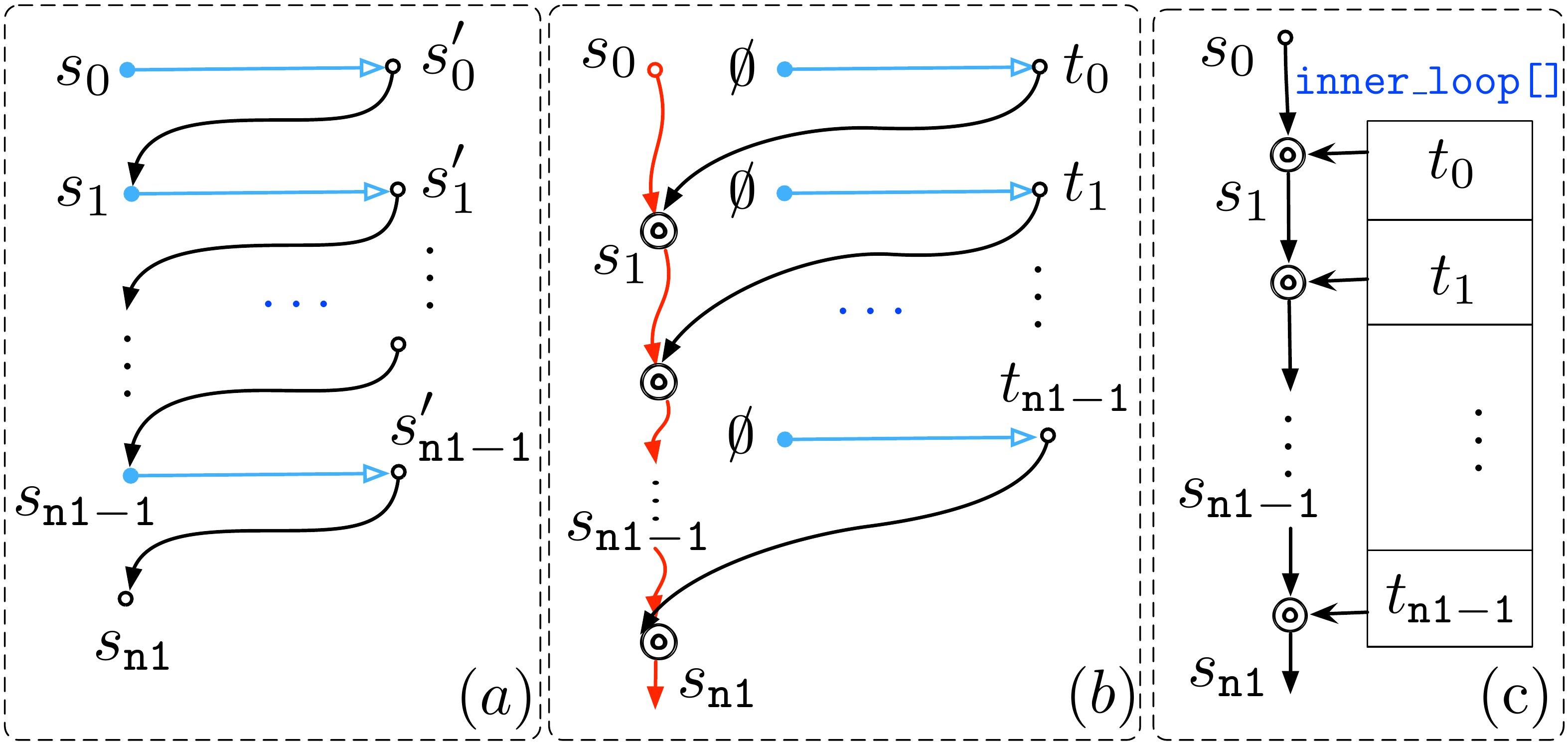} \vspace{-17pt}
\caption{Dependencies in a general sequential nested loop (a) vs. a memoryless one (b), which is summarized in (c).}
\label{fig:diagram}
\end{center}\vspace{-20pt}
\end{figure}

We propose a methodology that circumvents this problem through a modular solution. We divide the dependencies in Figure \ref{fig:diagram}(a) into two categories and resolve them separately. The black arrows force every instance of the inner loop nest to be executed only after the results of all previous instances are ready. Contrast this with the diagram in Figure \ref{fig:diagram}(b), where each instance of the inner loop nest starts from a fixed (constant) initial state $\zero$, and therefore, all instances can be run in parallel. The sequential binary operator $\circledcirc$ merges the results of the inner loop nest ($t_i$) with the current state of the outermost loop ($s_i$) and makes the required adjustments (to get $s_{i+1}$). We call such a loop nest {\em memoryless}. The terminology is inspired by the fact that all the instances of the inner loop nest implement the same function (that starts from the same initial state $\zero$). If a general loop nest is transformed to a memoryless one through the introduction of new computation (i.e.  $\circledcirc$), then this results in the removal of the black arrow dependencies. The inner loop nest  can be executed by a {\em parallel map}. The outermost loop remains sequential. Observe that the loop in Figure \ref{fig:mbbs}(a) is memoryless.

Transforming a general loop to a memoryless one is not always straightforward. Due to lack of information in the loop state, no such binary function (operator) $\circledcirc$ may exist. In these cases, one needs to deduce additional information to be computed by the inner loop nest to facilitate the existence of $\circledcirc$, that is, the inner loop nest has to be {\em lifted}. Transforming a general loop to a memoryless one involves solving two subproblems: (i) producing an implementation for $\circledcirc$, and (ii) discovery of auxiliary computation when such an operator does not exist. Solving these two problems are two of our key contributions (Sections \ref{sec:mjoin} and \ref{sec:mlift}).

\intextsep=0pt
\columnsep=2pt

When the loop is memoryless, the inner loop nest can be abstracted away to get a {\em summarized} (potentially simpler) loop. As shown in Figure \ref{fig:diagram}(c), the
results of the computations of the inner loop nest
are assumed to be stored in a (conceptual) array
(called {\tt inner\_loop[]}), and therefore the loop nest is removed. The {\em summarized loop} fetches the results from {\tt inner\_loop[]} to perform its computation. Any
\begin{wrapfigure}[5]{r}[0pt]{4.3cm}
%\hspace{-35pt}
\includegraphics[scale=0.29]{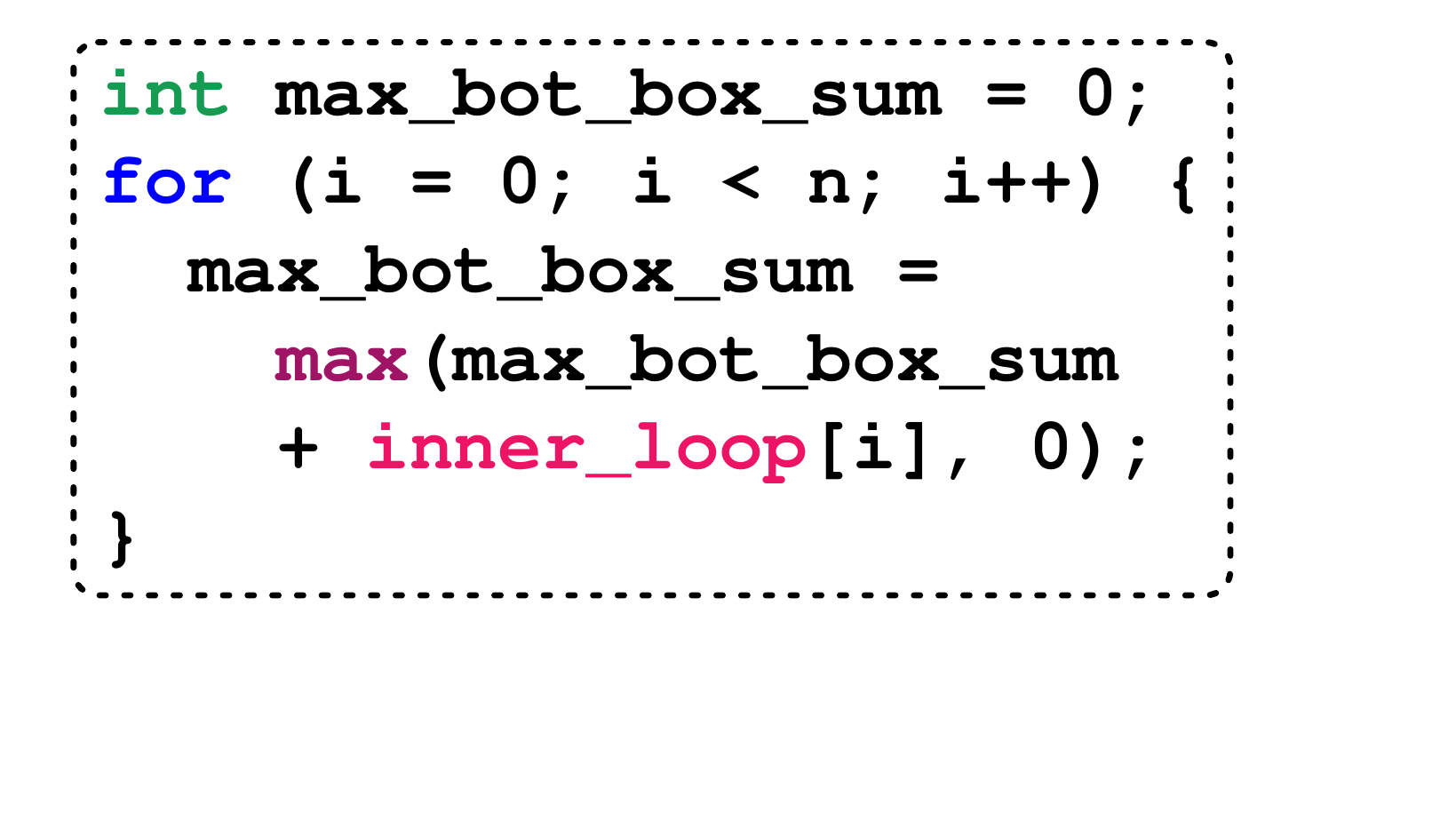}
\end{wrapfigure}
{\em memoryless} loop can be {\em summarized} this way. For example,
the 3-nested loop of Figure
\ref{fig:mbbs}(a) is summarized to a single loop (illustrated on the right).

The crucial observation is that the {\em summarized} loop is {\em efficiently parallelizable} if and only if the original one  is (Theorems \ref{thm:par-cond} and \ref{thm:convl}). Therefore, the problem of parallelizing the original loop is {\em soundly and completely} reducible to the problems of (i) producing the {\em summarized} loop, and (ii) parallelizing it.
Summarization can substantially simplify the parallelization task. For example, the approach in \cite{PLDI17} can parallelize the summarized loop above while it is not applicable to the original loop in Figure \ref{fig:mbbs}(a).

Summarization, however, does not always yield a non-nested loop like the one above, and therefore, the approach in \cite{PLDI17} cannot always parallelize a summarized loop.

To parallelize the summarized loop, two subproblems have to be solved: (a) Automatic {\em lifting} of {\em nested} loops to parallelizable code, and (b) automatic generation of the parallel join for {\em nested} loops. Problem (b) is easier to solve.  In Section \ref{sec:join}, we build on our technique from \cite{PLDI17} to extend it to nested loops. The lifting problem is more complex. We solve it by reducing it to two well-known problems, namely {\em normalization} (in term rewriting systems) and {\em recursion discovery}. In section \ref{sec:alg}, we discuss the reduction and propose simple heuristics for both problems.
Our modular parallelization methodology comprises theoretical results and algorithms for generating all required additional code. Figure \ref{fig:schema} outlines the applications of the theorems and the contributed algorithmic modules, and therefore, serves as detailed summary of our technical contributions.
Due to the undecidability of the problem, some of our algorithms are heuristics. We provide experimental results to demonstrate the effectiveness of these heuristics in fully automatically and efficiently producing divide-and-conquer parallelizations for some highly nontrivial nested loops. Beyond facilitating full automation, we believe that our methodology is also a systematic approach that can guide programmers in writing correct and efficient parallel code manually.

%%% Local Variables:
%%% mode: latex
%%% TeX-master: "paper"
%%% End:

%% file: overview.tex
% !TEX root =  paper.tex

\section{Motivating Examples} \label{sec:overview}

We use two difficult-to-parallelize examples to underline the challenges of parallelizing the class of nested loops targeted in this paper and outline the strengths of our methodology.

\intextsep=0pt
\columnsep=17pt

\subsection{Balanced Parentheses}\label{sec:wbp}\vspace{0pt}
This example demonstrates that transforming a nested loop to a {\em memoryless} one can be complicated. A string is {\em balanced} if the total number of left and right brackets match, and any prefix of the string has at least as many left brackets as right
\begin{wrapfigure}{h}{6.4cm}
  \includegraphics[scale=0.28]{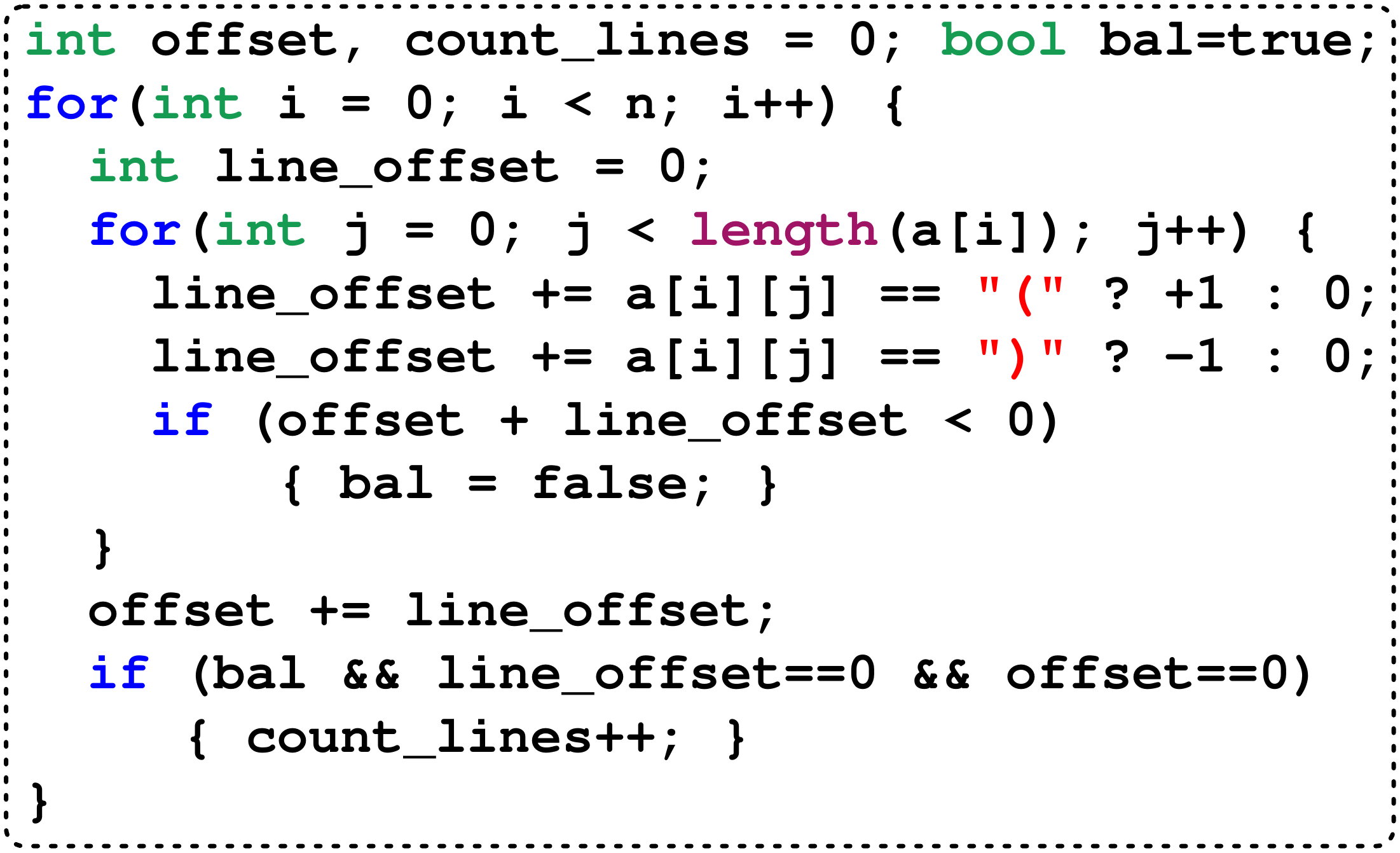}\vspace{-9pt}
  \caption{ Balanced Parentheses.\label{fig:wbp}}
\end{wrapfigure}
ones. Assume that the input is a
two-dimensional array containing a large bracketed math expression, one row per each line.  A line $l$ of input $x$ is {\em self-contained} if we have $x = x_1 \cdot l \cdot x_2$, where $l$ and $x_1$  are both balanced. The code in Figure \ref{fig:wbp} counts the number of {\em self-contained} lines  of its input through a nontrivial algorithm. {\tt offset} maintains the excess of left over right brackets seen so far. {\tt bal} tracks if {\tt offset} has always remained nonnegative.

We encourage the reader to manually parallelize the outer loop to get a sense of the difficulty of this problem.

The loop is not {\em memoryless}; unbreakable dependencies on {\tt bal} and {\tt offset} variables induce the black arrows from the diagram in Figure \ref{fig:diagram}(a). One cannot remove the dependency of the update to {\tt bal} on the value of {\tt offset} without having the inner loop compute an extra value.
Specifically, the minimum value of {\tt line\_offset}, during the execution of the inner loop, should be made available to the outer loop. If this does not cause {\tt offset} to dip below $0$, then {\tt offset + line\_offset} should have remained positive throughout the inner loop execution, and therefore the value of {\tt bal} can be recovered. The code in Figure \ref{fig:wbpm} illustrates the lifted code (modifications are highlighted).
The loop in Figure \ref{fig:wbpm} is {\tt memoryless} and can be summarized as below.
\begin{center}\vspace{-8pt}
\includegraphics[scale=0.28]{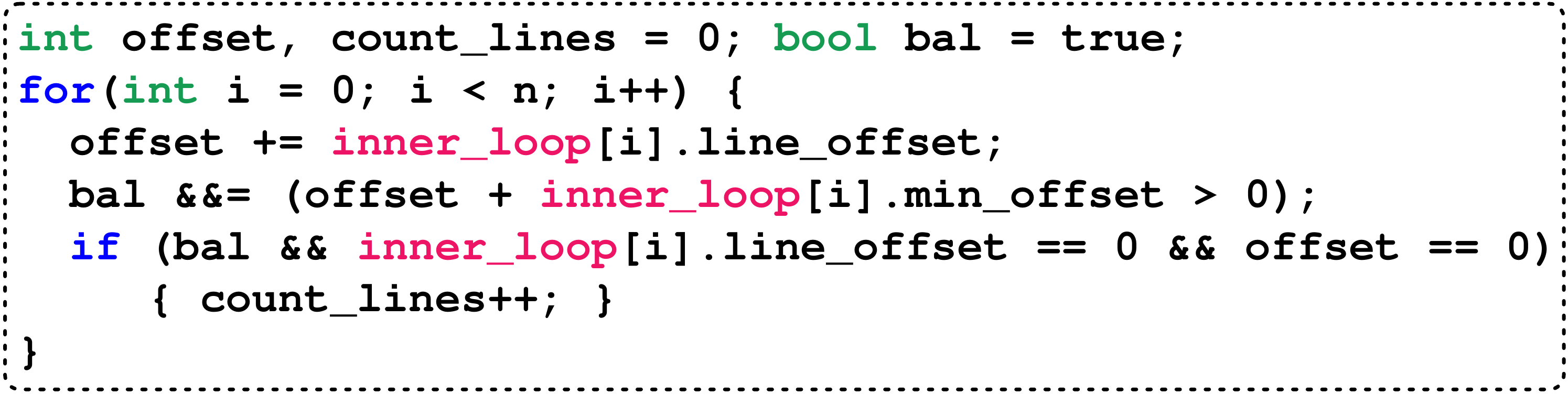}\vspace{-2pt}
\end{center}
In Sections \ref{sec:lift} and \ref{sec:alg}, we discuss how the {\tt min\_offset} accumulator can be discovered automatically. Can the summarized loop (above) be parallelized? No! The reader can verify that a parallel join does not exist. Furthermore, the loop {\em cannot be efficiently lifted} (theoretically impossible); that is, the addition of more scalar accumulators will not transform it to a homomorphism. The transformation of the loop to a memoryless one parallelizes all instances of the inner loop (implementable by a parallel {\em map}). But, the outer loop computation cannot be efficiently turned into a parallel {\em reduction}. Yet, the parallelization of the code through the discovery of the {\em map} alone yields a reasonable speedup (Section \ref{sec:experiments}).

\begin{figure}[t]
\includegraphics[scale=0.29]{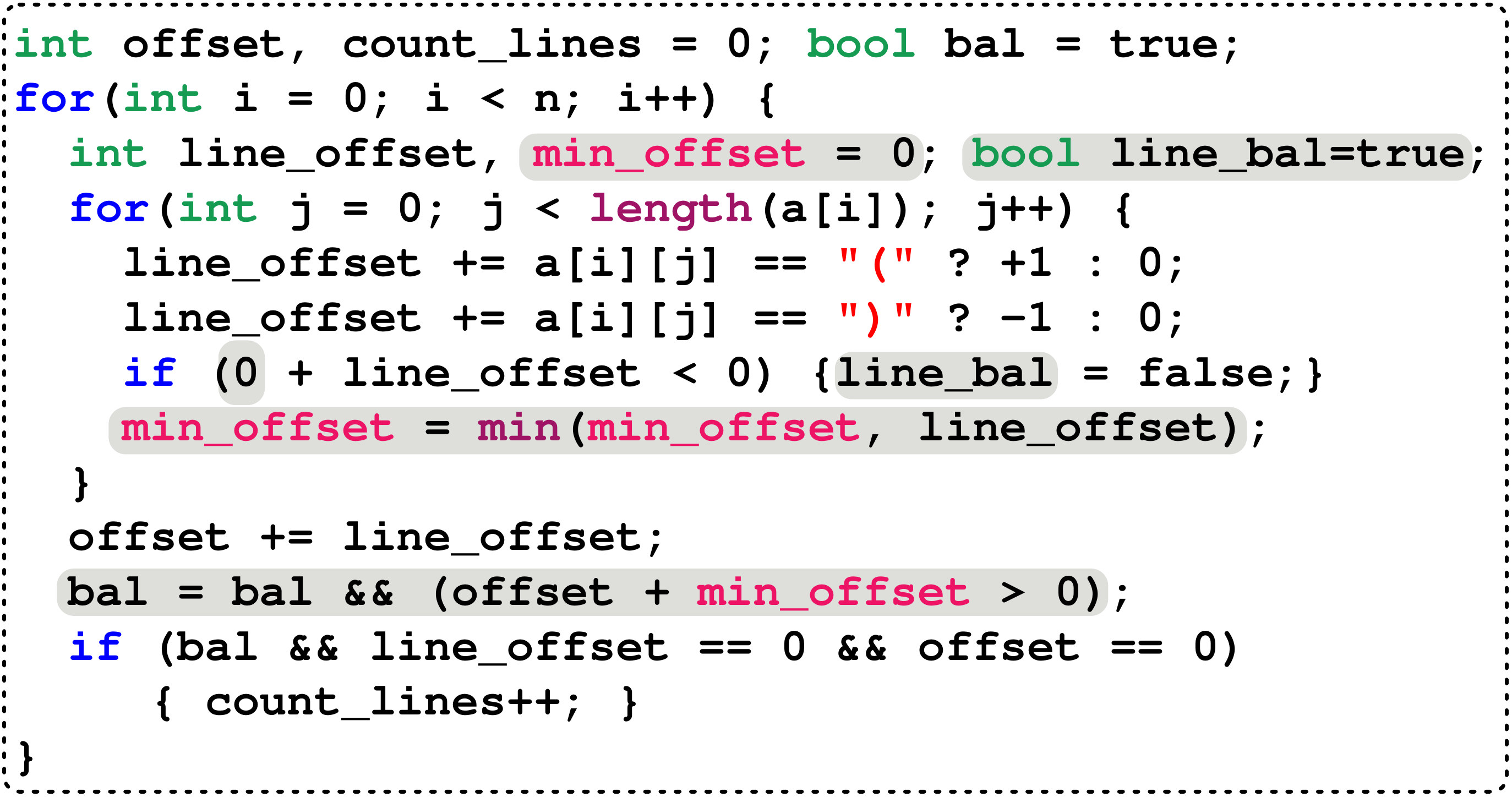}
\caption{Memoryless balanced parentheses.\label{fig:wbpm}}\vspace{-10pt}
\end{figure}\vspace{-10pt}

\begin{figure*}[t]
\begin{center}
\includegraphics[scale=0.29]{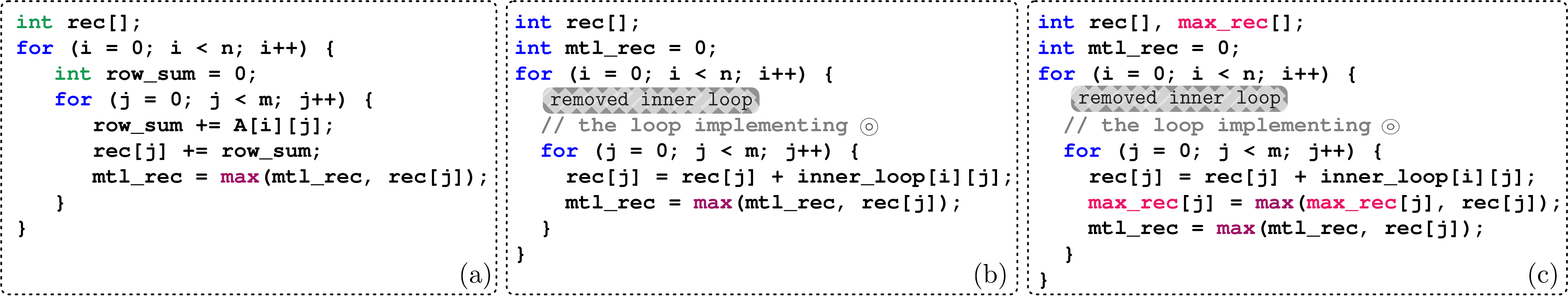}\vspace{-8pt}
\caption{Maximum top-left subarray sum (a), its summarized version (b), and the lifting to parallelizable code (c)}
\label{fig:mtlr}
\end{center}\vspace{-12pt}
\end{figure*}

\subsection{Maximum Top-Left Subarray Sum}\label{sec:mtlr}

This example demonstrates that parallelization of the outer loop may be nontrivial even after a successful summarization.
Consider a two-dimensional array of integers (with both positive and negative) elements. Assume that the goal is to compute the maximum sum of the elements of a subarray $A[0..k,0..\ell]$ for all $0 \le k < n$ and $0 \le \ell < m$, i.e.  all subarrays that include the top-left corner $(0,0)$.

The code in Figure \ref{fig:mtlr}(a) is a clever single-pass implementation of this function. Note that the inner loop has a state (variable) {\tt rec[]} that is the same size as the width of a row ($m$). In {\tt rec[j]}, the loop maintains the sum of all elements in the subarray $A[0..i,0..j]$. The loop is not {\em memoryless} due to the dependencies induced by both {\tt rec[]} and {\tt mtl\_rec}. Again, we encourage the reader to think about how they would parallelize the code manually.

Figure \ref{fig:mtlr}(b) illustrates the memoryless and summarized variation of the code. The transformation is straightforward, but the summarized loop is still a 2-nested loop and not parallelizable (i.e. not a homomorphism); that is, the operator $\circledcirc$ from Figure \ref{fig:diagram}(b) has to be implemented as a simple loop to correctly update variables {\tt rec[]} and {\tt mtl\_rec}. The transformation underlines a subtle point, namely that, the relevant information from the input array is the {\em sum} values of the subarrays starting from the $(0,0)$ and ending at $(i,j)$, and not the values of $\tt A[i][j]$'s. This {\em abstraction} is a key to the simplification of the {\em lifting} of the outer loop to a homomorphism for parallelization.

\begin{figure}[b]
\begin{center}
\includegraphics[scale=0.28]{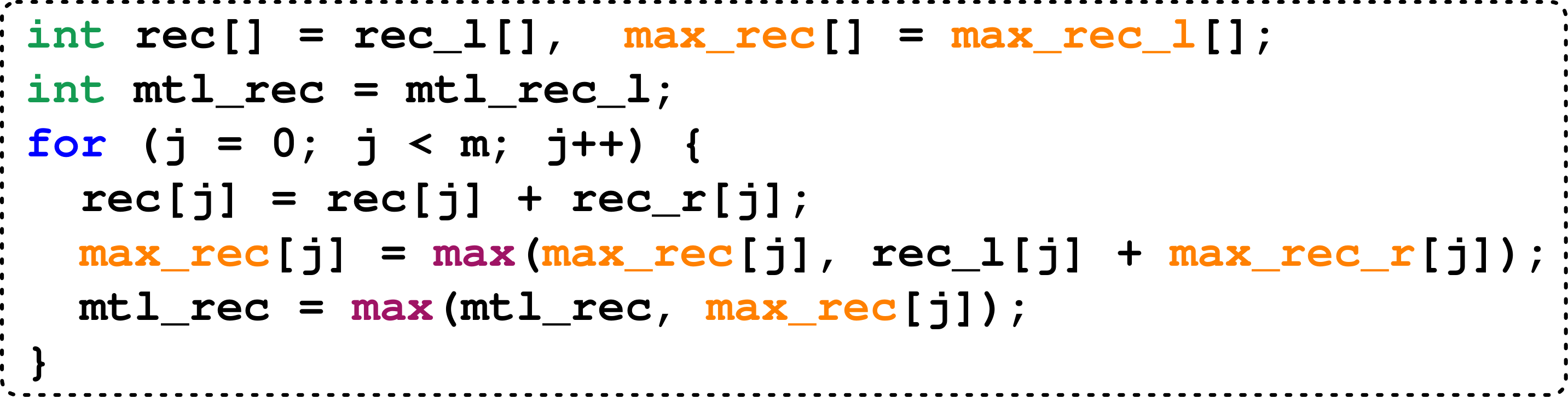}\vspace{-7pt}
\caption{The parallel join for Figure\ref{fig:mtlr}(c). \label{fig:mtlr-join}}
\end{center}
\end{figure}

The code needs to be lifted as illustrated in Figure \ref{fig:mtlr}(c). A new variable {\tt max\_rec[]} has to be introduced where each cell {\tt max\_rec[j]} maintains the maximum value of {\tt rec[j]} (for $0 \leq j < n$).
Discovery of such variables, that is arrays of accumulators, is not required for parallelization of simple loops \cite{PLDI17}. The time complexity budget for a parallel {\em join} operator of a simple loop is constant time, and therefore non-constant sized variables are pointless. For nested loops, however, as this example demonstrates, they may be essential.
In Section \ref{sec:alg}, we propose a new algorithm for discovering liftings like this automatically.

Now, a parallel join operator can combine the value of  {\tt rec[]} from the top thread and {\tt max\_rec[]} from the bottom thread to account for subarrays that intersect two adjacent array chunks, as illustrated in Figure \ref{fig:mtlr-join}.
The two challenges underlined by this example are (i) the synthesis problem of a parallel join operator which is a looping computation, and (ii) the discovery of auxiliary information for lifting which is not constant-sized.

%%% Local Variables:
%%% mode: latex
%%% TeX-master: "paper"
%%% End:

%% file: preliminaries.tex
% !TEX root =  paper.tex

\section{Notation and Background} \label{sec:preliminaries}
This section introduces the notation used in the remainder of the paper. While the formal work is based on studying functions on sequences, the description of the algorithm requires to define our inputs programs and a model for loop bodies which can be translated to a functional form.

\subsection{Sequences and Functions.}
We assume a generic type $\scalars$ that refers to any scalar type used in typical programming languages, such as {\tt int} and {\tt bool} whenever the specific type is not important in the context. Scalars are assumed to be of {\em constant} size, and conversely, any constant-size representable data type is assumed to be scalar. Consequently, all operations on scalars are assumed to have constant time complexity. Type $\seq$ defines the set of all {\em sequences} of elements of type $\scalars$. For any sequence $x$,  $x[i]$ (for $0 \le i < |x|$) denotes the element of the sequence at index $i$, and $x[i .. j]$ denotes the subsequence between indexes $i$ and $j$ (inclusive). The concatenation operator $\ccat: \seq \times \seq \to \seq$ is defined over sequences in the standard way, and is associative.  The sequence type stands in for {\em arrays}, {\em lists}, or any collection data type that admits a linear iterator and an {\em associative} composition operator.

\begin{definition}\label{def:rw}
A function $h: \seq \to D$  is rightward iff there exists a binary operator $\oplus: D \times \scalars \to D$ such that for all $x \in \seq$ and $a \in \scalars$, we have $h( x \ccat [a]) = h(x) \oplus a$.
\end{definition}
\noindent Note that the notion of associativity for $\oplus$ is not well-defined, since it is not a binary operation defined over a set (i.e. the two arguments to the operator have different types). A leftward  function is defined analogously using the recursive equation $h( [a] \ccat x) = a \otimes h(x)$.

Homomorphisms are a well-studied class of mathematical functions. We are interested in a special class of homomorphisms, where the source structure is a set of sequences with the standard concatenation operator.
\begin{definition}\label{def:hom}
A function $h: \seq \to D$ is $\odot$-homomorphic for binary operator $\odot: D \times D \to D$ iff for all sequences $x,y \in \seq$ we have $h(x \ccat y) = h(x) \odot h(y)$.
\end{definition}

Note that $\odot$ is necessarily associative since concatenation is associative (over sequences). Moreover, $h([])$ (where $[]$ is the empty sequence) is the unit of $\odot$, since $[]$ is the unit of concatenation. If $\odot$ has no unit, then $h([])$ is undefined.
There is formal connection between homomorphisms and divide-and-conquer style parallelism, when the divide operator is the inverse of concatenation:

\begin{proposition}(from \cite{Gibbons96}) \label{prop:map-reduce}
A function $f$ is a homomorphism if and only if it can be written as a composition of a map and a reduction.
\end{proposition}
\noindent In the context of this paper, parallelization is formally the above transformation to a map and a reduction composition.

\subsection{Model of a loop body}
Our input programs are imperative whereas the representation of the loop nests for the theoretical results in this paper and for algorithmic units is functional.  The input program is translated to nested systems of equations, which can easily be converted to a recursive functional form. Here, we quickly outline the steps of this transformation and define the program models at each stage.

\paragraph{Input programs}
Figure \ref{fig:syntax} presents the syntax of the input sequential programs. We assume an imperative language with basic constructs for branching and looping. Variables are of scalar types \texttt{int} or \texttt{bool} and we can build nested sequences from these types.

\begin{figure}[htbp]
\begin{center}
  \framebox{
    \begin{minipage}{7.7cm}
      \begin{minipage}[t]{\linewidth}
        $v \in \LhVar ::= v'[e] $ \hfill $v' \in \LhVar, e\in \Exp$\\
        \hspace*{55pt}$|\ x$ \hfill $x \in \Var$\\
      \end{minipage}%
      \hfill
      \begin{minipage}[t]{\linewidth}
        $e \in \Exp ::= e \circ e'$ \hfill $e,e' \in \Exp$\\
        \hspace*{45pt}$|\ e \olessthan e'$ \hfill $e,e' \in \Exp$\\
        \hspace*{45pt}$|\  be \curlywedge be' | \neg be$ \hfill $be,be' \in \Exp$\\
        \hspace*{45pt}$|\  v $ \hfill $v \in \LhVar$\\
        \hspace*{45pt}$|\  \ite{be}{e}{e'}$  \\
        \hspace*{45pt}$|\  k$ \hfill $ k \in \mathbb{Z}, \mathbb{Q}, \mathbb{R}$\\
        \hspace*{45pt}$|\  \true \ |\ \false$ \hfill
      \end{minipage}            %
      \hfill
      \begin{minipage}[t]{\linewidth}\small
        $\Pr ::= c;c'$ \hfill $c,c' \in \Pr$ \\
        \hspace*{27pt}$|v \texttt{:=} e$ \hfill $v \in \LhVar, e \in \Exp$ \\
        \hspace*{27pt}$| \texttt{if} \: (e) \: \{ c_\top \} \: \texttt{else} \: \{c_\bot\}$ \hfill $be \in \Exp, c_\top, c_\bot \in \Pr$\\
        \hspace*{27pt}$| \loopfor{i \in \mathcal{I}}{ c }$ \hfill $i \in \iter$
      \end{minipage}
    \end{minipage}}%
\caption{\small Program Syntax \label{fig:prog_syntax}.  The binary $\circ$
      operator represents any arithmetic operation $(+,-,*,/)$,  $\olessthan$ operator
      represents any comparator $(<, \le, >, \ge, = , \neq)$. $\mathcal{I}$ is an iteration domain, and $\curlywedge$ operator
      represents any boolean operation $(\wedge, \vee)$.
      \label{fig:syntax}}
  \end{center}
\end{figure}

For readability in our paper, we use simple iterators and integer indexes (instead of the generic $i \in \mathcal{I}$). In principle, any collection with an iterator and a split function that implements the inverse of concatenation works. There has been a lot of research on iteration spaces and iterators (e.g. \cite{ZuckPFGH02} in the context of translation validation and \cite{Kejariwal04} in the context of partitioning) that formalize complex traversals by abstract iterators.

\paragraph{State and Input Variables}  Let $\Var$ be the set of all variables that appear in the loop nest. We partition $\Var$ into two sets of variables: $\SVar$ denotes the set of {\em state variables} which are those that appear on the left-hand side of an assignment statement (anywhere, even outside the loop nest). $\IVar$ denotes the set of {\em input variables} and $\IVar = \Var - \SVar$. Note that state variables may be subscripted array accesses.

\paragraph{Nested systems of equations}
\begin{figure}[h]
  \begin{minipage}{5.8cm}\small
   \framebox{\begin{minipage}{\linewidth}
   $E = \left(
     \begin{matrix}
       v_1 = Exp_1(\SVar, \IVar)\\
       \cdots\\
       (s_{i_1}, s_{i_2}, \dots , s_{i_p}) = \loopfor{j \in \mathcal{J}}{E'}\\
       \cdots \\
       v_q = Exp_q(\SVar, \IVar)\\
     \end{matrix}
   \right)$
 \end{minipage}}
 \caption{\small Nested systems of equations: $E'$ is nested in $E$.}
 \label{fig:loop-body}
\end{minipage}
\end{figure}
A loop body is modelled by a system of ordered recurrence equations, where each equation is either a {\em simple} equation or a {\em loop} equation. Given state variables $\SVar = \{s_1, \dots, s_q\}$ and input variables $\IVar$, a simple equation is of the form $v_i = Exp_i (\SVar, \IVar)$ where $v_i \in \LhVar$ and  the right hand side is a constant-time computable expression of the input program (see Figure \ref{fig:syntax}). A loop equation is of the form of the middle line of Figure \ref{fig:loop-body}, where $\{s_{i_1}, s_{i_2}, \dots , s_{i_p}\}$ are all the variables modified by the loop body,  $j \in \mathcal{J}$ is an arbitrary iterator,  and $E'$ is the body of the nested loop.

The body of any loop in our input language can be translated to the system of recurrence equations defined above.

\paragraph{Conversion to a system of equations}
 Converting the body of a loop nest to a system of ordered recurrence equations (of the type outlined by Figure \ref{fig:loop-body}) is a process that involves a transformation of the loop and conditional statements, and a mapping of  simple assignments ($v \texttt{:=} e$) to equations.

For a conditional statement $\texttt{if} \: (e) \: \{ c_\top \} \: \texttt{else} \: \{c_\bot\}$ where $e$ is an expression of the input program and $c_\top$ and $c_\bot$ are two programs, we apply the conversion procedure recursively to each of the programs, and obtain two systems of ordered recurrence equations $E_{\top}$ and $E_{\bot}$. For each variable $v_i$ that appears either in $E_{\top}$ or $E_{\bot}$ on the left hand side of an equation, we add an equation of the form $v_i = \tern{e}{Exp_{\top}}{Exp_{\bot}}$ in the current system, where the expressions $Exp_\top$ and $Exp_\bot$ are the expressions on the right hand side of $v_i = \dots$ in $E_\top$ and $E_\bot$ respectively. If the equation assigning $v_i$ is not present in one of the branches, the expression on the right hand side is just the variable itself (the branch does not modify it).

For a loop  $\loopfor{i \in \mathcal{I}}{ c }$ where $c$ is a program, we apply the conversion procedure to $c$ and obtain a system $E'$. In the parent system, we add the equation $(s_{i_1}, s_{i_2}, \dots, s_{i_q}) = \loopfor{i \in \mathcal{I}}{ E'}$ where $(s_{i_1}, s_{i_2}, \dots, s_{i_q})$ are the state variables modified by the body $c$. If only one cell of a collection is assigned in the loop body $c$, we consider that the whole collection has been modified.

\paragraph{Conversion to functional form}

Given a loop body in the form of a system of ordered recurrence equations, one can produce a function (implemented in a simple functional language with let-bindings) by replacing each equation by a binding and creating a recursive function for each of the inner loops. We choose to represent arrays by lists, and an assignment to a cell in the system of equations is translated by binding a list where the corresponding element has been modified.

%%% Local Variables:
%%% mode: latex
%%% TeX-master: "paper"
%%% End:

%% file: function.tex
% !TEX root =  paper.tex
\section{Multidimensional Collections} \label{sec:mfunction}

Type $\seq^n$ is inductively defined as the set of all $n$-dimensional {\em sequences}  (for $n \ge 1$), with the base case of $\seq^0 = \scalars$ (set of scalars). We generalize the standard sequence concatenation operator $\ccat$ to a family of operators $\ccat: \seq^n \times \seq^{n} \to \seq^n$ (for all $n \in \mathbb{N}^+$). For any $\sigma \in \seq^{n-1}$, we have $[\sigma] \in \seq^n$ which is an $n$-dimensional sequence with a single element $\sigma$.

\subsection{Functions over Multidimensional Collections}
In Section \ref{sec:preliminaries}, we noted that loop nests are translated to functional form. We use this functional form as the formal representation for all of our theoretical results.

\begin{definition}(Multidimensional Rightward)\label{def:mrw}
A function $f: \seq^n \to D$ ($n > 1$) is rightward iff there exists a family of rightward (or leftward) functions ${\mathbbm G}:  D \to (\seq^{n-1}  \to D)$ and an operator $\otimes: D \times D \to D$ such that for all $\sigma \in \seq^n, \delta \in \seq^{n-1}$, we have
 $f(\sigma \ccat [\delta]) =  f(\sigma) \otimes \mathbbm{G}(f(\sigma))(\delta)$.
\end{definition}
The base case of $n = 1$ falls on the classic Definition \ref{def:rw}. A rightward function's computation is illustrated in the diagram in Figure \ref{fig:diagram}(a). Note that the value of $f(\sigma)$ (as the selector in the family of functions) serves as a type of {\em carry over state} and corresponds to the data flow represented by the black arrows in Figure \ref{fig:diagram}(a). The family of functions can be viewed as only differing in their recursion base case.

When $f$ corresponds to a loop nest, the family of rightward functions $\mathbbm{G}$ represents all the instances of the inner loop nest (in isolation from the outermost loop) and the operator $\otimes$  represents the (loop free) computation performed in the body of the outer loop. The domain $D$ corresponds to all valuations of the state variables ($\SVar$) of the loop nest.

A special case of Definition \ref{def:mrw} is when the family of functions collapses into exactly one function, which corresponds to {\em memoryless} loops as introduced in Section \ref{sec:mp}. We can formally define {\em memoryless functions} by removing the dependency on the context as follows:

\begin{definition}(Memoryless)\label{def:memoryless}
A function $f: \seq^n \to D$  is (rightward) memoryless iff there exists a rightward (or leftward) function $g: \seq^{n-1} \to D$ and a binary operator $\oplus: D \times D \to D$ such that for all $\sigma \in \seq^n,\delta \in \seq^{n-1}$ we have $
 f(\sigma \ccat [\delta]) = f(\sigma) \oplus g(\delta)$.
\end{definition}

The key difference between the formulation in Definition \ref{def:mrw}, and that of Definition \ref{def:memoryless} is the computation performed over $\delta$ (i.e. function $g$) has no dependency on the partially computed value of $f(\sigma)$; hence the use of terminology {\em memoryless}. Figure \ref{fig:diagram}(b) illustrates the computation of a memoryless function. As the example in Section \ref{sec:wbp} demonstrated, not all rightward functions are memoryless.

\begin{proposition}\label{prop:inmap}
For every rightward memoryless function $f$ (from Definition \ref{def:memoryless}), we have
$f(\sigma) = \foldl(\oplus) \circ \map(g)(\sigma)$.
\end{proposition}

The proof of the above proposition is straightforward. It suggests that all instances of $g$ (the inner loop nest) can be parallelized, through the $\map$, even if their results have to be combined sequentially in the outermost loop with $\foldl$.

\subsection{Multidimensional Homomorphisms}

Definition \ref{def:hom} applies to multidimensional rightward functions in a straightforward way. Function $h: \seq^n \to D$ is  $\odot$-homomorphic for the binary operator $\odot: D \times D \to D$ iff for all sequences $\sigma, \sigma' \in \seq^n$, we have $h(\sigma \ccat \sigma') = h(\sigma) \odot h(\sigma')$.
An interesting link exists between the structure of a multidimensional rightward function and its homomorphic properties, which is captured by the proposition below:

\begin{proposition}\label{prop:nec}
If a function $h: \seq^n \to D$ is a homomorphism, then it is memoryless.
\end{proposition}

\input{proofs/nec}

The converse of Proposition \ref{prop:nec} does not hold.

\begin{example}
Recall the maximum bottom box example from Section \ref{sec:introduction}. The function corresponding to Figure \ref{fig:mbbs} is memoryless, but as discussed, not a homomorphism.
\end{example}

For a memoryless function to be a homomorphism, an extra condition is required which is outlined below.

\begin{proposition}\label{prop:conv}
If a function $f: \seq^n \to D$ is (rightward) memoryless and defined
by function $g$ and binary operator $\oplus$ (of Definition \ref{def:memoryless}),
{\bf and} if the function $h: \seq_{D} \to D$ defined as
\begin{align*}
h([]) &= f([]) \\
\forall a \in D: h(x \ccat [a]) &= h(x) \oplus a
\end{align*}
is $\odot$-homomorphic for some binary operator $\odot: D \times D \to D$, then $f$ is $\odot$-homomorphic. We refer to function $h$ as the {\bfseries summarized} version of $f$.
\end{proposition}

Function $h$ corresponds to the concept of a {\em summarized} loop as introduced in Section \ref{sec:mp}. In fact, we can prove that the sufficient conditions in Proposition \ref{prop:conv} are also necessary.

\begin{theorem}\label{thm:par-cond}
The following two statements are equivalent:
\begin{enumerate}
\item Multidimensional rightward function $f$ is $\odot$-homomorphic for some binary operator $\odot: D \times D \to D$.
\item $f$ is memoryless {\bfseries and} function $h: \seq_{D} \to D$, the summarized version of $f$ (see Prop. \ref{prop:conv}) is $\odot$-homomorphic.
\end{enumerate}
\end{theorem}

\input{proofs/par-cond}

Theorem \ref{thm:par-cond} states the necessary and sufficient conditions for a recursive function to be parallelizable. For one-dimensional sequences, the statement becomes trivial when the summarized version of the function and the function itself coincide.

The condition of memorylessness captures the essence of modularity of our approach.  Instead of determining parallelizability of $f$ through a direct discovery of a join ($\odot$) for $f$, Theorem \ref{thm:par-cond} lets us check if $f$ is memoryless first,  and then discover a join for a simplified (summarized) version of $f$ (i.e $h$). Recall the diagram in Figure \ref{fig:diagram}(b). Memorylessness of $f$ corresponds to the existence of the {\em map} part a parallel computation of $f$. Parallelizability of $h$ corresponds to the existence of the {\em reduction} part of a parallelizaiton of $f$. The combination of the existence of both the map and the reduction is equivalent to $f$ being homomorphic (according to Proposition \ref{prop:map-reduce}). Theorem \ref{thm:par-cond} makes this formal.

%%% Local Variables:
%%% mode: latex
%%% TeX-master: "paper"
%%% End:

%% file: proofs/nec.tex
\begin{proof}
Since $h$ is $\odot$-homomorphic, for all sequences $\sigma,\sigma' \in \seq^n$ we have:
\[h(\sigma \ccat \sigma') = h(\sigma) \odot h(\sigma')\]
and therefore, more specifically, for all $\sigma\in \seq^n$ and $\delta \in \seq^{n-1}$ we have:
\[h(\sigma \ccat [\delta]) = h(\sigma) \odot h([\delta])\].
Now, let $g: \seq^{n-1} \to D'$ be defined so that $g(\delta)  = h([\delta])$ and let $\oplus = \odot$ in Definition \ref{def:memoryless}; we can conclude that  $h$ is memoryless.
\end{proof}

%%% Local Variables:
%%% mode: latex
%%% TeX-master: "../paper"
%%% End:

%% file: proofs/par-cond.tex
\begin{proof}

\begin{itemize}
\item[$\Rightarrow$:] By Proposition \ref{prop:nec}, we can conclude that $f$ is memoryless. Let $x, y \in \seq_D$,
$y = y_1 \ccat \dots \ccat y_k$,  $y_i = g(\delta_i)$ for some $\delta_i \in \seq^{n-1}$,
$x = x_1 \ccat \dots \ccat x_m$,  and $x_i = g(\gamma_i)$ for some $\gamma_i \in \seq^{n-1}$.

\begin{align*}
h(x \ccat y) &= h(x \ccat [y_1  \dots y_k]) \\
              &= h(x \ccat [y_1 \dots y_{k-1}]) \oplus y_k \\
              &= (h(x \ccat [y_1 \dots y_{k-2}]) \oplus y_{k-1}) \oplus y_k \\
              &= \dots \\
              &= ( \dots (h(x) \oplus y_1) \oplus \dots ) \oplus y_k\\
              &= \dots \\
              &= ( \dots (h([]) \oplus x_1) \oplus \dots ) \oplus y_k              \\
              &= ( \dots (f([]) \oplus g(\gamma_1)) \oplus \dots ) \oplus g(\delta_k) \\
			  &= \dots\\
			  &= f(\gamma_1 \ccat \dots \ccat \gamma_m \ccat \delta_1 \ccat \dots \delta_k)\ \\
			  &= f(\gamma_1 \ccat \dots \ccat \gamma_m) \odot f(\delta_1 \ccat \dots \delta_k) \\
			  &= \dots \\
			  &= h(x) \odot h(y)
\end{align*}
\item[$\Rightarrow$:] Let $\sigma = \gamma_1 \ccat \dots \ccat \gamma_m$, $\sigma' = \delta_1 \ccat \dots \delta_k$,
$y = y_1 \ccat \dots \ccat y_k$,  $y_i = g(\delta_i)$,
$x = x_1 \ccat \dots \ccat x_m$,  and $x_i = g(\gamma_i)$.

\begin{align*}
f(\sigma \ccat \sigma') &= f(\gamma_1 \ccat \dots \ccat \gamma_m \ccat delta_1 \ccat \dots delta_k) \\
						&= \dots\\
						&= ( \dots (f([]) \oplus g(\gamma_1)) \oplus \dots ) \oplus g(\delta_k) \\
						&= ( \dots (h([]) \oplus x_1) \oplus \dots ) \oplus y_k \\
						&= \dots \\
						&= h(x \ccat y) \\
						&= h(x) \odot h(y) \\
						&= \dots \\
						&= f(\gamma_1 \ccat \dots \ccat \gamma_m) \odot f(\delta_1 \ccat \dots \delta_k) \\
						&= f(\sigma) \odot f(\sigma')
\end{align*}

\end{itemize}

\end{proof}

%%% Local Variables:
%%% mode: latex
%%% TeX-master: "../paper"
%%% End:

%% file: lift.tex
% !TEX root =  paper.tex
\section{Manufacturing Homomorphisms} \label{sec:lift}

If a function is not a homomorphism, then the first step to parallelization is to {\em lift} it to a homomorphism.

\begin{definition}(Lifting)
Let $f: \seq^n \to D$ be a rightward multidimensional function. $\lift{f}{D'}: \seq^n \to D \times D'$ is a lifting of $f$ if and only if $\lift{f}{D'}$ is rightward and $f = \pi_D \circ \lift{f}{D'}$, where $\pi_D$ is the standard projection down to $D$.

\end{definition}
This definition is mostly consistent with the standard definition of lifting in category theory, other than the additional condition of rightward computability of the extension.

Two types of liftings of a non-homomorphic function $f$ are of interest in this paper: (1) a lifting of a non-memoryless $f$ to a memoryless function; we call this the {\em memoryless lift}, and (2) a lifting of a non-homomorphic $f$ to a homomorphism; this is called a {\em homomorphism lift}.

\subsection{Homomorphism Lift}\label{sec:hom-lift}

Every non-homomorphic function can be made homomorphic by a rather trivial lifting. The observation, previously made in \cite{Gorlatch99}, is formalized below:

\begin{proposition}\label{prop:hom-triv}
Given a rightward function $f: \seq^n \to D$, the function $f \times \iota$ (function product) is a homomorphism where $\iota: \seq^n \to \seq^n$ is the identity function.
\end{proposition}
Intuitively, the extension to the function remembers the entire input, and the join performs the original computation over the concatenated inputs from scratch, ignoring the partially computed results.

\input{proofs/hom-triv}

Note that this trivial lifting does not really correspond to a parallelization of the function. Formally, it provides us with an associative {\em reduction} (hence the applicability of Proposition \ref{prop:map-reduce}). Practically, it is analogous to a sequential computation.
Proposition \ref{prop:hom-triv} is trivial but significant in that it states that a function can always be made homomorphic. It is then important to seek {\em an efficient lifting} of a non-homomorphic function to a homomorphism for the purpose of code parallelization. In Section \ref{sec:eff-hom}, we formulate {\em efficient liftings}.

Here, we state a result which parallels Theorem \ref{thm:par-cond}, provides the theoretical guarantee that it is {\em sound and complete} to use the {\em summarized} loop for lifting instead of the original. Consider the diagram below:

\begin{center}
\includegraphics[scale=0.25]{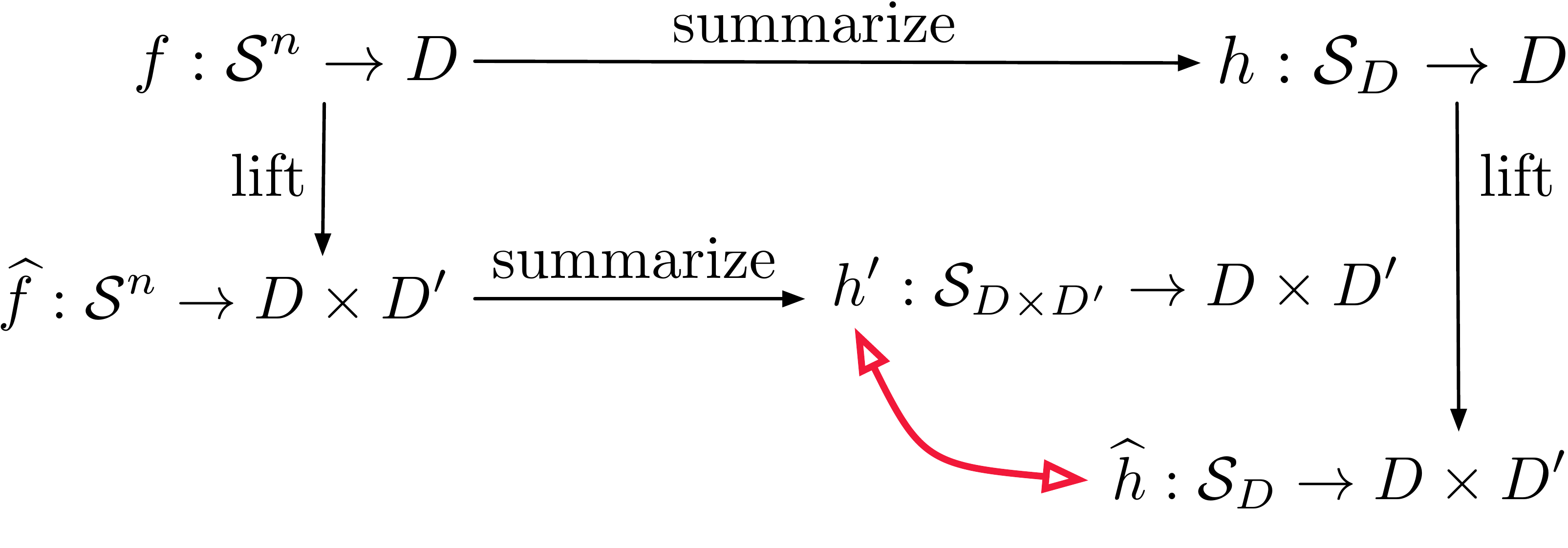}
\end{center}

$f$ is summarized and then lifted on the top, whereas it is first lifted and then summarized on the bottom part of the diagram.  Note that $\widehat{h}$ and ${h'}$ do not have the same function signature; they agree on their ranges, but their domains are sequences of two different types.
Therefore, this is not a clean commutative diagram.
The key insight is that the two functions are identical up to a limitation of $h'$ that forgets the extra information in its input sequences from $D'$; information that is provably redundant for the computation of $h'$. The diagram commutes after this restriction is applied to $h'$ to get to $\widehat{h}$.

The main ingredients of a lift, that is what the extra information $D'$ is and how it should be computed, are both discoverable through a lifting of the simple function $h$ in place of $f$.

\begin{theorem}\label{thm:convl}
Let  $f: \seq^n \to D$ be a (rightward) memoryless function,  and summarized as $h: \seq_{D} \to D$. There exists a homomorphic lifting $\widehat{h}: \seq_{D} \to D \times D'$ of $h$ if and only if there exists a homomorphic lifting $\widehat{f}: \seq^n \to D \times D'$ of $f$. Moreover, $\widehat{h}$ coincides with a summarization of $\widehat{f}$.
\end{theorem}
Additionally, the theorem guarantees that auxiliary code synthesized for the summarized loop constitutes a lifting of the original loop.

\input{proofs/convl}

\subsection{Memoryless Lift}\label{sec:mem-lift}

When a rightward function $f: \seq^n \to D$ is not memoryless, a {\em lifting} may be required to add extra information to the signature of the function (state of the loop) so that functions $g$ and $\circledcirc$ from Definition \ref{def:memoryless} exist. Every non-memoryless function can be made memoryless by a rather trivial lifting.
\begin{proposition}\label{prop:lift-triv}
Given a rightward function $f: \seq^n \to D$, the function $f \times \iota'$ (function product) is memoryless where $\iota': \seq^{n} \to \seq^{n-1}$ is defined as $
\forall \delta \in \seq^{n-1}, \sigma \in \seq^n:\ \iota'(\sigma \ccat [\delta]) = \delta$.
\end{proposition}

\input{proofs/lift-triv}

In this trivial lifting, the extension to the function remembers the last line of the input $\sigma \ccat [\delta]$, that is $\delta$, in a new component and the join effectively processes $\delta$ from scratch, ignoring the partially computed results by the inner loop computation.

It is essential, however, that the {\em cheapest} possible (non-trivial) lifting is used, to gain optimal parallelism. Recall the balanced bracket example from Section \ref{sec:wbp}. The lifting (additions of {\tt min\_offset} and {\tt line\_bal} state variables) in that example is an instance of a non-trivial lifting. Proposition \ref{prop:lift-triv}, in contrast, would suggest a simple admissible lifting which would not lead to as much parallelism.

\subsection{Algorithmic Memoryless Lift}\label{sec:mlift}
Algorithmically, the problems of lifting a function to a homomorphism or to a memoryless function are related.
When a function is not memoryless, it means that there is not enough information for a {\em memoryless join} operator ($\circledcirc$) to exist in the style of the diagram in Figure \ref{fig:diagram}(b). Where the {\em homomorphic lifting} algorithm asks what extra computation is required for the results of two instances of the entire loop nest to be joined together,  an algorithm for {\em memoryless lifting} asks what extra computation is required for an instance of the loop nest to be joined with an instance of the {\em inner loop nest}. Considering that the two functions share the same signature, the problem is formally that of joining an inner loop nest to an arbitrary state $\vec{s}$, which is the same problem as the homomorphism lift of the  inner loop nest. The following proposition makes this observation precise.

\begin{proposition}\label{prop:mlift}
A multidimensional rightwards function $f$ defined through a family of functions $\mathbbm{G}$ (as in Definition \ref{def:mrw}) can be lifted to a memoryless function if every member of $\mathbbm{G}$ can be lifted to a $\circledcirc$-homomorphism for some $\circledcirc$.
\end{proposition}

\input{proofs/mlift}

%%% Local Variables:
%%% mode: latex
%%% TeX-master: "paper"
%%% End:

%% file: proofs/hom-triv.tex
\begin{proof}
It is straightforward to see that $f \times \iota$ is a $\circledcirc$-homomorphic with the join operator $\circledcirc : (D \times \seq^n) \times (D \times \seq^n) \to D \times \seq^n$ which is defined as

\[\forall a,b \in D, \sigma, \sigma' \in \seq^n:\ (a, \sigma) \circledcirc (b,\sigma') = (f(\sigma \ccat \sigma'), \sigma \ccat \sigma')\]

since

\begin{align*}
f\times\iota(\sigma \ccat \sigma') &= (f(\sigma \ccat \sigma'), \sigma \ccat \sigma') \\
						&= (f(\sigma), \sigma) \circledcirc (f(\sigma'), \sigma')
						&= f\times\iota(\sigma) \circledcirc f\times\iota(\sigma')
\end{align*}
\end{proof}

%%% Local Variables:
%%% mode: latex
%%% TeX-master: "../paper"
%%% End:

%% file: proofs/convl.tex
% Theorem \ref{thm:convl} in the body of the paper is stated in a way that is easy to read and understand, but technically imprecise. We start by stating the technically precise version of the theorem, and then prove it correct. The precise version of the theorem is illustrated in the diagram below:

% \begin{center}
% \includegraphics[scale=0.23]{figures/thm-diagram1}
% \end{center}

% We then state and prove each path in the diagram as a separate proposition, which correspond to the {\em if} and the {\em only if} directions of Theorem \ref{thm:convl}.

In order to give a proof of Theorem \ref{thm:convl}, we state and prove each path in the diagram as a separate proposition, which correspond to the {\em if} and the {\em only if} directions of Theorem \ref{thm:convl}.

\begin{proposition}\label{aprop:lift}
Let  $f: \seq^n \to D$ be a rightward) memoryless function defined by helper function $g$ (from Definition \ref{def:memoryless}), and let $h: \seq_{D} \to D$ be its summarized version defined through $\oplus$. If $h$ can be lifted to a $\odot$-homomorphic function
$h': \seq_{D} \to D \times D'$, then  there exists a lifting $f': \seq^n \to D \times D'$ of $f$ that is $\odot$-homomorphic, and $h$ is equivalent to the summary of $f'$ up to the projection of its input sequence down to domain $D$.
\end{proposition}

\begin{proof}
Assume there exists a lifting of $h$ called $h'$ that is $\odot$-homomorphic. Then, for all sequences $x,y \in \seq_D$ we have:

$$h'(x \ccat y) = h'(x) \odot h'(y)$$

and therefore, more specifically, for all $x \in \seq_D$ and $d \in D$ we have:

$$h'(x \ccat [d]) = h'(x) \odot h'([d])$$.

Now, let $g': \seq^{n-1} \to D \times D'$ be defined so that $g'(\delta)  = h'([g(\delta)])$. Define $f'$ as:

\begin{align*}
f'([]) &= h'([]) \\
f'(\sigma \ccat [\delta]) &= f'(\sigma) \odot g'(\delta)
\end{align*}
which is by definition {\em memoryless}. For all $\delta \in \seq^{n-1}$, we have:
\begin{align*}
\pi_D \circ g'(\delta)  &= \pi_D \circ h' \circ g(\delta) \\
                        &= h \circ g(\delta) \\
                        &= f([\delta]) \\
                        &= g(\delta)
\end{align*}
based on the assumption that $h([])$ is defined and therefore has to be the unit of $\odot$. It is easy to show (by induction and definition) that for all $\delta_1, \dots, \delta_m \in \seq^n-1$ where $\sigma = [\delta_1] \ccat \dots \ccat [\delta_m]$, we have:

\begin{align}
f'(\sigma) &= h'\left([g(\delta_1)] \ccat \dots\relax \ccat [g(\delta_m)]\right) \label{aeq1}
\end{align}

Let $\sigma, \sigma' \in \seq^n$ and $\sigma' = [\delta_1] \ccat \dots \ccat [\delta_m]$ where $\delta_i \in \seq^{n-1}$. We have:
\begin{align*}
f'(\sigma \ccat \sigma') &= \left( \cdots  \left( f'(\sigma) \odot g'(\delta_1) \right) \odot \cdots \right) \odot g'(\delta_m)  \\
&= f'(\sigma) \odot \left(g'(\delta_1) \odot \cdots  \odot g'(\delta_m)\right) \\
&= f'(\sigma) \odot f'(\sigma')
\end{align*}
by associativity of $\odot$. Therefore, $f'$ is also $\odot$ homomorphic.  $f'$ is a lifting of $f$ since:
\begin{align*}
\pi_D \circ f'([]) &= \pi_D \circ h'([]) = h([]) = f([]) \\
\pi_D \circ f'(\sigma)  &= \pi_D \circ h'([g(\delta_1)] \ccat \dots \ccat [g(\delta_m)]) \\
                        &=  h([g(\delta_1)] \ccat \dots \ccat [g(\delta_m)]) \\
                        &=  f(\sigma)
\end{align*}

It remains to show that $h'$ is a summary of $f'$ up to projection. Let $\bar{h}: \seq_{D\times D'} \to D\times D'$ be the summary of $f'$ defined through $\odot$ that is
\begin{align*}
\bar{h}([]) &= f'([])\\
\forall y \in \seq_{D \times D'}, b \in D \times D':\ \bar{h}(y \ccat [b]) &= \bar{h}(y) \odot b
\end{align*}

Observe that $h'$ and $\bar{h}$ have the same range, but the sequences in the domain of $\bar{h}$ have strictly more information in each element of the sequence than those in $h'$. The claim that we want to prove is that
\[h' = \bar{h} \circ \bar{\pi}_D\]
where $\bar{\pi}_D$ is the natural extension of the projection function from elements to sequence of elements.

We have $h'([]) = f'([]) = \bar{h}([])$, by definition. This serves as our induction base case. Let $y \in \seq_{D \times D'}$, and assume that  $h' \circ \bar{\pi}_D(y) = \bar{h}(y)$. Let $a \in D \times D'$ and $a = g'(\delta)$ for some $\delta \in \seq^{n-1}$:
\begin{align*}
h' \circ \bar{\pi}_D(y \ccat [a]) &= h'(\bar{\pi}_D(y) \ccat \pi_D(a))\\
				&= h'\circ \bar{\pi}_D(y) \odot h'(\pi_D(a)) \\
				&= h'\circ \bar{\pi}_D(y) \odot h'(\pi_D(g'(\delta))) \\
				&= h'\circ \bar{\pi}_D(y) \odot h'(g(\delta)) \\
				&= h'\circ \bar{\pi}_D(y) \odot g'(\delta) \\
				&= \bar{h}(y) \odot a \\
				&= \bar{h}(y \ccat [a])
\end{align*}

\end{proof}

\begin{proposition}
Let  $f: \seq^n \to D$ be a (rightward) memoryless function defined by helper function $g$ (from Definition \ref{def:memoryless}), and let $h: \seq_{D} \to D$ be its summarized version. If $f$ can be lifted to a homomorphism
$f': \seq_{D} \to D \times D'$, then  there exists a lifting $h': \seq^n \to D \times D'$ of $h$ that is a homomorphism, and $h'$ is equivalent to the summary of $f'$ up to projection of its input sequence down to domain $D$.
\end{proposition}

\begin{proof}
Assume there exists a lifting $f'$ of $f$ which is $\odot$ homomorphic. Note that:
\begin{align*}
f'(\sigma \ccat \delta) &= f'(\sigma) \odot f'([\delta])
\end{align*}
Let $g': \seq^{n-1} \to D \times D'$ be defined as $g'(\delta) = f'([\delta])$.

Define $h'$ as:
\begin{align*}
h'([]) &= f'([]) \\
\forall x \in \seq_{D}, a \in D (\mbox{ s.t. }a = g(\delta)): h'(x \ccat a) &= h'(x) \odot g'(\delta)
\end{align*}
Let us argue that $h'$ is $\odot$-homomorphic and a lifting of $h$. The former is immediately implied by associativity of $\odot$. For the latter, we need to show that $\pi_D \circ h' = h$ (rightward computability of $h'$ is implied by the computability of $\odot$). Observe that:
\begin{align*}
\pi_D \circ h'([]) &= \pi_D \circ f'([]) = f([]) = h([]) \\
\end{align*}
Let $x \in \seq_D$ and $x = d_1 \ccat \dots d_m$ where $d_i = g(\delta_i)$ . Then:
\begin{align*}
\pi_D \circ h'(x) &= \pi_D(h'(x)) \\
                  &= \pi_D(f'([\delta_1] \dots \dots \ccat [\delta_m])) \\
                  &= f([\delta_1] \ccat \dots \ccat [\delta_m]) &&\text{($f'$ is a lifting of $f$)}\\
                  &= h([g(\delta_1)] \ccat \dots \ccat [g(\delta_m)]) &&\text{(by equation \ref{aeq1})} \\
                  &= h(x)
\end{align*}
Finally, it remains to show that $h'$ is a summarized version of $f'$ up to projection. Let $\bar{h}: \seq_{D\times D'} \to D\times D'$ be the summary of $f'$ defined through $\odot$ that is
\begin{align*}
\bar{h}([]) &= f'([])\\
\forall y \in \seq_{D \times D'}, b \in D \times D':\ \bar{h}(y \ccat [b]) &= \bar{h}(y) \odot b
\end{align*}

The claim that we want to prove is that
\[h' = \bar{h} \circ \bar{\pi}_D\]
where $\bar{\pi}_D$ is the natural extension of the projection function from elements to sequence of elements. The argument is identical to the one made at the end of the proof of Proposition \ref{aprop:lift} to prove the same claim.

\end{proof}

%%% Local Variables:
%%% mode: latex
%%% TeX-master: "../paper"
%%% End:

%% file: proofs/lift-triv.tex
\begin{proof}
Since $f$ is rightward, there exists a binary operator $\otimes$ and family of functions $\mathbbm{G}$ such that for all $\sigma \in \seq^n$ and $\delta \in \seq^{n-1}$:
\[f(\sigma \ccat \delta ) = f(\sigma) \otimes \mathbbm{G}(f(\sigma))(\delta)\]
Note that the signature of the lifted function $f \times \iota'$ is $\seq^{n} \to D \times \seq^{n-1}$.
Let  $\oplus : (D \times \seq^{n-1}) \times (D \times \seq^{n-1}) \to D \times \seq^{n-1}$ be defined as:
\[\forall a,b \in D, \delta,\delta' \in \seq^{n-1}:  (a, \delta) \oplus (b,\delta') = (a \otimes \mathbbm{G}(a)(\delta'), \delta')\]

Let $\iota'': \seq^{n-1} \to \seq^{n}$ be defined as $\forall \delta \in \seq^{n-1}: \iota''(\delta) = [\delta]$ and let $g: \seq^{n-1} \to D \times \seq^{n-1}$ to be function that on all inputs $\delta$ returns $(\zero, \delta)$ for some constant value $\zero \in D$.

It is straightforward to see that $f \times \iota'$ is memoryless with the loop join operator $\oplus$ and helper function $g$ since:

%\[
%\forall \sigma  \in \seq^n, \delta \in \seq^{n-1}:\ f \times \iota'(\sigma \ccat [\delta]) = f \times \iota'(\sigma) \circledcirc g \times \iota''(\delta, \zero) \]

\begin{align*}
f \times \iota'(\sigma \ccat [\delta]) & = (f(\sigma \ccat [\delta]), \iota'(\sigma \ccat [\delta])) \\
                        &= (f(\sigma) \otimes \mathbbm{G}(f(\sigma))(\delta), \delta) \\
                        &= (f(\sigma), \iota'(\sigma)) \oplus (\zero, \delta) \\
                        &= f \times \iota'(\sigma) \oplus g(\delta)
\end{align*}
Therefore, by definition \ref{def:memoryless}, we can conclude that $f \times \iota'$ is memoryless.
\end{proof}

\paragraph{Complexity preservation of the trivial memoryless lift.}
It is easy to intuitively see why a trivial lift like the above does not increase the time complexity of computation of $f$. To argue for this, it is easier to think about the loops (instead of functions). Imagine the original function corresponds to the loop:

\lstset{
  language=C,
  classoffset=0,
  basicstyle=\ttfamily\bfseries,
  keywordstyle=\color{blue}\ttfamily,
  stringstyle=\color{red}\ttfamily,
  commentstyle=\color{gray}\ttfamily,
  morecomment=[l][\color{magenta}]{\#},
  emph={int},
  emphstyle=\color{ForestGreen},
  classoffset=1,
  morekeywords={seq,int,bool},
  keywordstyle=\color{ForestGreen},
  classoffset=2,
  morekeywords={lemma,calc,assert,ensures,max,min},
  keywordstyle=\color{RedViolet},
  classoffset=3,
  morekeywords={HomomorphismSum,max_rec},
  keywordstyle=\color{WildStrawberry},
  frame=single
}

{\footnotesize
  \begin{lstlisting}
for(int i = 0; i < n; i++) {
  for(int j = 0; j < m; j++) {
    ...
  }
}
  \end{lstlisting}
}

which has complexity $O(nm)$. Then the lifted one would correspond to the loop:

{\footnotesize
  \begin{lstlisting}
for(int i = 0; i < n; i++) {
  for(int j = 0; j < m; j++) {
    ...
  }
  ...
  for(int j = 0; j < m; j++) {
    ...
  }
}
  \end{lstlisting}
}
where the second copy of the inner loop effectively redoes the computation of the inner loop. This still has the complexity $O(nm)$ albeit with larger constants.

%%% Local Variables:
%%% mode: latex
%%% TeX-master: "../paper"
%%% End:

%% file: proofs/mlift.tex
\begin{proof}
  Consider a multidimensional rightward function $f: \seq^n \to D$ that is not memoryless, defined by a family of functions $\mathbbm{G}: D \to (\seq^{n-1} \to D)$ as in Definition \ref{def:mrw}.
  The function is effectively defined using the recursive equation $f(\sigma \ccat [\delta]) =  f(\sigma) \otimes \mathbbm{G}(f(\sigma))(d)$.

Let us show that lifting $\mathbbm{G}$ to a family of homomorphisms (Definition \ref{def:hom-family}) is sufficient to lift $f$ to a memoryless function. For any $d \in D$, $\mathbbm{G}(d)$ is defined by:
\begin{align}
  \mathbbm{G}(d)([]) &= d \nonumber\\
  \mathbbm{G}(d)(\delta \ccat [\gamma]) &= \mathbbm{G}(d)(\delta) \ominus \gamma \nonumber
\end{align}

Imagine that we lift $g_\zero = \mathbbm{G}(\zero)$ for some $\zero \in D$ to a homomorphism. We will have a $\lift{g_{\zero}}{D'}$ such that there exists a $\boxdot$ operator that satisfies for all $\delta, \delta' \in S^{n-1}$:

\[ \lift{g_\zero}{D'}(\delta \ccat \delta') =  \lift{g_{\zero}}{D'}(\delta) \boxdot \lift{g_{\zero}}{D'}(\delta') \]

% The element $\zero' \in D'$ defined by $\zero' = \lift{g_{\zero}}{D'}([])$ (we have $\zero = \pi_D \left( \zero' \right)$) is the is the identity of $\boxdot$.
We define the lifting of $\mathbbm{G}(d)$ by using the homomorphic lifting of $g_\zero$:

\[ \lift{\mathbbm{G}(d)}{D'}(\delta) =  d' \boxdot \lift{g_\zero}{D'}(\delta) \]

where $\pi_D (d') = d$ and $\pi_{D'}(d') = \pi_{D'}(\lift{g_{\zero}}{D'}([]))$.

$\lift{\mathbbm{G}(d)}{D'}$ is naturally a homomorphism since $\lift{g_\zero}{D'}$ is one.
We can verify that it is a lifting of $\mathbbm{G}(d)$ by projecting to $D$. We use $w$ the weak inverse of $\lift{g_\zero}{D'}$ defined by $\forall d' \in D', \lift{g_\zero}{D'}(\delta) \circ w \ ( d') = d'$.\\
\begin{align}
  \pi_D \circ \lift{\mathbbm{G}(d)}{D'} (\delta \ccat [a]) &= \pi_D (d' \boxdot  \lift{g_\zero}{D'}(\delta)) \nonumber\\
                                                &= \pi_D (\lift{g_\zero}{D'}(w(d') \ccat \delta)) \nonumber\\
                                                &= g_\zero (w(d') \ccat \delta) \nonumber\\
                                                &= (\ldots (g_\zero (w(d')))
                                                  \ominus \delta_1 \ldots \ominus \delta_n) \nonumber\\
                                                &= (\ldots (d
                                                  \ominus \delta_1) \ldots \ominus \delta_n) \nonumber\\
                                                &= \mathbbm{G}(d)(\delta) \nonumber
\end{align}

Since $\lift{\mathbbm{G}(d)}{D'}$ is a lifting of $\mathbbm{G}(d)$ we can use its projection on $D$ to redefine $f$:

\[ f(\sigma \ccat [\delta]) = \vec{s} \otimes (\pi_D \circ (\lift{\mathbbm{G}(d)}{D'}([]) \boxdot \lift{g_\zero}{D'}(\delta'))) \]

$f$ can be lifted to a memoryless function, explicitly by defining a lifted operator $\lift{\otimes}{D'} : D' \times D' \rightarrow D$ such that for any $\vec{t}, \vec{t}' \in D'$, with $\vec{s} = \pi_D \circ \vec{t}$ and $\vec{s}' = \pi_D \circ \vec{t}'$
\begin{align}
  \pi_D \circ (\vec{t} \ \lift{\otimes}{D'} \ \vec{t}') &= \vec{s} \otimes (\pi_D \circ (\lift{\mathbbm{G}(d)}{D'}([]) \boxdot \vec{t}')) \nonumber \\
  \pi_{D'} \circ (\vec{t} \ \lift{\otimes}{D'} \ \vec{t}') &= \pi_{D'} \circ \vec{t}' \nonumber
\end{align}

The lifted function $\lift{f}{D'}$ is defined by:
\[ \lift{f}{D'}(\sigma \ccat [\delta]) = \lift{f}{D'}(\sigma) \ \lift{\otimes}{D'} \ \lift{g_\zero}{D'}(\delta) \]
which matches the definition of a memoryless function. Remark that it is a valid lifting of $f$ since $\pi_D \circ \lift{f}{D'} = f$ by construction of the lifted operator.
\end{proof}

%%% Local Variables:
%%% mode: latex
%%% TeX-master: "../paper"
%%% End:

%% file: parallelizability.tex
% !TEX root =  paper.tex

\section{Algorithmic Parallelization} \label{sec:par}
In Sections \ref{sec:mfunction} and \ref{sec:lift}, we presented the theoretical foundations of our approach. Theorem \ref{thm:par-cond} guarantees that it is {\em sound and complete} to parallelize the summarized loop in place of the original loop nest. Proposition \ref{prop:lift-triv} guarantees that any loop nest can be transformed into one that is summarizable. Finally, Theorem \ref{thm:convl} guarantees that a summarized loop can be {\em soundly and completely} lifted to a homomorphism in place of the original loop. In this section, we outline our algorithmic approach to parallelization.

\subsection{Efficient Divide-and-Conquer Solution}\label{sec:eff-hom}
Consider a loop nest $L$ of depth $n$ where the number of iterations of every loop is bounded by a parameter $m$. Assuming no function calls are made, the loop nest has a time complexity of $O(m^n)$. Since the translation to functional form preserves time complexity, this is also the time complexity of the function $h_L: \seq^n \to D$ corresponding to the loop nest. For a parallel implementation of $h_L$ based on a join operator $\odot$ to have reasonable speedups over constantly many processors, the (sequential) complexity of the implementation based on the join should not be higher than that of the original code. Constantly many processors cannot compensate for a variable increase in complexity.

\begin{proposition}\label{prop:tcomp}
Let $h_L \in O(m^n)$ be $\odot$-homomorphic. The sequential implementation of $h_L$ based on $\odot$ is in $O(m^n)$ if $\odot \in O(m^{n-1})$.
\end{proposition}

\begin{proof}
It is straightforward to see that for a rightward function $f: \seq^n \to D$ with time complexity $O(m^n)$ defined through the recursive equation $f(\sigma \ccat [\delta]) =  f(\sigma) \otimes \mathbbm{G}(f(\sigma, d))(\delta)$
we have:
\begin{itemize}
\item Every $g \in \mathbbm{G}$ is a  (leftward or) rightward function of complexity $O(m^{n-1})$.
\item Any $d \in D$ is strictly of space complexity $O(m^{n-1})$.\footnote{This is under the assumption that the data is fully read. So, this excludes, for example, operations on lists performed through reference manipulation without reading the entire list content.}
\item $\otimes$ is computable in time $O(m^{n-1})$.
\end{itemize}
Since if any of the upper bounds are violated, then one can show that the time complexity of $f$ would surpass $O(m^n)$. Now, if $h_L$ is a homomorphism and we want the parallel computation based on the homomorphism's {\em join} operator $\odot$ to have the same complexity as $h_L$. $\odot$ replaces $\otimes$, and we can conclude $\odot \in O(m^{n-1})$.

\end{proof}

This observation leads to a formal definition of parallelizability.

\begin{definition}(Parallelizability)\label{def:par}
A rightward (respectively leftward) function $h_L \in O(m^n)$ is efficiently parallelizable if and only if it is $\odot$-homomorphic and $\odot \in O(m^{n-1})$.
\end{definition}

The deduced upper bound on $\odot$ is crucial to justify the algorithmic choices made in Sections \ref{sec:join}, where the time complexity budget for {\em join} informs the choices of {\em syntax} for syntax-guided synthesis \cite{Alur15}.
Similarly, there are time and space complexity budgets for an efficient {\em lifting}.

\begin{corollary}\label{cor:ltcomp}
  If a function $h_L \in O(m^n)$ is lifted to $\lift{h_L}{D'} \in O(m^n)$, then any $d' \in D'$ has space complexity $O(m^{n-1})$.
\end{corollary}
The proof follows directly from that of Proposition \ref{prop:tcomp}, which also imposes the time complexity of $O(m^{n-1})$ for computing $d'$. The time and space complexity bounds for $d'$ inform the syntactic form of the auxiliary accumulators and the computation that produces them.  In Section \ref{sec:incompleteness}, we provide a variation of the example from Section \ref{sec:mtlr}, and a proof that any lifting of that function to a homomorphism has a space complexity beyond the budget specified in Corollary \ref{cor:ltcomp}. This information-theoretic proof is very involved, but makes the important point that an efficient lifting may not always exist, and consequently neither does a {\em complete} lifting algorithm.

\subsection{Incompleteness}\label{sec:incompleteness}

\intextsep=0pt
\columnsep=10pt

\begin{figure*}[ht]
  \begin{center}
    $A = \left(
      \begin{array}{c c c}
        f([1:1]) & -L & -L \\
        -L - f([1:1]) & L + f([2:2]) & 0 \\
        0 & -L -f([2:2]) & L + f([3:3]) \\
        \hline
        L + f([1:2])/2 & L + f([1:2])/2 & -L - f([3:3])\\
        -L - f([1:2])/2 & f([2:3])/2 - f([1:2])/2 & L + f([2:3])/2 \\
        \hline
        L + f([1:3])/3 & f([1:3])/3 - f([2:3])/2 & f([1:3])/3 - f([2:3])/2
      \end{array}
    \right).$
    \caption{Definition of Matrix A}
    \label{fig:matrix}
  \end{center}
\end{figure*}

Consider a two-dimensional array of integers (with both
\begin{wrapfigure}{h}{0cm}
\includegraphics[scale=0.30]{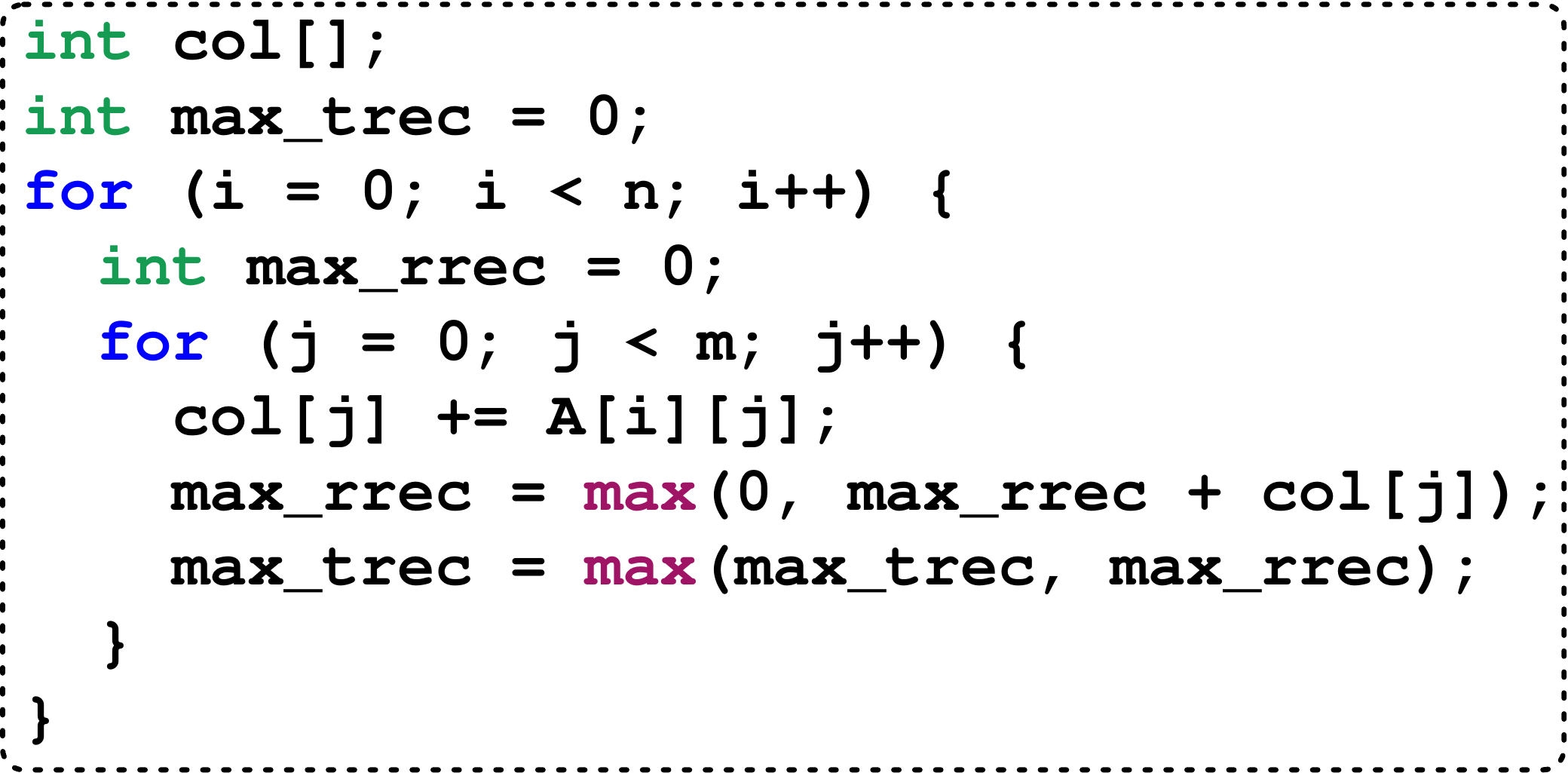}\vspace{-7pt}
\caption{\small Maximum top subarray sum.\label{fig:mtr}}
\end{wrapfigure}
 positive and negative) elements. Assume that the goal is to compute the maximum sum of the elements of a subarray $A[0..\ell, j..k]$ for all $0 \le \ell <n$ and $0 \le j \le k < n$, i.e.  all subarrays that start from the top row of the original array, but can include a subset of its rows and columns.
The code in Figure \ref{fig:mtr} is a clever single-pass implementation of this function. We will prove that this code, although very similar syntactically to the one in Figure \ref{fig:mtlr}, does not admit an efficient divide-and-conquer parallelization.

\def\Jn{\mathbb{I}_{n}}
\def\Jnm#1{\mathbb{I}_{n}^{#1}}
\def\Jnr{\Jnm{r}}

Let $\mathbb{R}^{+}$ be the set of positive real numbers. For integers $i \leq j$, let $[i:j]$ denote the set $\{i, \dots, j\}$, and for each integer $n$ and $r \leq n$ define \[\Jnm{r} := \left\{[i: j]: 1 \leq i \leq j \leq n\ |\ j < i + r\right\}.\]
We write $\Jn$ for $\Jnm{n}$. Note that $|\Jnm{r}| = O(nr)$ and in particular $|\Jn| = O(n^{2})$.

We call a function $f:\Jnr\rightarrow \mathbb{R}^+$ {\em graded} if for any $J, J' \in \Jnr$ with $|J| < |J'|$:
\[f(J) / |J| > f(J') / |J'|.\]

\begin{lemma}
\label{grading}
Given an arbitrary function $f: \Jnr \rightarrow \mathbb{R}^{+}$ and  $X \geq \max_{J, J' \in \Jnr} |f(J) - f(J')|$, the function $\tilde{f}$ defined as
\begin{equation}
\tilde{f}(J) = |J|(n - |J|) X + |J| f(J),
\end{equation}
is graded.
\end{lemma}

\begin{proof}
If $|J| < |J'|$:
\[ \frac{\tilde{f}(J)}{|J|} - \frac{\tilde{f}(J')}{|J'|} = (|J'| - |J|) X + f(J) - f(J') \geq X + f(J) - f(J') > 0.\]
\end{proof}

Given a matrix $A = [a_{ij}]$ with $n$ columns, the {\em column-interval maximum prefix sum} of $A$ is a real-valued function $\mu_{A}$ defined over $\Jn$, as
\[\mu_{A}(J) := \max_{k \geq 1}\sum_{i = 1}^{k}\sum_{j \in J}a_{ij},\]
and the {\em maximum top subarray sum}  of $A$ is
\[\mu(A) := \max_{J \in \Jn}\mu_{A}(J).\]

\begin{figure*}
\begin{center}
\includegraphics[scale=0.21]{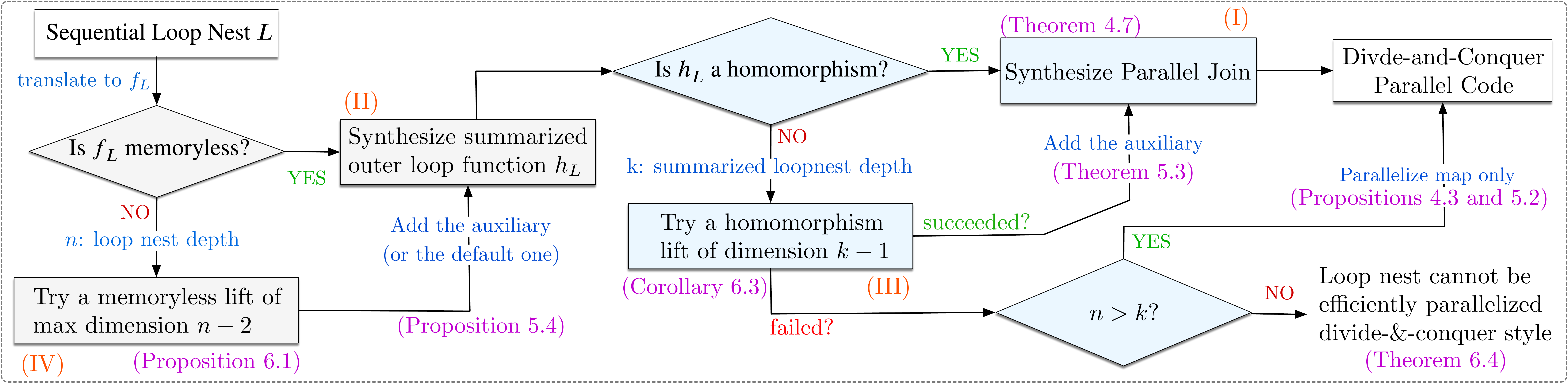}
\caption{Parallelization Schema\label{fig:schema}}
\end{center}
\end{figure*}

\begin{lemma}
\label{reduction}
For any graded function $f:\Jnr \rightarrow \mathbb{R}^{+}$, there is an $n$-column matrix $A$ with $O(nr)$ rows for which $\mu_{A} = f$.
\end{lemma}

\begin{proof}

Choose $L > \max_{J\in\Jnr}f(J)$. We construct the matrix $A$ by stacking $r$ {\em groups} of rows, where the $i$th group is a set of $n - i + 1$ rows, each of which corresponding to one of the distinct sets $J \in \Jnr$ with $|J| = i$, thus the total number of rows adding up to $\sum_{i = 1}^{r} (n - i + 1) = O(nr)$.

Let $k_J$ denote the row that corresponds to the set $J \in \Jnr$. The entries in row $k_J$ are determined as follows: for any column $j$, let $s_{j} = \sum_{i = 1}^{k_{J} - 1}a_{ij}$. Then we set
\[a_{ij} = \left\{\begin{array}{lll}f(J)/|J| - s_{j} & \qquad & j \in J\\
-L - s_{j} &&\text{otherwise.}\end{array}\right.\]
For example, for $n = 3$ and $f$ defined on $\Jn$ the matrix $A$ would be as illustrated in Figure \ref{fig:matrix}.

The matrix for $\Jnr$ would simply keep the first $r$ groups of the one for $\Jn$.

We claim that for any $J \in \Jnr$, $\mu_{A}(J) = f(J)$. For any row number $k$ define for a column $j$,
\[s_j(k) := \sum_{i = 1}^{k} a_{ij}, \]
and for an $J \in \Jn$ define
\[s_{J}(k) := \sum_{j \in J}s_{j}(k). \]

It can be readily confirmed from our construction that for any $J \in \Jnr$, and any column $j$:
\[s_{j}(k_{J}) = \left\{\begin{array}{lll}
f(J) / |J| &\qquad& j \in J\\
-L &&\text{otherwise.}
\end{array}\right.\]

In particular, we have $s_{J}(k_{J}) = f(J)$, immediately implying that $\mu_{A}(J) \geq f(J)$. To prove $\mu_{A}(J) = f(J)$, we show next that for any $k \neq k_{J}$, $s_J(k) < f(J)$.
By our construction, if $k = k_{J'} \neq k_{J}$, for some $J' \in \Jnr$ then
\begin{equation}
s_{J}(k) = \sum_{j \in J'\cap J}f(J')/|J'| - \sum_{j \in J \setminus J'}L  =  |J' \cap J|f(J')/|J'| - L|J \setminus J'|.
\end{equation}
If $k < k_{J}$, then
\[
s_k(J) = |J' \cap J|f(J')/|J'| - L|J \setminus J'| \leq f(J') - L|J \setminus J'| < 0,
\]
where the last inequality follows because $J' \neq J$ and $f(J') < L$. In fact, the above inequality holds true by the same reasoning, even if $k > k_{J}$ but $J \setminus J' \neq \emptyset$. To complete the proof, suppose $k > k_{J}$ and $J \subset J'$, then
\[
s_{k}(J) = |J|f(J')/|J'| < f(J),
\]
where this final inequality follows from the fact that $f$ is graded.
\end{proof}

The construction in the proof of the above Lemma can be regarded as an {\em encoding}: any function $f:\Jnr\rightarrow \mathbb{R}^{+}$, can be encoded by a tuple $(A, X)$, where $A$ is the constructed $O(nr) \times n$ matrix and $X$ is strict upper bound on the absolute value of the difference between values of $f$ from the proof of the Lemma \ref{grading} needed for turning $f$ into the graded $\tilde{f}$. The decoding for input $J \in \Jnr$ is done by as:
\[f(J) = \mu_{A}(J) / |J| - (n - |J|) X. \]

Let us now restrict the range of values of $f$ to the set of positive integers $\{1, \dots, 2^{k}\}$ representable by $O(k)$ bits. Thus $X = 2^{k} + 1$ would be a valid upper bound to form values $\tilde{f}$ which would then be $O(k)$ bit integers themselves. Observe also that by construction $\tilde{f}(J)$ is always divisible by $|J|$. Following the arithmetics in Lemma \ref{reduction}, it can be verified that the entries of matrix $A$ will all be $O(k)$-bit integers. We can thus state the following.

\begin{lemma}
\label{compression}
Let $\mathcal{C}$ be an algorithm that given an $m \times n$ matrix $A$, with $O(k)$-bit integer entries and with $m \leq n(n+1)/2$, produces a data structure $\mathcal{C}(A)$ that can be used (independently of $A$) to evaluate $\mu_{A}(J)$ for every $J \in \Jn$, then $\mathcal{C}(A)$ needs $\Omega(km)$ bits of space.
\end{lemma}

\begin{proof}
Let $r = O(m / n)$. From the above discussion, we can encode any mapping from $\Jnr$ to $k$-bit integers using an $m \times n$ matrix $A$ of $O(k)$-bit integers (plus $O(k)$ bits to represent $X$ from Lemma \ref{grading}). Since there are $2^{O(knr)}$ such mappings, $\mathcal{C}(A)$ must use at least the logarithm of that many bits to be able to distinguish different functions from each other, meaning it must have size $\Omega(knr) = \Omega(km)$.
\end{proof}

\begin{theorem}
For any divide and conquer algorithm for computing $\mu(\cdot)$ of $n$-column matrices of $O(k)$-bit integers, the output of a sub-problem of $r \leq n(n+1)/2$ rows has size $\Omega(kr)$. In particular, solutions to subproblems of size $O(n^{2})$ require $\Omega(kn^{2})$ bits.
\end{theorem}

\begin{proof}
Let $A_{0}$ and $A_{1}$ be two consecutive subproblems with $A_{1}$ consisting of $r$ rows. Let $A_{01}$ represent the concatenation of $A_{0}$ and $A_{1}$. We show that setting a single row of $A_{0}$ adversarially is enough to force the join of $A_{0}$ and $A_{1}$ to compute $\mu_{J}(A_{1})$ for any $J \in \Jnr$. Lemma \ref{compression} then implies that the output computed for $A_{1}$ must have size $\Omega(kr)$. Let $[i:j] \in \Jnr$ and let $L > \max_{J \in \Jnr}\mu_{A_{1}}(J)$. We set the entries one row of $A_{0}$ to  $L$ for columns in $[i:j]$ and $-L$ for the remaining columns. All other rows of $A_{0}$ are set to zero. Since $\mu_{A_{1}}(J) \ll L$ for all $J$. $\mu(A)$ has to use as many $L$ entries in our set row and no $-L$ entries to be maximized. Therefore, $\mu(A_{01}) = |J|L + \mu_{A_{1}}(J)$.
\end{proof}

Therefore, we can conclude with the following theorem:

\begin{theorem}\label{thm:hard}
An efficient lifting of a (multidimensional) rightward function may not always exist.
\end{theorem}

\subsection{Parallelization Schema}

The diagram in Figure \ref{fig:schema} illustrates the algorithmic steps in our methodology to parallelize an input sequential program. The light grey section performs the summarization of the loop, which corresponds to the discovery of a {\em map}. The light blue section parallelizes the summarized loop, which corresponds to the discovery of a {\em reduction}. The key algorithmic steps are the synthesis of the summarized outer loop and the parallel join (respectively labeled as (II) and (I) in Figure \ref{fig:schema}), which are solved using syntax-guided synthesis (Section \ref{sec:join}), and the memoryless lift and the homomorphism lift (boxes respectively labeled as (IV) and (III) in Figure \ref{fig:schema}) which are performed using a deductive-style algorithm (Section \ref{sec:alg}). Each step of the process is labeled with the theorem justifying it.

The test for a homomorphism in the schema is only nominal and implemented practically through the success or failure of join synthesis algorithm. It is important to note that Rice's theorem dictates that deciding whether a computable function is a homomorphism in general is undecidable, and therefore there exists no decision procedure for this test.

If a function is not memoryless, then in (IV), an efficient {\em memoryless lift} is attempted; that is, the most efficient lift whose time complexity is no more than $O(m^{n-2})$. If this fails, we know we can always rely on the default admissible memoryless lift (which is incidentally of complexity $O(m^{n-1})$). If a homomorphism lift within the complexity budget (determined by $k$, the summarized loop depth) exists, then a classic divide-and-conquer parallel program is produced. Otherwise, we opt to parallelize the inner loop through the map and leave the outer loop's computation as sequential (as is the case for the example in Section \ref{sec:wbp}). When summarization does not reduce the depth of the loop (i.e $n = k$), then there is no benefit in parallelizing the inner loop nest through a parallel {\em map}; i.e. the parallelization has failed.

%%% Local Variables:
%%% mode: latex
%%% TeX-master: "paper"
%%% End:

%% file: join.tex
% !TEX root =  paper.tex
\section{Join Synthesis} \label{sec:join}
In this section, we address the algorithmic problems of generation of the {\em parallel join} operator and the {\em summarized outer loop}, respectively steps (I) and (II) from the general schema of Figure \ref{fig:schema}. Although the two problems seem independent, the latter turns out to be a special instance of the former.

\subsection{Syntax-guided Synthesis of Parallel Join}

We employ syntax-guided synthesis (SyGuS)\cite{Alur15} to generate the parallel join. Given a \emph{correctness specification} $\phi$ and \emph{syntactic constraints} describing the syntactic space $\mathbbm{S}$ of possible implementations for join, a syntax-guided synthesis solver finds a solution $x \in \mathbbm{S}$ where $\phi(x)$ holds. The correctness specification for the join operator $\odot$  is that $h_L$ (the summarized loop function) forms a homomorphism with $\odot$ (i.e. Definition \ref{def:hom}); i.e., $\phi(\odot) \equiv \forall x, y \in S_D, h_L(x \ccat y) = h_L(x) \odot h_L(y)$.

The main challenge in SyGuS is to appropriately define the syntactic restrictions. On one hand, they need to be expressive enough to include an efficient $\odot$ that satisfies $\phi$ if one exists. On the other hand, the smaller the state space $\mathbbm{S}$, the more tractable the search problem for its synthesis.

We use an insight to define $\mathbbm{S}$ effectively. A function $f^{-1}$ is a weak inverse of a function $f$ iff $f \circ f^{-1} \circ f = f$, and $f$ always has at least one weak inverse iff $f$ is computable and its domain is countable. All the functions of interest in this paper have weak inverses of signature $D \rightarrow \seq^n$. In the proof of the third homomorphism theorem in \cite{Gibbons96}, it is observed that a {\em join} $\odot$ for a homomorphism $h_L$ can be constructively defined based on $h_L$'s weak inverse. That is, for all $d,d' \in D$ we have $d \odot d' = h_L(h_L^{-1}(d) \ccat h_L^{-1}(d'))$. This implies an $\odot$ with a similar syntactic structure to $h_L$ exists. Moreover,  Proposition \ref{prop:tcomp} implies that for $\odot$ to remain within the complexity budget,  $h_L^{-1}(d)$ and $h_L^{-1}(d')$ need to have {\em constant} length.

\begin{example}\label{ex:7-1}
  Recall the maximum top-left subarray sum example from Figure \ref{fig:mtlr}(c). The summarized function $h_{mtls}$'s signature
  is the tuple of state variables $\langle${\tt \small rec,max\_rec,row\_mrec,} {\tt \small mtl\_rec}$\rangle$ and its weak inverse is a 2-row array with the same width as the original input. It is illustrated below.
  \parpic[rb]{\includegraphics[scale=0.2]{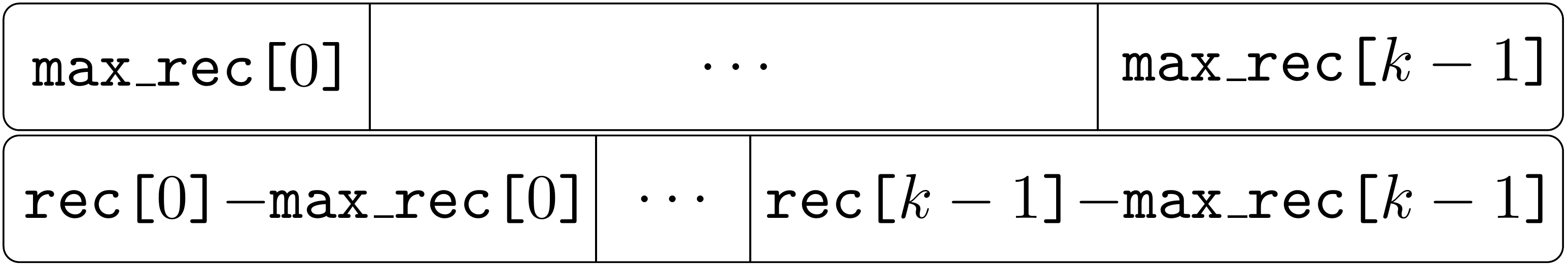}}
  A join constructed based on this weak inverse executes $h_{mtls}$ on 4 rows; the concatenation of 2 sets of 2 rows coming from the left and the right threads.
\end{example}

In syntax-guided synthesis, $\mathbbm{S}$ is defined by a {\em sketch} (a program with unknowns) and an {\em expression grammar} $G$ specifying possible completions for the holes (unknowns). Intuitively, the solver searches for substitutions from expressions in $G$ for all holes in the sketch, such that the resulting program satisfies the correctness specification. The construction of the sketch we use is an extension of the one in \cite{PLDI17}. It needs to be extended, since this paper introduces a technique to synthesize superscalar joins, and in \cite{PLDI17} only constant-time computable joins are considered.

\paragraph{Sketch}
 To obtain the sketch, a compilation function $\cf$ (presented in Figure \ref{fig:compilation_full}) is applied to the system of recurrence equations representing the body of the summarized loop. The result is a system of equations where the right-hand side of the equations become expressions with holes. We define $\cf$ on expressions first, and then extend it to a systems of equations.

In Figure \ref{fig:compilation_old}, we recall the compilation function $\cf$ of \cite{PLDI17}. It transforms expressions of the input loop body into expressions of the sketch by replacing variables with holes. Recall that the join takes as input the computation results of two threads: we will refer to them as the \emph{left} thread and the \emph{right} thread. In order to reduce the size of the state space of solutions, we identify two types of holes: (1) right holes $\rhole$, which can be completed by expressions using only variables from the right thread and (2) left holes $\lrhole$, which can completed using variables from both threads.  The compilation function $\cf$ is defined recursively on expressions $e$ of the input language. $op$ is an operator,  $x$ is a (possibly subscripted) variable, and $c$ is a constant.

Since the join will use recursion (or, equivalently iteration with accumulation), we need to allow the use of recursion variables in the join. We extend the compilation function with $\cfrec$ and add a third type of hole: recursive hole $\lrrechole$ which can completed using variables from both threads and local variables defined in the join. $\cfrec$ is defined in Figure \ref{fig:compilation}: it coincides with $\cf$ on constants and expressions, and replaces state variables with recursive holes $\lrrechole$ instead of left-right holes $\lrhole$.

\begin{figure}[h]
 \vspace{5pt}
  \begin{subfigure}[t]{0.5\textwidth}
    \footnotesize
    \framebox{
      \begin{minipage}{8.2cm}
        \begin{center}
          \begin{tabular}{rcl}
            $\cf\left(k\right)$ & $=$ & $??_R$ \\
            $\cf\left(x \right)$ & $ = $ &
                                           $\left\{ \begin{matrix}
                                               \rhole & \IF \ x \in \IVar \\
                                               \lrhole  & \IF \ x \in \SVar
                                             \end{matrix} \right.
                                                          $ \\
            $\cf\left(op(e_1, .., e_m)\right)$&$ = $&
                                                      $op\left(\cf(e_1), ..., \cf(e_m)\right)$
          \end{tabular}
        \end{center}
      \end{minipage}
    }
    \caption{Sketch compilation for scalar joins from \cite{PLDI17}.}
    \label{fig:compilation_old}
  \end{subfigure}

  \begin{subfigure}[t]{0.5\textwidth}
    \footnotesize
    \framebox{
      \begin{minipage}{8.2cm}
        \begin{center}
          \begin{tabular}{rcl}
            $\cfrec\left(x \right)$ & $ = $ &
                                              $\left\{
                                              \begin{matrix}
                                                \rhole & \IF \ x \in \IVar \\
                                                \textcolor{blue}{\mathbf{\lrrechole}} & \IF \ x \in \SVar
                                              \end{matrix} \right.
                                                                                        $ \\
          \end{tabular}
        \end{center}
      \end{minipage}
    }
    \caption{Extension of the compilation function with $\cfrec$.}
    \label{fig:compilation}
  \end{subfigure}

  \begin{subfigure}[t]{0.5\textwidth}
    \footnotesize
    \framebox{
      \begin{minipage}{8.2cm}
      \[
        \cf(E) = \left(
          \begin{matrix}
            s_1  = \cf(Exp_1(\IVar, \SVar))\; \\
            \vdots \\
            (s_{i_1},s_{i_2}, \dots, s_{i_p}) = \loopfor{i \in \mathcal{I}}{ \: \cfrec(E_{loop})}\\
            \vdots \\
            s_q = \cfrec(Exp_q(\IVar, \SVar)) \\
          \end{matrix}
        \right)
      \]
    \end{minipage}
  }
    \caption{Compilation of the sketch body from the system of equations $E$.}\label{fig:compilation_system}
  \end{subfigure}
  \vspace{-5pt}
  \caption{Sketch compilation function $C$}\label{fig:compilation_full}
  \vspace{5pt}
\end{figure}

Finally, the compilation function is defined over a system of recurrence equations $E$ in Figure \ref{fig:compilation_system}.
$\cf(E)$ is the result of applying the compilation functions to the expressions on the right-hand side of the equations in the loop body.
% Recursion in the join (through a loop in the system of equations) is only useful when we can use recursion variables, so the $\cfrec$ compilation function is used when necessary.
$\cf$ is applied to the right-hand side of simple equations appearing before all loop equations. $\cfrec$ is used for the body of a loop equation ($\cfrec(E_{loop})$) and all the equations after. $\cfrec(E)$ is defined similarly in a recursive manner, but only $\cfrec$ is applied to the expressions of the right hand side of each equation in $E$.
Remark that for local variables to be used in the loop of the sketch, the local variables will first need to be initialized, and we will add equations $s_i = \lrhole$ for any variable $s_i$ that can be used in a loop and that has not been initialized before that loop.
To include the solutions described in Example \ref{ex:7-1}, we allow bounded repetitions of the sketch. To produce the exact join of this example, four would have been necessary. But, practically, in the vast majority of the cases one repetition of the sketch is sufficient. Additionally, we extend the state space of solutions represented by the sketch to include potential summarized solutions: any equation sketch $s_i = \cfrec(Exp)$ appearing after a loop is copied in the body of the preceding loop. The construction still ensures that the solution based on the weak inverse is in the space of possible solutions.

\parpic[rt]{\includegraphics[scale=0.33]{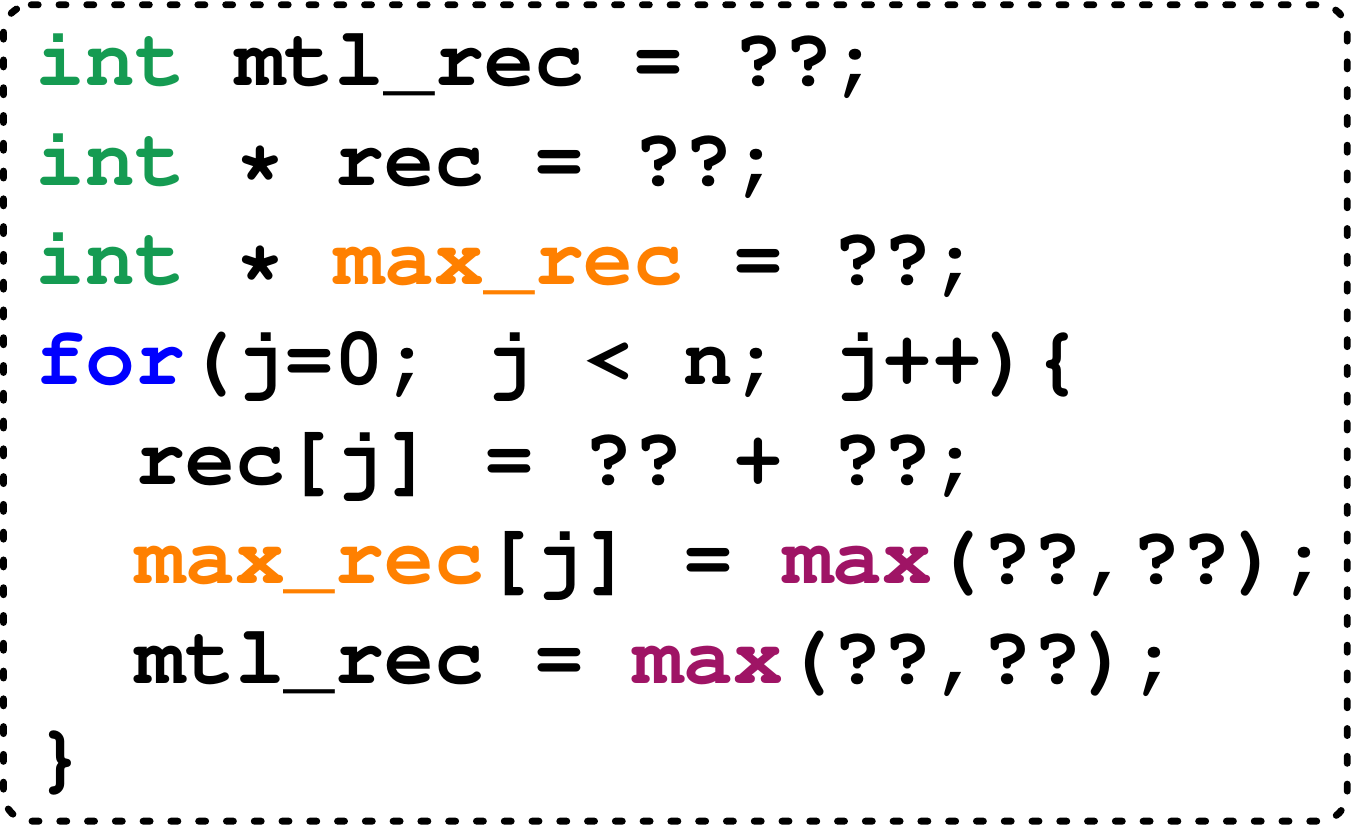}}
For example, consider the sketch illustrated on the right, written in the syntax of the input language for simplicity.
The crucial difference between a looped sketch like this one and those in \cite{PLDI17} is that in a loop, variables may have to be referenced on the right-hand side of the assignments to effectively implement {\em recursion}. Therefore, the extended sketch allows for join variables to appear on the righthand side of the expressions (i.e. {\tt ??} stands for all variables in contrast to just left and right variables). A complete sketch admits bounded repetitions of the above loop (not illustrated), which then produces exactly the solution from Example \ref{ex:7-1} in 4 repetitions. But, one can piggy back on the first loop to update {\tt mtl\_rec} simultaneously with {\tt max\_rec[]} instead of having to wait for the next loop in the repetition. This leads to the discovery of an optimal join (i.e. the one in Section \ref{sec:mtlr}), compared to a less efficient join of Example \ref{ex:7-1}. Note that both joins are valid solutions of the sketch. In the implementation we first search for the simplest join by initially allowing only one repetition. % of the sketch.

  Assuming that $h_L^{-1}$ returns a constant-length (multidimensional) sequence, and $h_L$ is a homomorphism, then a join is guaranteed to exist  in the space described by our sketch. The synthesis procedure can {\em soundly} declare $h_L$ not to be a a homomorphism when it cannot find a join.

  \paragraph{Expression grammar}
  The grammar of expressions used to complete the holes $\rhole$, $\lrhole$ and $\lrrechole$ during the syntax-guided synthesis of the join operator is presented in Figure \ref{fig:grammar}. This grammar is parameterized by (1) the operators that can be used in the expression and (2) the set of expressions ($\mathit{nVars}$ and $\mathit{bVars}$) that can be used in the leaves of the expression tree and (3) the maximal expression height allowed $\kappa$.
  % Clarification
Sets of operators of different types ($\numbinary, \compare$ and $\boolbinary$) are given in the figure and can be extended if the input program uses additional operators. In practice, the set of available operators is gradually increased until a solution is found, starting with the set of operators that appear in the input program.
% For example, if a hole is expected to be of type \verb!int!, the topmost operator of the grammar will need to be either a binary arithmetic operation or a conditional, but cannot be a binary boolean operation.
The set of variables available in the expression depends on the hole type ($\rhole$, $\lrhole$ or $\lrrechole$), as discussed previously.
Finally, the parameter $\kappa$ is gradually increased until a solution is found; in practice, we observed that $\kappa \leq 2$.

For example, in the sketch presented above, most holes only need to be replaced by a single variable and only one hole needs to be replaced by an expression of height one (\texttt{rec\_l[j] + max\_rec\_r[j]}) to get the solution presented in Figure \ref{fig:mtlr-join}.

\begin{figure}
  % Grammar of expressions :
  \framebox{
    \begin{minipage}{7cm}
      \begin{minipage}[t]{7cm} \footnotesize
        $ ne_0 $ \hspace*{20pt} $::=  x\ |\ x[ej]\ | \ c $ \hfill $x \in nVars, c$ a numeric constant \\
        $ ne_{\kappa>0}$ \hspace*{12pt} $::= \ ne_{\kappa-1} \ \numbinary \  ne_{\kappa-1}\ | \ -ne_{\kappa-1}$ \\
        \hspace*{45pt}$| \ \texttt{if} \ (be_{\kappa-1}) \ {ne_{\kappa-1}} \ \texttt{else} \ {ne_{\kappa-1}}$ \\
      \end{minipage}              %
      \hfill
      \begin{minipage}[t]{7cm}\footnotesize
        $ be_0 $ \hspace*{20pt} $::=  b\ | b[ej]\| \ \mathit{true}\ |\ \mathit{false}$ \hfill $b \in bVars$ \\
        $ be_{\kappa>0}$ \hspace*{12pt} $::= \ be_{\kappa-1} \ \boolbinary \  be_{\kappa-1}\ | \  \neg be_{\kappa-1}$ \\
        \hspace*{45pt}$| \ ne_{\kappa-1} \ \compare \ ne_{\kappa-1}$ \\
        \hspace*{45pt}$| \ \texttt{if} \ (be_{\kappa-1}) \ {be_{\kappa-1}} \ \texttt{else} \ {be_{\kappa-1}}$ \\
      \end{minipage}\\
      \begin{minipage}[t]{7cm}\footnotesize
        \begin{tabular}{lclc}
          $ ej$ & $:=$ & $ j \numbinary c$ & $j$ an iterator, $c$ integer constant\\
        \end{tabular}
      \end{minipage}            %
      \hfill
      \hrule
      % Operators used in the grammar:\\
      \begin{minipage}[t]{7cm}\footnotesize
        \begin{tabular}{lclc}
          $\numbinary$ & $:=$ & $ \  +, -, min, max,\times, \div$ & binary numeric\\
          $\compare$ & $:=$ & $ \  >, >=, <, <=, =  $ & comparisons\\
          $\boolbinary$ & $:=$ & $ \ \wedge, \vee $ & binary boolean\\
        \end{tabular}
      \end{minipage}
    \end{minipage}                %
  }

  \caption{\small Grammar of expressions used for $\rhole$, $\lrhole$ and $\lrrechole$ holes. $ne_{\kappa}$ and $be_{\kappa}$ correspond to expressions of depth up to and equal to $\kappa$. $\mathit{nVars}$ and $\mathit{bVars}$ stand for numeric and booleans variables.  \label{fig:grammar}}
  \vspace{-15pt}
\end{figure}

\subsection{Summarized Loop Synthesis}\label{sec:mjoin}
Assuming that the loop is memoryless, summarization of the loop boils down to the synthesis of the operator $\circledcirc$ from Figure \ref{fig:diagram}. We argue why this problem is nearly identical to the synthesis of a homomorphic join.

\begin{proposition}
  A multidimensional rightwards function is memoryless iff we have $\circledcirc$ and $\zero_g$ that satisfy the specification $ \phi(\circledcirc, \zero_g) \overset{def}{\equiv} \forall d \in D, \delta \in S^{n-1}, \ \mathbbm{G}(d)(\delta) = d \circledcirc \mathbbm{G}(\zero_g)(\delta)$.
\end{proposition}

It is straightforward to see why the above characterization is equivalent to the one given in Definition \ref{def:memoryless}. One can also show that $\phi$ is identical to the definition of a homomorphism.

\begin{proposition}\label{prop:mjoin}
  $\phi$ holds for a family of functions $\mathbbm{G}$ iff every member of the family is $\circledcirc$-homomorphic.
\end{proposition}

\input{proofs/memoryless-join}

Therefore, the operator $\circledcirc$ can be synthesized as a homomorphism join operator of the functions in family $\mathbbm{G}$, which is the problem that we have already addressed in the previous subsection. The only point of difference is that the complexity budget set for a memoryless join operator $\circledcirc$ and the parallel join operator $\odot$ (previously discussed) are different. The budget for $\odot$ is determined by the depth of the summarized loop, whereas the budget for $\circledcirc$ is determined by the original loop nest's depth, as indicated in Figure \ref{fig:schema} respectively by $k$ and $n$.

There are two modifications to the sketch compilation function of the join synthesis. First, instead of replacing state variables $x \in \SVar$ by holes, we simply put the variable $x_L$ from the left thread, since the right operand of $\circledcirc$ is the result of applying $g$ to only one element of $\seq^{n-1}$ and looking for a join that iterates on its inputs more than once makes no sense. Then, the sketch compiled from the body of $g$ is wrapped in a loop iterating over the size of an element of $D$ instead of a constant. That is, if $D$ contains scalar elements, the sketch is constant-time.

If the loop is memoryless, this synthesis procedure always succeeds, even if it has to fall back on returning the trivial answer (see Proposition \ref{prop:lift-triv}). But, as discussed in section \ref{sec:lift}, the goal is to find the simplest join. This is achieved by two mechanisms. First, the sketch complexity is at most the complexity of the data. For example, a linear join for scalar variables will only be necessary if a linear variable has been added through lifting. Second, the expressions completing the holes are the simplest possible, because the search for a solution increases the depth $\kappa$ until a solution is found.

%%% Local Variables:
%%% mode: latex
%%% TeX-master: "paper"
%%% End:

%% file: proofs/memoryless-join.tex
\begin{proof}
  Consider a family of rightwards (or leftwards) functions $\mathbbm{G} : D \rightarrow (\seq^{n-1} \rightarrow D)$ defined by $\forall d \in D, \delta \in \seq^{n-1}, \ \mathbbm{G}(d)(\delta) = \foldl(\oplus) \ \delta \ d$ with some operator $\oplus : D \times \seq^{m-2} \rightarrow D$.

  Let us first define what it means for a function in $\mathbbm{G}$ to be homomorphic, since their signature is slightly different from the functions in Definition \ref{def:hom}.
   % \azadeh{This last sentence is a bit vague. Perhaps it will help to repeat the signature and point out the difference?} \azadeh{I don't see a change here? still the same sentence as before?}
   Then, we remark in Proposition \ref{prop:hom-family-reduction} that for every $g$ in $\mathbbm{G}$ to be $\circledcirc$-homomorphic, that is for the \emph{family of functions} to be homomorphic (defined below in Definition \ref{def:hom-family}), we only need to prove that the function $\mathbbm{G}(\zero_g) : \seq^{n-1} \rightarrow D$ is $\circledcirc$-homomorphic, for some $\zero_g$ in $D$. This leads us to our conclusion, equating $\phi$ to the specification of a family of homomorphisms.

\begin{definition}\label{def:hom-family}
  A family of rightwards (or leftwards) functions $\mathbbm{G}: D \rightarrow (\seq^{n-1} \rightarrow D)$ is a family of $\circledcirc$-homomorphisms for binary operator $\circledcirc : D \times D \rightarrow D$ with identity element $\zero \in D$ iff for all sequences $\delta, \delta' \in \seq^{n-1}$ and $\forall d \in D$ we have
  $\mathbbm{G}(d)(\delta \ccat \delta') = \mathbbm{G}(d)(\delta) \circledcirc \mathbbm{G}(\zero)(\delta')$.
\end{definition}

Remark the asymmetry in the definition: the right hand operand of the $\circledcirc$ operator is independent from $d$. This is necessary, as illustrated by the following example. Take the family of sum functions initialized with an arbitrary integer: we have $D = \int$ and $\mathbbm{G}(d)(\delta) = d + sum(\delta)$ where $sum$ returns the sum of all the elements of $\delta$.  Then, for every integer $d$:
\[\mathbbm{G}(d)(\delta \ccat \delta') = d + sum(\delta) + 0 + sum(\delta')
= \mathbbm{G}(d)(\delta) + \mathbbm{G}(0)(\delta').\]
Using $d$ on both side of $+$ would have yielded the wrong answer. The asymmetry allows for every member of the family to have the same homomorphic join operator, in this case $+$.

To prove that a family of rightwards functions is homomorphic, there is no need to prove that for every $d$, the function $\mathbbm{G}(d)$ is homomorphic. It suffices to prove it for $\mathbbm{G}(\zero)$, as the following proposition states.

\begin{proposition}\label{prop:hom-family-reduction}
  A family of rightwards (or leftwards) functions $\mathbbm{G}$ is a family of homomorphisms iff there is an element $\zero \in D$ such that $\mathbbm{G}(\zero)$ is $\circledcirc$-homomorphic for some $\circledcirc : D \times D \rightarrow D$.
\end{proposition}

  Remark that if $\mathbbm{G}$ is a family of homomorphisms, then in particular $\mathbbm{G}(\zero)$ is a homomorphism.

  Now, assume that we have an element $\zero$ and operator $\circledcirc$ such that  $\mathbbm{G}(\zero)$ is $\circledcirc$-homomorphic. We are only interested in computable functions that have a countable domain, and therefore have weak inverses. We denote the weak inverse of $\mathbbm{G}(\zero)$ by $\mathbbm{G}(\zero)^{-1}$: we have $\forall d \in D$, $\mathbbm{G}(\zero)(\mathbbm{G}(\zero)^{-1}(d) = d$.
  Let $\delta = [\delta_0, \ldots, \delta_n]$ a sequence of length $n$, we can develop the function application as follows (where $\gamma = \mathbbm{G}(\zero)^{-1}(d) = [\gamma_0, \ldots, \gamma_{n'}]$):
  \begin{align}
    \mathbbm{G}(d)(\delta) &= (\ldots (d \oplus \delta_0) \oplus \ldots \oplus \delta_n)\nonumber\\
                           &= (\ldots (\mathbbm{G}(\zero)\left(\mathbbm{G}(\zero)^{-1}(d)\right) \oplus \delta_0) \oplus \ldots \oplus \delta_n)\nonumber\\
                           &= (\ldots ((\ldots (\zero \oplus \gamma_0) \ldots \oplus \gamma_{n'}) \oplus \delta_0) \oplus \ldots \oplus \delta_n)\nonumber\\
                           &= \mathbbm{G}(\zero) (\gamma \ccat \delta) =  \mathbbm{G}(\zero)(\mathbbm{G}(\zero)^{-1}(d) \ccat \delta) \nonumber
  \end{align}

  Let $\delta$ and $\delta'$ two sequences. We use the previous result, and the fact that $\mathbbm{G}(\zero)$ is homomorphic:
  \begin{align}
    \mathbbm{G}(d)(\delta \ccat \delta') &= \mathbbm{G}(\zero)(\mathbbm{G}(\zero)^{-1}(d) \ccat \delta \ccat \delta')\nonumber\\
                                         &= \mathbbm{G}(\zero)(\mathbbm{G}(\zero)^{-1}(d) \ccat \delta) \circledcirc \mathbbm{G}(\zero)(\delta')\nonumber\\
                                         &= \mathbbm{G}(d)(\delta) \circledcirc \mathbbm{G}(\zero)(\delta')\nonumber
  \end{align}
  Therefore, $\mathbbm{G}$ is a family of homomorphisms.

Proposition \ref{prop:hom-family-reduction} justifies the Definition \ref{def:hom-family} by proving that the latter matches exactly the definition of a homomorphism in Definition \ref{def:hom}. %The only difference is that a parameter state needs to be taken in account.

Let us come back to the original problem. Recall the correctness specification used in the summarized loop synthesis:
\[ \phi(\circledcirc, \zero) = \forall d \in D, \forall \delta \in \seq^{n-1},  \mathbbm{G}(d)(\delta) = d \circledcirc \mathbbm{G}(\zero)(\delta)\]
We want to prove that $\phi$ holds for the family of function $\mathbbm{G}$ iff every $g$ in $\mathbbm{G}$ is a homomorphism.

If $\mathbbm{G}$ is a family of homomorphisms, then $\phi$ is satisfied: it is the homomorphism definition with $\delta = []$ and $\delta = \delta'$.

If $\phi$ is satisfied, we have $\circledcirc$ and $\zero$ such that $\forall d \in D, \forall \delta \in \seq^{n-1}$:
\[ \mathbbm{G}(\zero)(\mathbbm{G}(\zero)^{-1}(d) \ccat \delta) = \mathbbm{G}(\zero)(\mathbbm{G}(\zero)^{-1}(d))  \circledcirc \mathbbm{G}(\zero)(\delta)\]
which shows that $\mathbbm{G}(\zero)$ is $\circledcirc$-homomorphic, and, by Proposition \ref{prop:hom-family-reduction}, every $g$ in $\mathbbm{G}$ is $\circledcirc$-homomorphic.
\end{proof}

%%% Local Variables:
%%% mode: latex
%%% TeX-master: "../paper"
%%% End:

%% file: algorithm.tex
%!TEX root =  paper.tex
\section{Automatic Lifting} \label{sec:alg}
As argued in Section \ref{sec:mlift}, a {\em memoryless lift} is a special instance of the homomorphism lift and both problems admit the same algorithmic solution. Here, we present an algorithm for discovering a {\em homomorphism lift}, which would respectively apply to modules (III) and (IV) in Figure \ref{fig:schema}.

\subsection{Rewriting Oracle}\label{sec:lift-alg}
Assume a memoryless function $f: \seq^n \to D$ defined recursively as $f(\sigma) = \foldl(\oplus) \circ \map(g)(\sigma)$ is not a homomorphism. Let $h: \seq_{D} \to D$ be the summarization of $f$, as defined in Proposition \ref{prop:conv}, that is:
\begin{align*}
h([]) &= f([]) \\
\forall a \in D: h(x \ccat [a]) &= h(x) \oplus a
\end{align*}
for $x \in \seq_{D}$ and $a \in D$. Recall that according to Theorem \ref{thm:convl}, a lifting for $h$ can be computed instead of a lifting for $f$.

Let $x,y \in \seq_{D}$, $\vec{s} = h(x)$, and $y = [a_1, \ldots, a_k]$ (with $a_i \in D$). The sequential computation of $h(x \ccat y)$ can be written as:
\begin{align}
h(x \ccat [a_1, \ldots, a_k]) = (\cdots \left(\vec{s} \oplus a_1\right) \oplus \cdots ) \oplus a_k. \label{eqn3}
\end{align}

Lifting $h$ to a homomorphism $\lift{h}{D'}$ corresponds to extending the image of $h$ to $D \times D'$ and lifting the initial state to $(\vec{s_0},\vec{s_0'}) = \lift{h}{D'}([])$. If $\lift{h}{D'}$ is a homomorphism, then there exists a binary operator $\odot$ such that ($(\vec{s},\vec{s'}) = \lift{f}{D'}(\sigma)$):
\begin{align}
  &\lift{h}{D'}(x \ccat y) = (\vec{s},\vec{s'}) \odot \lift{h}{D'}([a_1, \ldots, a_k]) \nonumber \\
                                        &= (\vec{s},\vec{s'})  \odot
                                       \left(
                                       \cdots (((\vec{s_0},\vec{s'_0})
                                       \lift{\oplus}{} a_1)
                                       \lift{\oplus}{} \cdots)
                                       \lift{\oplus}{} a_k
                                       \right).
                                       \label{eqn2}
\end{align}

First, the following proposition, adapted from Theorem 6.2 of \cite{PLDI17},  indicates that $\vec{s'}$ is not relevant to the discovery of the lifting $\lift{h}{D'}$.

\begin{proposition}\label{thm:asym}
If there exists a $\odot$-homomorphic lifting  $\lift{h}{D'}$ of $h$, then there exists a $\oast$-homomorphic lifting $\lift{h}{D'}$ where for all $c,d \in D$ and $c',d' \in D'$, there exists functions $\theta$ and $\theta'$ such that
\[
(c,c') \oast (d, d') =
(\mathit{\theta}(c,d,d'), \mathit{\theta'}(c,d,c',d')). \]
\end{proposition}
The significance of Proposition \ref{thm:asym} is that the value of the $D$ component of the join result (i.e. $\mathit{\theta}(c,d,d')$) need not depend on the value of the $D'$ component of its left input (i.e. $c'$).
Interpreting this for equation \ref{eqn2}, we conclude that the value of $\pi_D (h(x) \odot  \lift{h}{D'}([a_1, \ldots, a_k]))$, only depends on $h(x)$ and $\lift{h}{D'}([a_1, \ldots, a_k])$. Therefore, one can imagine an algorithm that starts from an arbitrary state $\vec{s}=h(x)$ and tries to {\em guess} what $\lift{h}{D'}([a_1, \ldots, a_k]))$ should look like (as in what $D'$ should be) so that such a join exists. Specifically, there is a function ${\color{red}\theta}$ such that:
\begin{align}
  ( \cdots (
  (\vec{s} &\oplus a_1)
  \oplus \cdots )
  \oplus a_k)
  =  \label{eqn4}\\
   &{\color{red} \theta(} \vec{s},
     \   h([a_1, \ldots, a_k]),
     \   \pi_{D'} \circ \lift{h}{D'}([a_1, \ldots, a_k])) \nonumber
\end{align}
and, the third component of ${\color{red}\theta}$ is the auxiliary computation that needs to be discovered. Note that ${\color{red}\theta}$ could stand for the computation of a loop nest, in contrast to the setting in \cite{PLDI17} where it stood for an expression (i.e. loop-free code), which means the lifting algorithm in \cite{PLDI17} is not applicable and a new lifting algorithm is required. Equation \ref{eqn4} is the key to our algorithmic solution. The left hand side corresponds to the sequential execution and the right hand side corresponds to a parallel one. Since the join does not have access to the input (i.e. $a_1, \dots, a_k$),  the value of $\pi_{D'} \circ \lift{h}{D'}([a_1, \ldots, a_k])$ (i.e. the extra information in the signature of $\lift{h}{D'}$) has to be computed by the worker threads and passed on to the join.

\begin{figure}[t]
\begin{center}
\includegraphics[scale=0.22]{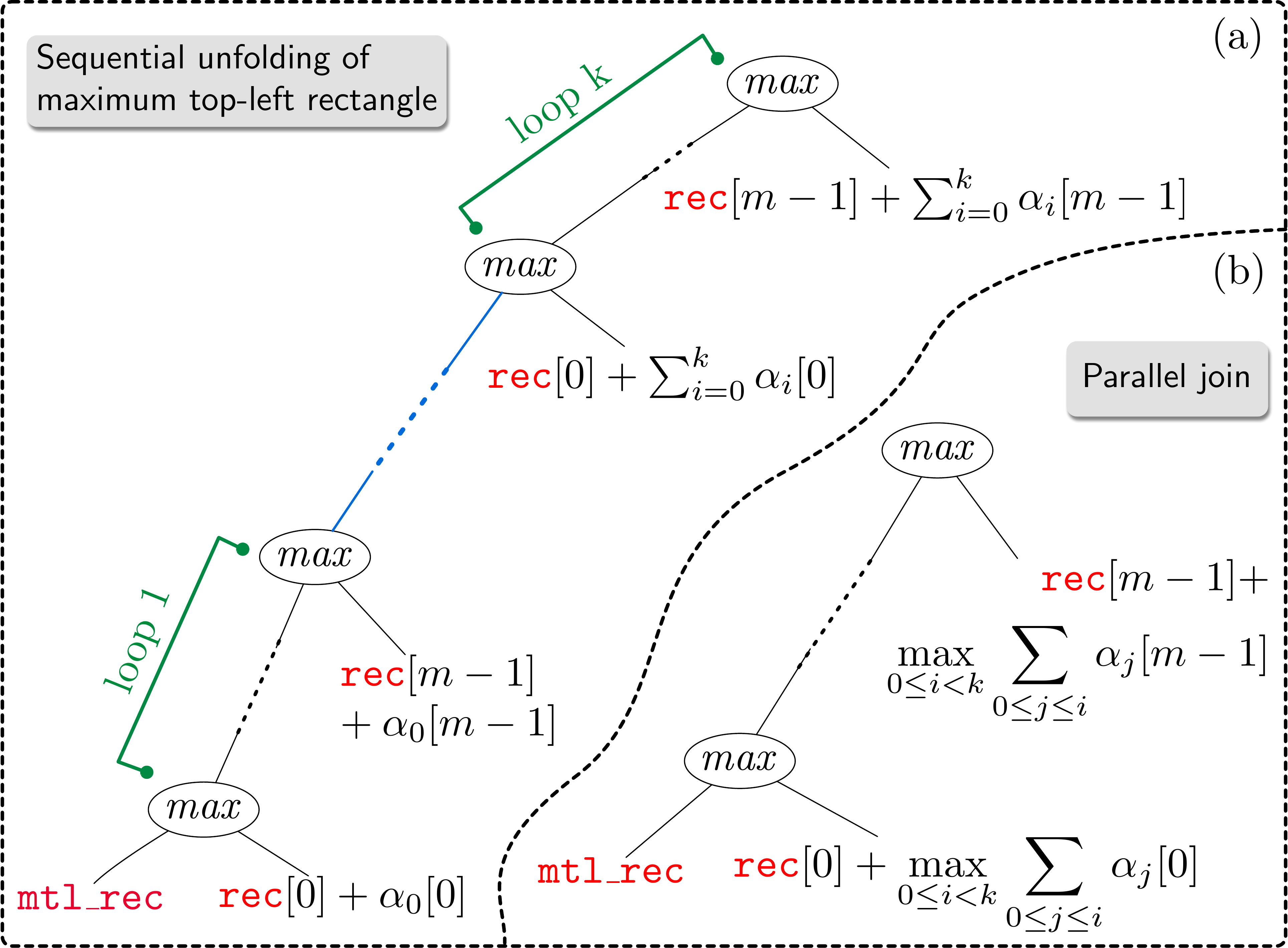} \vspace{-15pt}
\caption{Sequential (a) vs parallel (b) computations.} \vspace{-15pt}
\label{fig:ex8}
\end{center}
\end{figure}

\begin{example}
Figure \ref{fig:ex8} illustrates the two sides of Equation \ref{eqn4} for the maximum top-left rectangle example of Section \ref{sec:mtlr}. Variables in red indicate the values of state variables from $\vec{s}$. Each $\alpha$ technically should include a field for {\tt rec[]} and a field for {\tt mtl\_rec} of the saved states after summarization. But, we abuse notation and take $\alpha[i]$ to mean $\alpha.\mbox{\tt rec}[i]$ for brevity. Note that the sequential computation executes $k$ instances of a loop that iterates $m$ times to update the value of {\tt mtl\_rec} variable; one for each row of the original input.
However, in the parallel join, there is budget only for one (or generally constantly many) of these loops. The two (expression) trees (a) and (b) correspond to equivalent computations. The tree (b) is more compact provided that the values of the terms $\max_{0\le i < k}  \sum_{0 \le j \le i} \alpha_j[l]$ are available (i.e. computed before the join). These are exactly the auxiliary computation that are extrapolated, once the left tree (a) is rewritten to the equivalent right tree (b), and are stored in the {\tt max\_rec[]} variable in Figure \ref{fig:mtlr}(c).
\end{example}

\subsection{The Algorithm}\label{sec:heuristic}

If one starts from an arbitrary unfolding of the sequential computation (i.e. the lefthand side of Equation \ref{eqn4} for an arbitrary $k$), and attempts to {\em rewrite} it to a form that adheres to the righthand side, then one can extrapolate a guess for $\pi_{D'} \circ \lift{h}{D'}$ from $\pi_{D'} \circ \lift{h}{D'}([a_1, \dots, a_k])$.
Let us assume an oracle {\em Normalize} performs the left-to-right hand side transformation, and another oracle {\em Discover-Recursion} discovers the implementation of $D'$ components of $\lift{h}{D'}$. The {\em Normalize} oracle transforms an expression into another, while the goal of {\em Discover-Recursion} is to synthesize a function $f'$ from the expression of its unfolding on an input sequence. We propose heuristic algorithms implementing these two oracles.

The algorithm for {\em Normalize} uses (generic) algebraic equalities, applies them step by step until it reaches the desired form. The key question is, how would the algorithm know that it has reached its target? To characterize this, we need to define a normal form for ${\color{red}\theta}$.

\intextsep=0pt
\columnsep=3pt

\paragraph{Normal form.}\label{sec:normalform}
 \begin{wrapfigure}{h}{1.9cm}
 \includegraphics[scale=0.25]{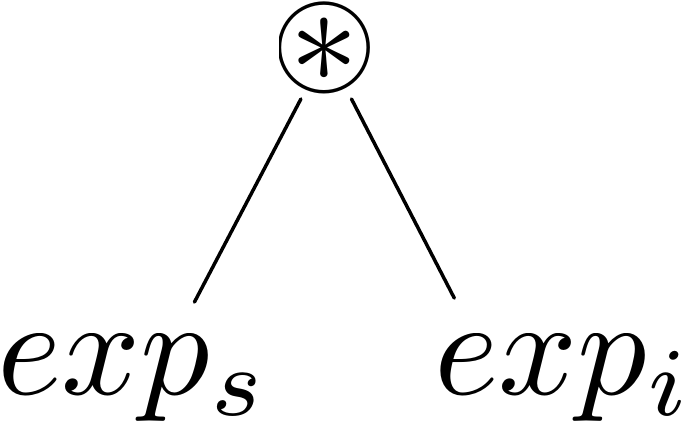}\vspace{-7pt}
 \end{wrapfigure}
An expression $e_c$ over scalar variables is in {\em constant normal form} if it is of the form illustrated on the right where $\mathit{exp}_s$ is an expression containing only state variables, $\mathit{exp}_i$ is an expression of input variables of ${a_1, a_2, \dots a_k}$, and $\circledast$ stands for a constant-size expression skeleton consisting of operators and constants (i.e. no variables).

 \begin{definition}
   An expression $e$ is in $\ovee$-\emph{recursive normal form} if it is defined recursively using an operator $\ovee$ as $e = e_c\ |\ e_c \ovee e$, where $e_c$ is in constant normal form (the base case) and $\ovee$ satisfies the same condition as $\circledast$.
 \end{definition}

 For example, in Figure \ref{fig:ex8}(b), each leaf of the tree is in constant normal form, and therefore, the entire tree is in $\max$-recursive normal form.

The normal form does not have to be unique, and in the context of parallelization, one aims for the {\em simplest} normal form.
Intuitively, the constant normal form corresponds to a constant time join. The expressions over input variables are precisely what need to be {\em additionally} computed and made available to the {\em parallel join}.
However, the parallel join, in general, may not be constant time and the recursive nature of the definition addresses this. Note that the definition can be easily generalized to higher dimensions by replacing the constant normal form by another recursive normal form (over a distinct fresh operator, e.g. $\owedge$).

\paragraph{Normalization.}  Implementing an ideal rewriting oracle is impossible, since the problem of existence of a normal form is undecidable in the general case (equivalent to {\em the word problem}). There has been a lot of research in the area of rewrite systems notably for associative and commutative operators  \cite{Narendran1991, Marche1998} that can inspire several heuristics for normalization. We employ a cost-based search for the normal form as a heuristic to workaround the undecidability. Our algorithm uses a set of standard algebraic equalities as rewrite rules $\rrules$ and searches for the {\em shortest} normal form. The rewrite rules in $\rrules$ are applied to an expression if they reduce its cost, and the rewriting process terminates when no rule can be applied.

Our algorithm operates in two phases. In the first phase, its goal is to find a constant normal form. If it succeeds, the task is done (e.g. this is the case for lifting the example from the introduction). If it fails, the second phase tries to rewrite the result of the first phase into a recursive normal form. Both phases perform a cost-based search and are distinguished by their cost functions.
The cost function for the first phase is defined by the number of occurrences and the depth of the state variables (of $h$) in the expression tree.
The cost function is identical to the one from \cite{PLDI17}, which is no coincidence since \cite{PLDI17} focuses solely on constant normal forms.

In the second phase, the algorithm makes a guess about $\ovee$, inspired by expression $e$ which is the result of the first phase, and attempts to rewrite $e$ to a $\ovee$-recursive normal form. Since phase one forces $\vec{s}$ (or $h(x)$) to the lowest possible depth, operators that appear near the root of expression $e$ are good candidates for $\ovee$.
The algorithm chooses the simplest $\ovee$ such that $e = e_c \ovee e'$, where $e_c$ is in constant normal form and $e'$ is an arbitrary expression.
For a fixed $\ovee$, a cost function is defined. It combines the sum of the sizes of expression not in constant normal form and the count of expressions in constant normal form. Formally:
\begin{definition}
  The cost function $\mathrm{Cost}_{\ovee} : \Exp \to \int \times \int$ relative to operator $\ovee$ returning a pair $(size, c_{\ovee})$ is defined by:
  {\small\[
  \mathrm{Cost}_{\ovee}(e) =
  \left\{\begin{matrix}
      \mathrm{Cost}_{\ovee}(e_1) + \mathrm{Cost}_{\ovee}(e_2) \quad \text{ if }e=e_1 \ovee e_2, \hfill\\
      (0 , 1)  \quad \hspace{5em} \text{ if } e \text{ is in constant normal form,}\hfill\\
      (expsize(e), 0)  \quad \hspace{4em}\text{otherwise.}\hfill
    \end{matrix}\right.
  \]}
  where $expsize$ returns the size of the expression tree.
\end{definition}
Intuitively, the expression is in $\ovee$-recursive normal form when it has cost $(0,\_)$. Moreover, we are interested in the normal form with the smallest count of subexpressions in constant normal form. A rule is applied if it decreases $size$ or, if it increases $c_{\ovee}$ when $size$ cannot be decreased and $size > 0$.

In the example of Figure \ref{fig:ex8}, the expression in (a) is initially in $\max$-recursive normal form with cost $(0, km + 1)$. When the expression is rewritten in the first phase (using the cost function from \cite{PLDI17}) with the aim of  reducing the occurrences and depths of state variables, the cost goes down, but the normal form is lost.  Since a constant normal form is not reached at the end of the first phase, the second phase is applied, using the cost function above, which yields the expression in Figure \ref{fig:ex8}(b). This expression is in $\max$-recursive normal form with cost $(0, m + 1)$.

If the process fails to find a normal form for $\ovee$, then another operator $\owedge$ is guessed and the process is repeated.
Since the expressions are finite-sized, only a finite number of guesses are necessary, and the process is guaranteed to terminate.
This simple heuristic is a small part of the contributions of this paper, though it is promisingly effective as demonstrated in Section \ref{sec:experiments}.

\paragraph{Recursion discovery} The {\em normalize} heuristic distills the $\pi_D \circ \lift{h}{D'}([a_1, \ldots, a_k])$ part from its input expression, which we know is required for a join operation to exist. It remains to discover the recursive (i.e. looped) computation that can be added to the original program that would produce this required information. More precisely, the goal of recursion discovery is to synthesize a function that computes the expression $u_k \equiv \pi_{D'} \circ \lift{h}{D'} ([a_1, \ldots, a_k])$ for any $k > 0$.
Recursion discovery, based on input/output examples, has been previously studied \cite{kitzelmann2006inductive}. Our specific instance of the problem is simpler and amenable to a simple heuristic solution.

Since $u_k$ is recursively (rightward) computable, there is an operator $\boxplus$ such that
$u_k =  u_{k-1} \boxplus a_k$ for $k > 0$. If $u_{k-1}$ and $u_k$ are simple expressions, we can extrapolate a hint on what $\boxplus$ is by identifying $u_{k-1}$ as a subexpression of $u_k$ (that is precisely the subtree isomorphism \cite{shamir1999faster}). In general, however, $u_k$, $u_{k-1}$ and $a_k$ can be collections of complex expressions; i.e. lists of expressions as is the case for the example of Section \ref{sec:mtlr}. The solution remains the same, except, we identify families of subtree isomorphisms. In our implementation, we simplify the problem further by looking for specific patterns of subtree isomorphisms corresponding to recursion schemes such as $\mathit{zip}$, $\mathit{scans}$ or $\mathit{folds}$. For example, a $zip$ operator translates to having each expression in $u_{k-1}$  isomorphic to a subtree of one expression in $u_k$.

For example, the list of expressions $\max_{0\le i < k}  \sum_{0 \le j \le i} \alpha_j[l]$ (for $0\le l <m$) from Figure \ref{fig:ex8}(b) can be computed in an auxiliary array {\tt max\_rec[]} using a $\mathit{zip}$ operation and the state variable  {\tt rec[]} as follows:
\[ \mathtt{max\_rec} = \mathit{zip} \ (max) \ \mathtt{rec} \ \mathtt{max\_rec}.  \]
The next section illustrates how this auxiliary can be found.

%%% Local Variables:
%%% mode: latex
%%% TeX-master: "paper"
%%% End:

%% file: extended_example.tex
\section{An extended example}\label{sec:extended-example}
In this section, we go through steps of the automated parallelization process for the maximum top-left subarray sum example of Section \ref{sec:mtlr}. The goal is to give a good intuition of how the heuristic algorithms described in this paper work on an example but not to describe them precisely.

The code is given in Figure \ref{fig:mtlr}(a). We call the corresponding function $mtls$, which is of type $\int [][] \rightarrow \int [] \times \int \times \int$ (before lifting) and the state is a triple of variables $\SVar = \{ \mathtt{rec[]}, \mathtt{row\_sum}, \mathtt{mtl\_rec}\}$.

\subsection{Summarizing the loop nest}
The first step in the parallelization process is to synthesize the summarized loop ((II) in Figure \ref{fig:schema}). As stated in Section \ref{sec:join}, the problem is similar to synthesizing a parallel join with a few modifications. We explain here how the sketch is generated and the solution found by the syntax-guided synthesis solver.

From the code in Figure \ref{fig:mtlr}(a) we synthesize the sketch that corresponds to finding a parallel join for the inner loop, with the complexity budget that is enforced by the budget for the summarized function. Remark that for this reason, the sketch resulting from the application of the compilation function of Section \ref{sec:join} can be repeated more than a constant number of times. Since there is a linear variable (\texttt{rec[]}), we will need a loop on the dimension of this linear variable. The solution will require a loop of the size of \texttt{rec[]}.

\begin{figure}[b]
\begin{subfigure}[b]{0.4\textwidth}
{\footnotesize
  \begin{lstlisting}[mathescape=true]
int mtl_rec = $\lrhole$;
int rec[] = $\lrhole$;
int row_sum = $\lrhole$;
for(j=0; j < m; j++){
  row_sum =  $\lrrechole$ + $\rhole$;
  rec[j] = $\lrrechole$ + $\lrrechole$;
  mtl_rec = max($\lrrechole$, $\lrrechole$);
}
\end{lstlisting}
}~
\end{subfigure}
\begin{subfigure}[b]{0.4\textwidth}
{\footnotesize
  \begin{lstlisting}
int mtl_rec = mtl_rec_l;
int rec[] = rec_l[];
int row_sum = 0;
for(j=0; j < m; j++){
  row_sum = row_sum + rec_r[j];
  rec[j] = rec[j] + rec_r[j];
  mtl_rec = max(mtl_rec, rec[j]);
}
\end{lstlisting}
}
\end{subfigure}
\caption{Sketch and solution for the memoryless join of $mtls$ (Sec. \ref{sec:mtlr})}\label{fig:memless-mtlr}
\end{figure}

Figure \ref{fig:memless-mtlr} presents the sketch for the memoryless join $\circledcirc$ and a solution (each hole has been completed). Given the generic correctness specification $\phi(\circledcirc, \zero_g) \equiv \forall d \in D, a \in int[]$,$\mathbbm{G}(d)(a) = a \circledcirc \mathbbm{G}(\zero_g)(a)$, the sketch and the generic grammar of Figure \ref{fig:grammar}, the solver finds the solution presented. Variables ending by \texttt{\_l} are variables from the left input of the join and the variables ending by \texttt{\_r} are the right input of the join. The identity state $\zero_g$ is $\{ [0], 0, 0 \}$.

Then, we obtain the summarized loop of Figure \ref{fig:mtlr}(b) by removing \texttt{row\_sum} which is not necessary for the computation of \texttt{mtl\_rec} anymore, and inserting the loop implementing $d \circledcirc \mathtt{inner\_loop[i]}$ where \texttt{inner\_loop[]} is the conceptual array representing the results of mapping the inner loop instances to the input. Remark that in the general case, the elements of \texttt{inner\_loop[]} are of type $D$, but in this case, the only element of $D$ that is required is the variable \texttt{rec[]}.

In the case where the variables are linear and every cell is modified in the inner loop, we cannot reduce the complexity of the loop nest. However, the summarized loop abstracts all the unnecessary information that the inner loop has computed. This will become more apparent in the next section, where the unfoldings of the function on symbolic inputs need to be inspected. In this example, the elements of \texttt{rec} that were column sums of prefix sums are summarized as simple column sums. In the symbolic execution, we reduce the amount of information in the cells from quadratic to linear.

\subsection{Lifting}\label{sec:developed-mtlr-lift}

The first attempt at synthesizing a join will fail for this example since the summarized loop cannot be parallelized. We need to lift the summarized function to a homomorphism ((III) in Figure \ref{fig:schema}), if possible.
The codomain of the summarized function $h_{mtls}$ is $D = \int[] \times \int$ before lifting. We drop \texttt{row\_sum} from the state, since it does not appear in the summarized loop of Figure \ref{fig:mtlr}(b).

\paragraph{Normalization}
\begin{figure}[h]
  \begin{subfigure}[t]{0.49\textwidth}
        \hspace{-5pt}\includegraphics[scale=0.55]{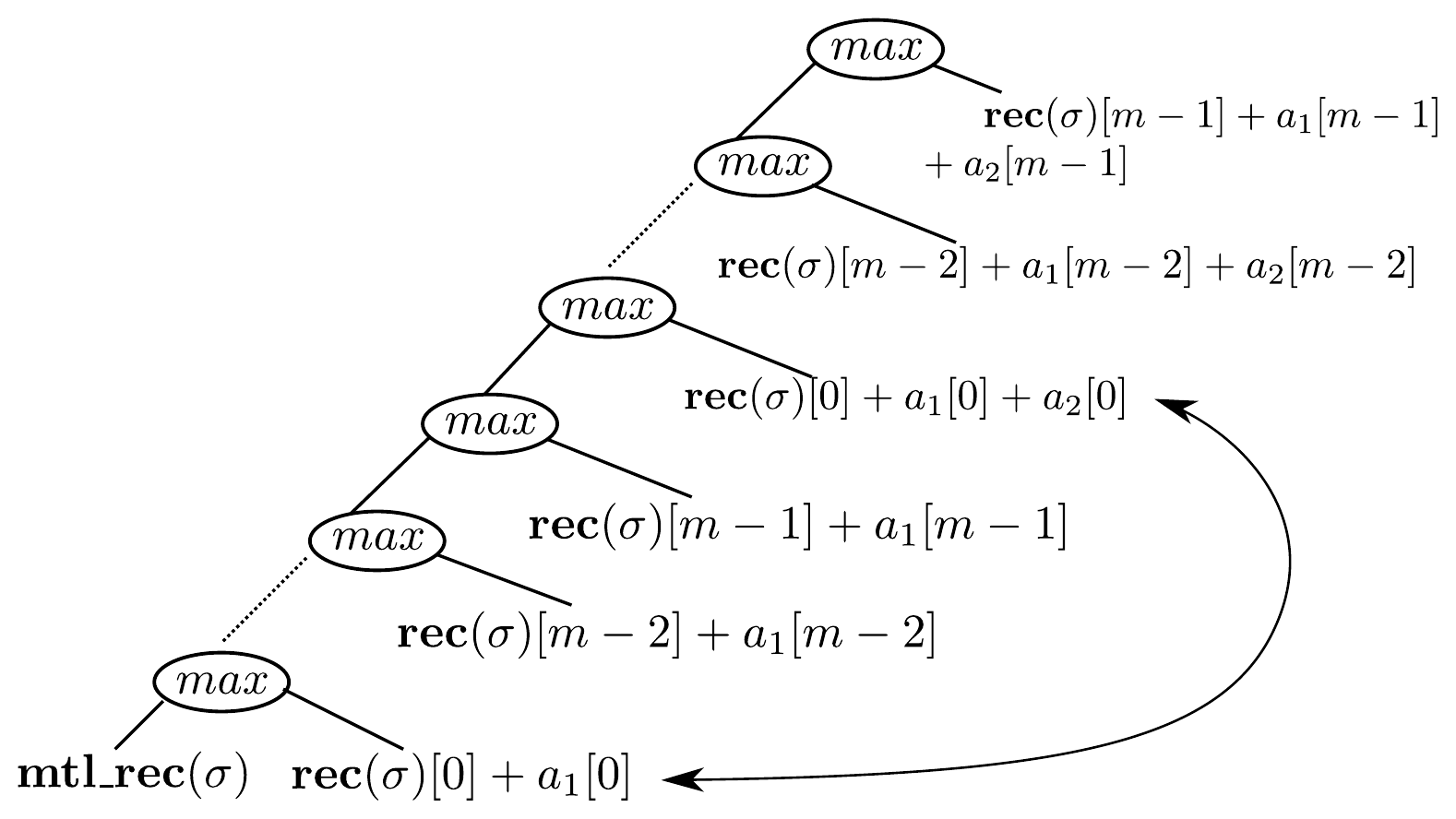}
    \caption{Expression tree after unfolding}\label{fig:unfold2}
  \end{subfigure}
  \begin{subfigure}[t]{0.49\textwidth}
    \hspace{-5pt}\includegraphics[scale=0.19]{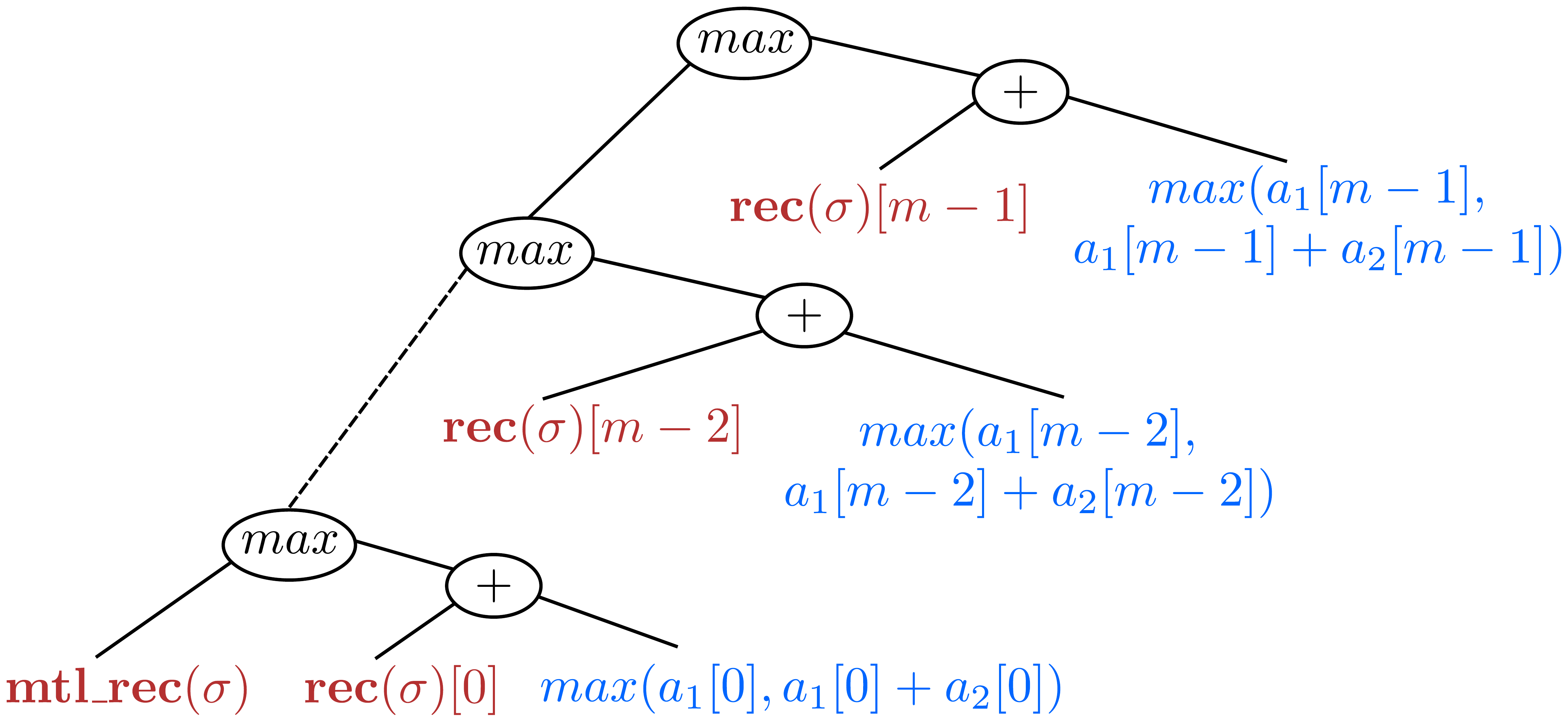}
    \caption{Normalized expression tree}\label{fig:unfold2norm}
  \end{subfigure}
  \caption{Expression trees of $mtl\_rec(\sigma \ccat [a_1, a_2])$}
\end{figure}

First, we compute the unfolding of $h_{mtls}$ over an input $\sigma \ccat [a_1, \ldots, a_k]$. Each $a_i$ for $0 \leq i \leq $ is itself an integer array of size $m$.
We denote the different projections of $h_{mtls}(\sigma)$ over the different parts of the codomain by $rec(\sigma)$ and $mtl\_rec(\sigma)$.

Let us start with $k=1$:
\\
\begin{tabular}{rcl}            %
$mtl\_rec(\sigma \ccat [a_1])$&$=$&$max(a_1[m-1] + rec(\sigma)[m-1],$\\
$                                  $&$$&$ \ max(a_1[m-2] + rec(\sigma)[m-2],$\\
$                                  $&$$&$ \ \ldots,$\\
$                                  $&$$&$ \ max(a_1[0] + rec(\sigma)[0],$\\
$                                  $&$$&$ mtl\_rec(\sigma)) \ldots)$))
\end{tabular}
\\
We can visualize the expression as a tree:

\vspace{0.3cm}
\hspace{-7pt}\includegraphics[scale=0.55]{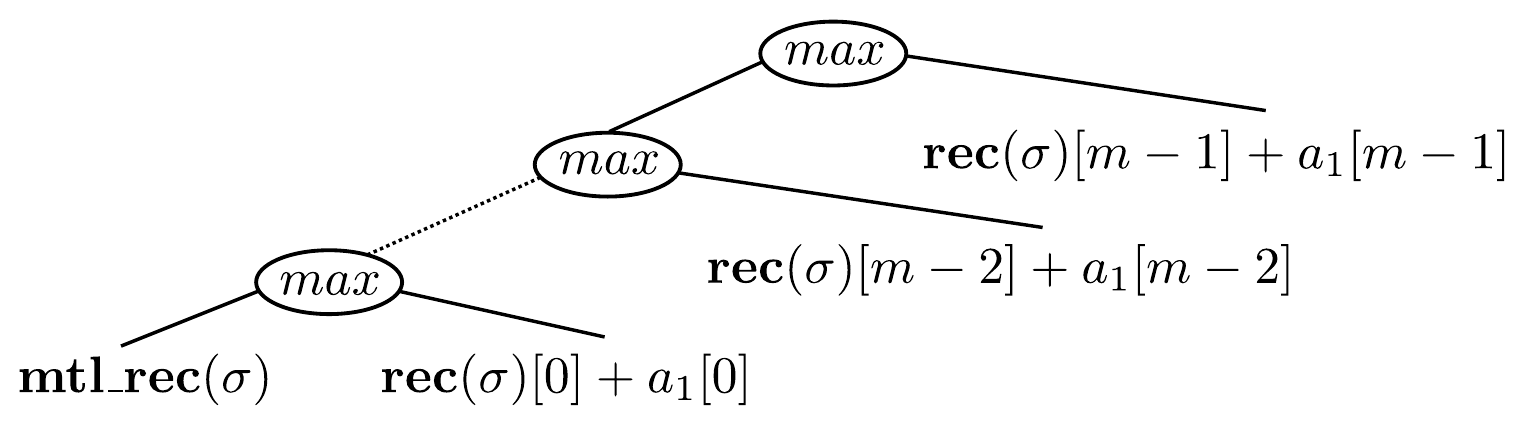}

Remark that this expression is in $max$-recursive normal form, and computable with a loop of size $m$ and it is already in the shortest ($max$)-recursive normal form. This is not surprising, since we only unfolded the function on one line of input.

The second unfolding will be computed with a loop of size $2m$: we need a loop of size $m$ to compute $mtl\_rec(\sigma \ccat [a_1])$ and $m$ additional iterations to compute the rest of the expression.
\\
\begin{tabular}{rcl}
  $ mtl\_rec(\sigma \ccat [a_1, a_2])$&$=$&\hfill\\
\end{tabular}
\\
\begin{tabular}{rcl}
 $\quad                                  $& &$max(a_1[m-1] + a_2[m-1] + rec(\sigma)[m-1],$\\
 $\quad                                  $&$$&$\ max(a_1[m-2] + a_2[m-1]+ rec(\sigma)[m-2],$\\
 $\quad                                  $&$$&$\ \ \ldots,$\\
 $\quad                                  $&$$&$\ \  max(a_1[0] + a_2[0] + rec(\sigma)[0],$\\
 $\quad                                  $&$$&$mtl\_rec(\sigma \ccat [a_1])) \ldots)))$
\end{tabular}
\\

Remark that the expression is already in $max$-normal form, since each subexpression (of the form $a_1[j] + a_2[j] + rec(\sigma)[j])$ or $a_1[j] + rec(\sigma)[j])$) is in constant normal form. However, this normal form does not represent a computationally efficient parallel join, since for an arbitrary $k$ we would need $km$ iterations, which is not acceptable under Proposition \ref{prop:tcomp} (the join would be $O(m^n)$ for $n$ lines of input).
However, we can normalize the expression to a shorter normal form. We write the expression tree in Figure \ref{fig:unfold2}.

Some terms in this expression tree can be factored together using the associativity of $max$ to reorganize the tree, and the factorization rule $max(a + b, a + c) \to a + max(b, c)$.  For example, the two expressions linked by the double arrow can be factored together:
$max(\mathbf{rec}(\sigma)[0] + a_1[0], \mathbf{rec}(\sigma)[0] + a_1[0] + a_2[0]) = \mathbf{rec}[0] + max(a_1[0], a_1[0] + a_2[0])$.
This rule decreases the number of occurrences of state variables (the $\mathbf{rec}[j]$) and is applied during the first phase of the normalization process. We can apply the same rule for all $m$ pairs of subexpressions and this yields the expression in Figure \ref{fig:unfold2norm}.

No rewrite rule can be applied to this expression in the first phase, but the above is not in constant normal form. When starting the second phase, the $max$ operator is picked and the matching expression cost is $(0, m + 1)$. The $max$-recursive normal form uses $m+1$ $max$ operators, and is computable using a loop of size $m$ and one additional statement. Each of the subexpressions under the tree spine of $max$ operators is in constant normal form: it is either only a state variable ($mtl\_rec(\sigma)$) or an expression that has the root operator $+$ and on one side, the state variable $\textcolor{statecolor}{mtl\_rec(\sigma)[j]}$ and on the other side the input-only dependent expression $\textcolor{inputcolor}{max(a_1[j], a_1[j] + a_2[j])}$, for $0 \leq j < m$. The second phase concludes immediately, and returns the expressions in $max$-recursive normal form.

Remark that this generalizes to any unfolding $k$ (as seen in Section \ref{sec:alg}). We will have a $max$-recursive normal form with a spine of length $m+1$, and the constant normal forms $e_c[j]$ that have input and state variables will be:
\[ e_c[j] =
  \textcolor{statecolor}{mtl\_rec(\sigma)[j]} +
  \textcolor{inputcolor}{max(a_1[j], a_1[j] + a_2[j], \ldots , \sum_{l=1}^k a_l[j])} \]

\paragraph{Recursion discovery}
In the next step, we want to extract the information that needs to be computed in the threads for the join from the expression in normal form. We denote the expressions that consitute this information by $u_k = \pi_{D'} \circ \lift{h_{mtls}}{D'}([a_1, \ldots, a_k])$. In the expression tree, it corresponds to the blue input dependent expression in each of the constant normal forms. For any unfolding $k > 0$ we have:
\[ u_k = \left\{max(a_1[j], a_1[j] + a_2[j], \ldots , \sum_{i=1}^k a_i[j]) \mid 0 \leq j < m \right\}\]

We have collected  $u_k$ and now need to discover a recursive function that can compute them. The codomain of the function needs to be extended. Since in this case $u_k$ is an integer collection of size $m$, we extend the codomain by an array of integers $D' = int[]$. We have to discover an operator $\boxplus: D' \times int[] \to D'$ such that $u_{k+1} = u_k \boxplus a_{k+1}$. To do so, we look for subtree isomorphisms (or, subexpression relations) between the elements of $u_k$ and the elements of $u_{k+1}$.

For any $j \in [0 .. m]$ the following equality defines a family of subtree isomorphisms for each pair $u_{k+1}[j], u_{k}[j]$:

\[ u_{k+1}[j] = max(u_k[j], \sum_{i=1}^{k+1}a_k[j])\]

This is not sufficient yet, since the sum term needs to be computed in non-constant time. But this sum is already part of the computation we are lifting, we have:
\[rec([a_1, \ldots, a_k, a_{k+1}])[j] = \sum_{i=1}^{k+1}a_k[j]\]
The operator of the lifting $\boxplus$ can in use parts of the existing function. In this case, the $max\_rec$ auxiliary is created, defined by:
\[ max\_rec = \mathbf{zip} (max) \ rec \ max\_rec \]
Which concludes our description of the lifting.
A description of the last step in the parallelization process, the join synthesis ((I) in Figure \ref{fig:schema}), has already been included as part of the example developed in Section \ref{sec:join}.

%%% Local Variables:
%%% mode: latex
%%% TeX-master: "paper"
%%% End:

%% file: experiments.tex
% !TEX root =  paper.tex
\section{Experimental Evaluation} \label{sec:experiments}
\input{table}
\paragraph{Implementation}\label{sec:implem}
Our methodology is implemented as an addition to our existing too \tool from \cite{PLDI17}.
We employ a new incremental approach to synthesizing the join to mitigate the large search problem. The state of the loop is partitioned into substates, and the join is synthesized for each substate separately, while preserving correctness. All synthesis times improve as a result of this strategy, but the complex ones improve more substantially; for example, for {\em maximum top-left subarray}, the synthesis time is reduced from over 1000 to 116.3 seconds.

The main idea behind this optimization is that if a function $h: \seq^n \rightarrow D$ is homomorphic, then any projection of $h$ over a domain $D' \subset D$ will be homomorphic, if the projection is well defined. The projection on $D'$ is well defined if there is no variable in $D'$ that depends on a variable in  $D \setminus D'$. We first define a sequence of sets $D_1 \subset D_2 \subset \dots \subset D$ such that the projection of $h$ over the $D_i$ are well defined.\\
Then, we start by solving the synthesis problem for $h$ projected on $D_1$.
We continue to $\pi_{D_2} \circ h$. But at this point, we can use the solution of the previous problem to solve part of the current problem. Intuitively, only the part concerning the variables in $D_2 \setminus D_1$ needs to be dealt with. We go on incrementally with the next projections, and finally with $h$.

\subsection{Evaluation}
\label{sec:eval}
To the best of our knowledge, there is no tool or technique that can parallelize the class of programs considered in this paper, and therefore, our evaluation of \tool \ is not comparative. The theoretical results presented in this paper justify many of the choices made in our proposed methodology. There are two key parts of the algorithmic modules, however, that require empirical evaluation: (1) The effectiveness of the heuristics proposed in Section \ref{sec:alg} for lifting; the (necessarily) incomplete algorithm has no theoretical guarantees for success. And, (2) the efficiency of SyGuS-based solution for parallel join generation and loop summarization; since it is impossible to predict synthesis times for the search-based routine.

\paragraph{\bf Benchmarks.}
We use a diverse set of benchmarks, where some implement highly non-trivial single-pass algorithms. Since a standard set of benchmarks for this problem space does not exist, we included any example we could find in the related work and parallel programming/algorithms text books, but only those that admit a divide-and-conquer solution according to Definition \ref{def:par}. Table \ref{tab:benchmarks} includes a complete list.  It is important to note that the difficulty of an instance is not directly co-related with code size or the classic notions of dependence such as sparsity of dependencies. Rather, it depends on how complex the required added computation is.

The code of the benchmarks is available at \url{https://github.com/victornicolet/parsynt-pldi19-benchmarks}.

\paragraph{\bf Performance of \btool.}

Table \ref{tab:benchmarks} presents the times spent in the two steps summarized loop and parallel join generation.
The synthesis time have been measured on a laptop with an dual core Intel Core m3-6Y30 CPU and 6 Gb RAM running 64-bit Ubuntu 16.04. Our goal is to provide an implementation that can synthesize the parallel implementation of a sequential algorithm in reasonable time, and much faster than a programmer. It is not aimed at being integrated in a compiler, but rather in a programmer-aid tool.

We do not report the individual times for liftings, since they are negligible compared to the synthesis times; largest lifting time was 12ms (3 orders of magnitude less than smallest synthesis times).

Table \ref{tab:benchmarks} reports how many auxiliary accumulators were discovered during the lifting. To get a sense of how significant the syntactic restrictions based on the weak inverse are, consider as an example that synthesizing a join for benchmark \texttt{max top strip} without them would take 12.1 seconds. Moreover, using a straightforward syntax-guided synthesis scheme (instead of deductive style algorithm of Section \ref{sec:alg}), it took over 40 minutes to find the auxiliary for {\tt mbbs} which is arguably the simplest instance that requires lifting.

\paragraph{\bf Quality of the Synthesized Code.}

\begin{figure}[t]
  \begin{center}
    \includegraphics[scale=0.25]{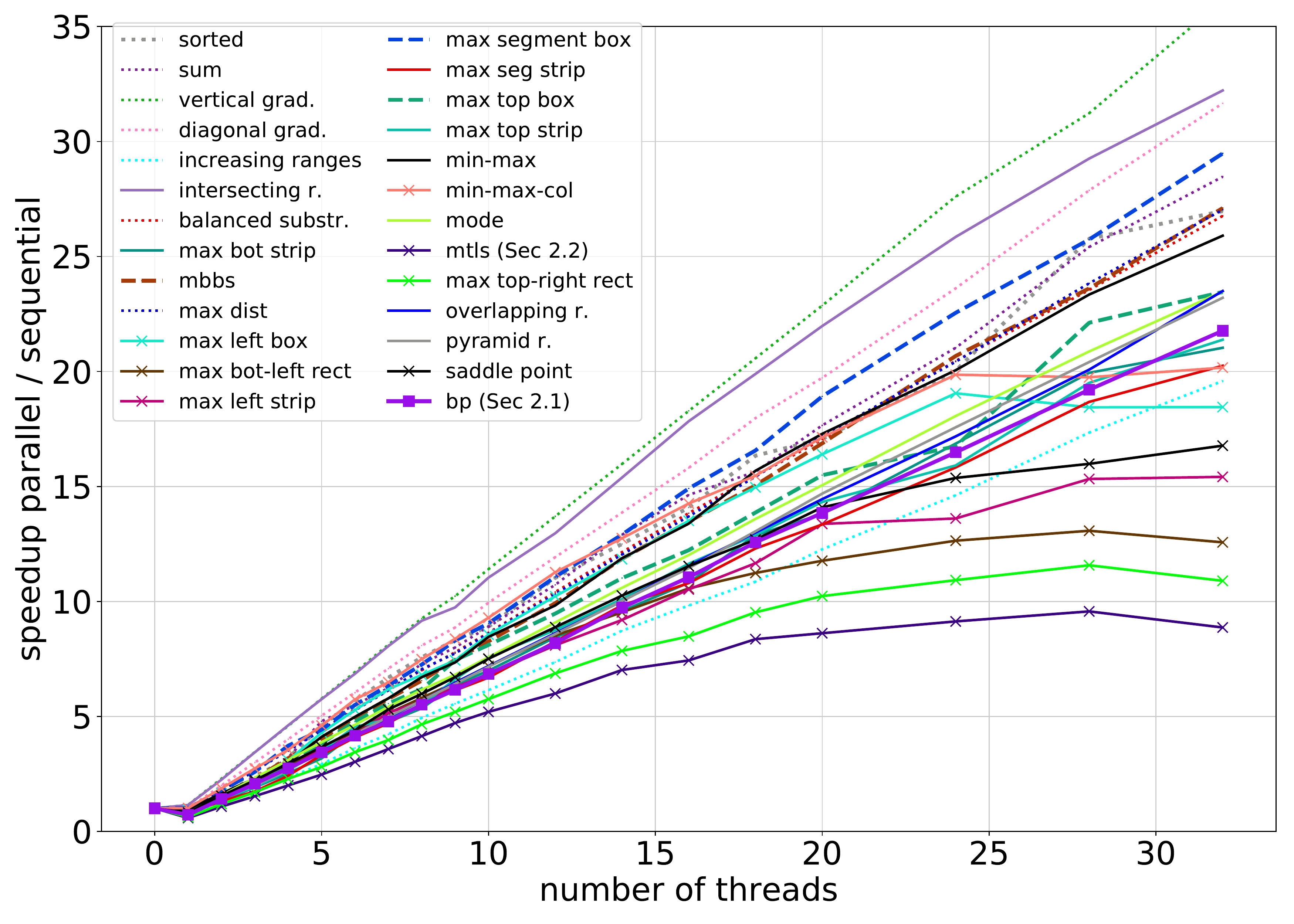}\vspace{-10pt}
    \caption{\small Speedups relative to the sequential implementation.
      {\em Hardware:} 8 eight-core Intel X6550 processors (64 cores total) and 256G of RAM running 64-bit Ubuntu
    }
    \label{fig:plot}\vspace{-10pt}
  \end{center}
\end{figure}

A divide-and-conquer parallelization, in the style of this paper, is data parallel program with no inter-thread dependencies, and therefore, reasonable parallelization speedups are expected. We implemented our produced parallel solutions using Intel's Thread Building Blocks (TBB) \cite{TBB} as well as OpenMP, to measure the speedups over the sequential variations. TBB turned out to produce better performing parallel programs. Speedups for four instances (one from each category) at 16 threads are compared below.
\begin{center}
\small
\begin{tabular}{|l|c|c|c|c|}\hline
 & {\tt max bot strip} & mbbs & mode & bp \\ \hline
OpenMP  & 11.0  & 8.6 & 11.0 & 7.8 \\ \hline
  TBB  & 12.7 & 10.7 & 11.5 & 8.9 \\ \hline
\end{tabular}
\end{center}
Figure \ref{fig:plot} illustrates the TBB speedups for up to 32 cores. In the  experiments, the size of the input arrays is about 2bn elements and the grain size is set at 50k. We used \texttt{gcc 7.3.0} to compile the parallelized benchmarks. On average, the speedups are close to linear on the number of cores up to around 32 cores.  Speedup is measured by dividing the running time of the parallel implementation with the running time of the sequential implementation {\em without auxiliaries}. Therefore, the speedup reported for 0 is always 1, and the speedup for 1 core is the parallel implementation {\em with auxiliaries} running on a single thread (which explains the inflexion of the curve at the beginning of the graph).
The examples with  smaller speedups over larger number of cores are those that have a more complex parallel join operator; in particular, those with looped joins. It is also known \cite{ContrerasM08} that the overhead of TBB increases with more cores and becomes very significant at 32 cores.

\paragraph{\bf Correctness.} {\sc Rosette} performs bounded verification for the solutions it generates. To have correctness over all inputs, we use the same scheme as \cite{PLDI17} to produce correctness proofs verified through Dafny \cite{dafny}. The majority of the programs were verified using the same proof generation scheme as \cite{PLDI17}. However, the bold benchmarks in the table required additional {\em simple and generic} lemmas that lift standard algebraic identities over integers to those over sequences of integers; for example, $\vec{x} + \mathit{max}(\vec{y}, \vec{z}) = \mathit{max}(\vec{x} + \vec{y}, \vec{x} + \vec{z})$.

%%% Local Variables:
%%% mode: latex
%%% TeX-master: "paper"
%%% End:

%% file: table.tex
\begin{table*}[t]
\begin{center}
\bgroup
\setlength\tabcolsep{1.5pt}
\small
\begin{tabular}{l|c|c|c|c|c|c|c|c|c|c|c|c|c|c|c|c|c|c|c|c|c|c|c|c|c|c|c|c|c}
 \cline{2-28}
  & \multicolumn{18}{|c|}{{\tt \scriptsize 2D} loop/input}
  & \multicolumn{4}{|c|}{{\tt \scriptsize 3D} loop/input}
  & \multicolumn{5}{|c|}{{\tt \scriptsize 2D} loop, {\tt \scriptsize 1D} inputs}\\
  \cline{2-28}
  % 2d loop / input
  & \rotatebox{90}{\tt \scriptsize sum}
  & \rotatebox{90}{\tt \scriptsize sorted}
  & \rotatebox{90}{\tt \scriptsize vertical gradient}
  & \rotatebox{90}{\tt \scriptsize diagonal gradient}
  & \rotatebox{90}{\tt \scriptsize min-max}
  & \rotatebox{90}{\tt \scriptsize min-max col.}
  & \rotatebox{90}{\tt \scriptsize saddle point}
  & \rotatebox{90}{\tt \scriptsize max top strip}
  & \rotatebox{90}{\tt \scriptsize max bottom strip}
  & \rotatebox{90}{\tt \scriptsize max segment strip}
  & \rotatebox{90}{\tt \scriptsize max left strip}
  & \rotatebox{90}{ \tt \bfseries\scriptsize mtls (Sec. \ref{sec:mtlr})}
  & \rotatebox{90}{\tt \bfseries \scriptsize max bot-left rect.}
  & \rotatebox{90}{\tt \bfseries \scriptsize max top-right rect.}
  & \rotatebox{90}{\tt \scriptsize bp  (Sec. \ref{sec:wbp})}
  & \rotatebox{90}{\tt \scriptsize increasing ranges}
  & \rotatebox{90}{\tt \scriptsize pyramid ranges}
  & \rotatebox{90}{\tt \scriptsize overlapping ranges}
  %  3D loop / input
  & \rotatebox{90}{\tt \scriptsize max top box}
  & \rotatebox{90}{\tt \scriptsize mbbs (Sec. \ref{sec:introduction})}
  & \rotatebox{90}{\tt \scriptsize max segment box }
  & \rotatebox{90}{\tt \scriptsize max left box }
  %  2d loop / 1D input
  & \rotatebox{90}{\tt \scriptsize mode}
  & \rotatebox{90}{\tt \scriptsize max-dist}
  & \rotatebox{90}{\tt \scriptsize balanced substr.}
  & \rotatebox{90}{\tt \scriptsize inter. ranges}
  & \rotatebox{90}{\tt \scriptsize LCS}
  \\
  \hline
  \multicolumn{1}{|l|}{Summarization time}
  & 1.2 % sum
  & 1.3 % sorted
  & 1.1 % vertical gradient
  & 1.2 % diagonal gradient
  & 1.2 % min-max
  & 1.5 % min-max col.
  & 4.6 % saddle point
  & 1.2 % max top strip
  & 1.2 % max bottom strip
  & 1.2 % max segment strip
  & 1.6 % max left strip
  & 1.4 % mtlr (Sec. \ref{sec:mtlr})
  & 30.2 % max bot-left rect.
  & 1.4 % max top-right rect.
  & 5.3 % wb  (Sec. \ref{sec:wbp})
  & 6.2% increasing ranges
  & 2.5 % pyramid ranges
  & 1.3 % overlapping ranges
  %  3D loop / input
  & 1.3 % max top box
  & 1.3 % mbbs (Sec. \ref{sec:introduction})
  & 1.3 % max segment box
  & 2.1  % max left box
  %  2d loop / 1D input
  & 2.4 % mode
  & 1.4 % max-dist
  & 54.9 % balanced substr.
  & 1.3  % intersecting ranges
  & 2.3 % LCS
  \\ \hline
 \multicolumn{1}{|l|}{\# Aux required}
  & - % sum
  & 1 % sorted
  & - % vertical gradient
  & - % diagonal gradient
  & - % min-max
  & - % min-max col.
  & - % saddle point
  & - % max top strip
  & 1 % max bottom strip
  & 2 % max segment strip
  & - % max left strip
  & 1 % mtlr (Sec. \ref{sec:mtlr})
  & 1 % max bot-left rect.
  & 1 % max top-right rect.
  & 1*% wb  (Sec. \ref{sec:wbp})
  & 1*% increasing ranges
  & 2 % pyramid ranges
  & 2 % overlapping ranges
  %  3D loop / input
  & - % max top box
  & 1 % mbbs (Sec. \ref{sec:introduction})
  & 2 % max segment box
  & - % max left box
  %  2d loop / 1D input
  & - % mode
  & - % max-dist
  & - % balanced substr.
  & - % intersecting ranges
  & \ding{55} % LCS
  \\
  \hline
  \multicolumn{1}{|l|}{Join synthesis time}
  & 2.3 % sum
  & 1.1 % sorted
  & 1.1 % vertical gradient
  & 1.1 % diagonal gradient
  & 2.5 % min-max
  & 2.3 % min-max col.
  & 5.4 % saddle point
  & 6.1 % max top strip
  & 11.8 % max bottom strip
  & 64.1 % max segment strip
  & 11.2 % max left strip
  & 116.3 % mtlr (Sec. \ref{sec:mtlr})
  & 216.2 % max bot-left rect.
  & 313.5 % max top-right rect.
  & 8.1 \textdagger% wb  (Sec. \ref{sec:wbp})
  & 10.5 % increasing ranges
  & 2.6 % pyramid ranges
  & 3.3 % overlapping ranges
  %  3D loop / input
  & 2.5 % max top box
  & 7.1 % mbbs (Sec. \ref{sec:introduction})
  & 52.3 % max segment box
  & 11.2  % max left box
  %  2d loop / 1D input
  & 22.7 % mode
  & 4.0 % max-dist
  & 11.5 % balanced substr.
  & 1.5 % intersecting ranges
  & \ding{55} % LCS
  \\ \hline
\end{tabular}
\egroup
\end{center}
\caption{\small
  Experimental results for performance of \tool. Times are in seconds.  ``--'' indicates that lifting was not required. Hardware: desktop with 8G RAM and Intel dual core m3-6Y30. The starred numbers of the middle row correspond to auxiliaries for a {\em memoryless lift},  and otherwise for a {\em homomorphism lift}. \textdagger: time reported is spent by the solver to answer unsat, since the summarized loop is not parallelizable.
\label{tab:benchmarks}} \vspace{-15pt}
\end{table*}

%%% Local Variables:
%%% mode: latex
%%% TeX-master: "paper"
%%% End:

%% file: conclusion.tex
% !TEX root =  paper.tex
\section{Discussion and Future Work} \label{sec:discussion}

\paragraph{\bf Input Programs.} Theoretically, our approach admits any program that can be semantically translated to a nested system of recurrence equations. In this model, each loop is a system of recurrence equations nested in the corresponding system for its surrounding loop. This is in contrast to other widely used models like System of Affine Recurrence Equations (SARE) \cite{srikant2002compiler}, designed to track dependencies, that represent the loop nest as a flat system of equations associated with an iteration domain. The only strict limitation for our input programs is that the input should not be modified by the program. More details on our input model can be found in \cite{extended}.

\paragraph{\bf Limitations.} Not all loop nests admit an efficient divide-and-conquer parallel solution with a syntactic divide operator that is the inverse of concatenation; for example, {\em quick sort} is a divide-and-conquer solution with a non-trivial {\em divide} operation whereas {\em merge sort} is a divide-and-conquer solution with a divide that is inverse of concatenation. Synthesis of non-trivial divide operators is an interesting topic of research for future work.

The dynamic programming instances considered in Bellmania \cite{Itzhaky0SYLLC16} also do not admit an efficient parallelization according to Definition \ref{def:par}. Bellmania uses {\em tiling} (a different {\em divide}) and a necessary scheduling of dependent tiles to correctly parallelize dynamic programming code. The produced code, however, is not fully data parallel in the manner resulting from manufacturing homomorphisms. A homomorphism can be executed in parallel without the need for scheduling.

{\tt LCS} (longest common substring), a dynamic programming algorithm, can be rewritten to admit an efficient parallelization (according to Definition \ref{def:par}). We used \tool on the modified code, which failed to parallelize it.  This is due to the fact that the auxiliary accumulators required for its lifting are {\em conditional}, and therefore fall beyond the reach of the heuristics of recursion discovery currently implemented in \tool. To sum up, within the class of programs that do admit efficient parallelizations (as defined in Definition \ref{def:par}), the limitations of \tool are mainly due to the heuristic natures of the implementations of the {\em normalization} and {\em recursion discovery} methods. For more complex computations, one can imagine that limitations of {\sc SyGuS} solvers can play a role in not discovery a join that theoretically exists.

\paragraph{\bf Predictability.} Due to the semantic nature of our approach, one cannot predict parallelizability of a loop nest based on any of its syntactic properties. Since parallelizability is equivalent to the computation being semantically a homomorphism (or liftable to one), the only way to know if a loop nest is parallelizable is to try to parallelize it. In fact, a negative is quite hard to prove as the proof of Theorem \ref{thm:hard} indicates.

\paragraph{\bf Future work.} The focus of this paper (and its predecessor \cite{PLDI17}) has been on synthesizing homomorphisms which have fixed simple divide operations. An interesting direction for future research would be solutions for  synthesizing non-trivial divide operations. Note that with the join and (potentially) the lifting being unknown, this adds one more unknown dimension to the problem which can make it substantially more difficult/interesting to solve.

%%% Local Variables:
%%% mode: latex
%%% TeX-master: paper
%%% End:

%% file: relwork.tex
% !TEX root =  paper.tex
\section{Related Work} \label{sec:relwork}
There is a vast body of literature on parallelizing code. We review only the closely related work to our approach here.

\paragraph{\bf Homomorphisms for Parallelization}
Our work is most closely related to those that exploit homomorphisms for parallelization \cite{Gorlatch96, Morita07}, and builds up on our recent work \cite{PLDI17}, where sequence (list) homomorphisms are automatically synthesized to parallelize simple (non-nested) loops. This paper is a highly non-trivial generalization of the work in \cite{PLDI17} to arbitrarily nested loops. Less recent attempts in using derivation of list homomorphisms for parallelization included methods based on the third homomorphism theorem  \cite{Geser97,Gorlatch96, Morita07}, function composition \cite{Fisher1994}, and quantifier elimination \cite{MorihataM10}, as well as those based on recurrence equations \cite{Ben01}. These techniques are either not fully automatic, or rely on additional guidance from the programmer beyond the input sequential code.

\paragraph{\bf Simple Loop Parallelization}
More recently in \cite{Raychev2015}, symbolic execution is used to identify and break dependencies in loops that are hard to parallelize. This approach can be regarded as a dynamic counterpart to that of \cite{PLDI17}, and its scope is similarly limited to simple loops. In distributed computing, a related vein of research has been focused on automatic production of map/reduce programs, for example, by means of specific rewrite rules \cite{Radoi2014TIC} or synthesis \cite{Smith2016MPS}. GraSSP \cite{Fedyukovich2016} parallelizes a sequential implementation by analyzing data dependencies and its scope is functions over lists. The (constant sized) prefix information used in \cite{Fedyukovich2016} is essentially a special case of the auxiliary accumulators in \cite{PLDI17}.

\paragraph{\bf Program Synthesis for Parallel Code Generation}
In this paper, the specification for synthesis is the sequential input program, and no other information (such as input/output examples or a sketch) is required from the programmer. Synthesis techniques have been leveraged for parallel programs before, instances of which include synthesis of distributed map/reduce programs from input/output examples \cite{Smith2016MPS} and optimization and parallelization of stencils \cite{Kamil2016}. Aside from the use of synthesis, these problem areas and the solutions have little in common with the scope and approach in this paper. Bellmania \cite{Itzhaky0SYLLC16} synthesizes divide-and-conquer variations of a class of dynamic programming algorithms {\em with programmer's guidance} and the notion of divide-and-conquer (with a fixed divide) in this paper differs from the one that Bellmania uses.

\paragraph{\bf Parallelizing Compilers and Runtime Environments}
Automatic parallelization in compilers has been a prolific and highly effective field of research, with source-to-source compilers using highly sophisticated methods to parallelize generic code \cite{Blume1996PPP, Bacon1994, Padua1986,hwansoo2001comparison} or more specialized nested loops with polyhedral optimization \cite{Verdoolaege2013, Bastoul2004, Vasilache2006, Bondhugula2008}. There is a body of work specific to reductions and parallel-prefix computations \cite{blelloch1990prefix, hillis1986data, ladner1980parallel} that deals with dependencies that cannot be broken.  In contrast to correct source-to-source transformation achieved through provably correct program transformation rules,  the aim of this paper is to use search (in the style of synthesis), which facilitates the discovery of equivalent parallel implementations that are not reachable through a pre-established set of correct transformation rules. There is work in the literature on breaking static dependencies at runtime \cite{Pingali2011} based on the observation that actual runtime dependencies happen rarely in some sparse problems. The scope of applicability of our method is different and we consider these techniques to be complementary. In \cite{BlellochFGS12}, a static two-phase solution is proposed that resolves dependencies in the first phase, and can proceed to perform independent parallel tasks in the second. We view the approach in this paper as complimentary to these techniques.

%%% Local Variables:
%%% mode: latex
%%% TeX-master: "paper"
%%% End: